\documentclass[pra,amsmath,amssymb,showpacs, superscriptaddress,notitlepage,nofootinbib,onecolumn]{revtex4-2}

\usepackage{mdframed}

\usepackage{graphicx}
\usepackage[caption=false]{subfig}
\usepackage{siunitx,booktabs}
\usepackage{float}
\usepackage{soul}
\usepackage{dcolumn}
\usepackage{amsmath,amssymb,mathtools,physics}
\usepackage{amsthm}
\usepackage{mathtools}
\usepackage{dsfont}
\theoremstyle{plain}
\newtheorem*{theorem*}{Theorem}
\newtheorem{theorem}{Theorem}
\newtheorem{lemma}{Lemma}
\newtheorem*{lemma*}{Lemma}
\newtheorem*{corollary*}{Corollary}
\newtheorem{corollary}{Corollary}

\newtheorem*{remark*}{Remark}
\newtheorem{remark}{Remark}

\newtheorem*{conjecture*}{Conjecture}

\usepackage[utf8]{inputenc}

\usepackage[hidelinks]{hyperref}
\usepackage{footnote}
\usepackage{enumitem}
\hypersetup{
	colorlinks   = true, 
	urlcolor     = blue, 
	linkcolor    = blue, 
	citecolor   = blue 
}
\usepackage[capitalize]{cleveref}

\makeatletter
\AddToHook{cmd/appendix/before}{\def\cref@section@alias{appendix}\def\cref@subsection@alias{appendix}}
\makeatother

\usepackage [english]{babel}
\usepackage [autostyle, english = american]{csquotes} \MakeOuterQuote{"}
\usepackage{bm}
\usepackage{enumitem}

\usepackage[dvipsnames]{xcolor}

\newcommand{\nph}{n_\mathrm{ph}}
\newcommand{\eph}{e_\mathrm{ph}}
\newcommand{\ephone}{e_{\mathrm{ph},(1)}}

\newcommand{\nphone}{n_{\mathrm{ph},(1)}}
\newcommand{\nphoneU}{\nphone^\mathrm{U}}

\newcommand{\affvqcc}{Vigo Quantum Communication Center, University of Vigo, Vigo E-{36310}, Spain}
\newcommand{\affuvigo}{Escuela de Ingeniería de Telecomunicación, Department of Signal Theory and Communications, University of Vigo, Vigo E-36310, Spain}
\newcommand{\affatlantic}{atlanTTic Research Center, University of Vigo, Vigo E-36310, Spain}
\newcommand{\afftoyama}{Faculty of Engineering, University of Toyama, Gofuku 3190, Toyama 930-8555, Japan}

\begin{document}
	
	\setlength{\parskip}{4pt}
	\setlength{\parindent}{0pt}

 \author{Guillermo Currás-Lorenzo}
 \email{gcurras@vqcc.uvigo.es}
\author{Margarida Pereira}
 \author{Alessandro Marcomini}
\affiliation{\affvqcc} \affiliation{\affuvigo} \affiliation{\affatlantic} 
\author{Kiyoshi Tamaki}
\affiliation{\afftoyama} 
\author{Marcos Curty}
\affiliation{\affvqcc} \affiliation{\affuvigo} \affiliation{\affatlantic}
    
\title{Security of decoy-state quantum key distribution with correlated bit-and-basis encoders}

    \begin{abstract}
Practical quantum key distribution (QKD) modulators inevitably introduce correlations, causing the state emitted in a given round to depend on the setting choices made in previous rounds. These correlations break the round-by-round independence structure on which many widely used security proof techniques rely, leaving a significant gap between available theoretical guarantees and the reality of practical implementations. In this work, we develop a {fully analytical} finite-key security proof for decoy-state BB84 against general coherent attacks that rigorously incorporates correlations introduced by Alice's bit-and-basis encoder, while requiring only partial characterization of such correlations. {Crucially, it does so without requiring the round partitioning of the data used for phase-error estimation introduced in previous works, showing that this technique is not indispensable for dealing with correlated encoders, and resulting in improved finite-key performance. Moreover, we also analyze the case of ideal single-photon sources, and our results reveal a counterintuitive effect: in the presence of strong {encoding} correlations, decoy-state BB84 can outperform single-photon implementations.}

\end{abstract}

	\maketitle

    \section{Introduction}

Quantum key distribution (QKD) enables two distant parties, Alice and Bob, to establish shared secret keys whose security rests on fundamental principles of quantum mechanics rather than computational assumptions. Translating this theoretical promise into practice, however, requires security proofs that faithfully account for the inevitable imperfections of real devices. Over the past two decades, significant progress has been made in developing security proofs that incorporate a range of source and detector imperfections~\cite{gottesmanSecurityQuantum2004,fungSecurityProof2009,tamakiLosstolerantQuantum2014,pereiraQuantumKey2020,curras-lorenzoSecurityFramework2025,tupkaryPhaseError2025,sixtoQuantumKey2025,kaminRenyiSecurity2025,marwahProvingSecurity2024,naharImperfectDetectors2026}. Among these imperfections, encoding correlations---in which each emitted quantum state depends on the setting choices made in previous rounds, due to the finite bandwidth of practical optical modulators~\cite{grunenfelderPerformanceSecurity2020,agulleiroModelingCharacterization2025}---pose a particularly serious challenge for security analyses.

The difficulty stems from the fact that encoding correlations break the round-by-round independence structure that many widely used security proof techniques rely upon. For instance, the Postselection Technique~\cite{christandlPostselectionTechnique2009,naharPostselectionTechnique2024} and Quantum de Finetti approaches~\cite{rennerSymmetryLarge2007} require the global protocol state to be permutation-invariant, a symmetry that is destroyed when the physical state of each round depends on the temporal ordering of previous settings. Similarly, the Marginal-constrained Entropy Accumulation Theorem (MEAT)~\cite{arqandMarginalconstrainedEntropy2025,kaminRenyiSecurity2025} models the protocol as a sequence of channels acting on states satisfying a fixed marginal constraint; with correlations, the source state in round $k$ depends on all previous settings, violating the required factorization structure. These incompatibilities have left a significant gap between the theoretical security guarantees available and the reality of practical QKD implementations, which inevitably suffer from such correlations.

On the other hand, security proofs based on phase-error estimation---encompassing both the entropic uncertainty relation (EUR) with leftover hashing lemma (LHL) approach~\cite{tomamichelUncertaintyRelation2011,tomamichelTightFinitekey2012,tomamichelLargelySelfcontained2017} and proofs based on phase-error correction (PEC)~\cite{koashiSimpleSecurity2009}---face no fundamental barriers to incorporating encoding correlations, as rigorously established in Ref.~\cite{curras-lorenzoRigorousPhaseerrorestimation2026} (see also Ref.~\cite{pereiraQuantumKey2020}). That work introduced a modular framework for extending phase-error-estimation-based security proofs to correlated sources: the key idea is to partition the protocol rounds into $(l_c+1)$ groups (where $l_c$ is the maximum correlation length), apply the existing uncorrelated phase-error estimation procedure independently to each group, and then combine the resulting bounds into a single upper bound on the overall phase-error rate of the full sifted key. Crucially, this approach enables security to be established through a single privacy amplification step on the entire key, avoiding the composability concerns associated with earlier proposals~\cite{pereiraQuantumKey2020,mizutaniSecurityRoundrobin2021,pereiraModifiedBB842023}. 

Despite these advances, two important limitations remain. First, Ref.~\cite{curras-lorenzoRigorousPhaseerrorestimation2026} currently considers only protocols in which Alice's emitted states are pure given a fixed sequence of setting choices. Thus, it does not directly consider decoy-state protocols~\cite{hwangQuantumKey2003,loDecoyState2005,wangBeatingPhotonnumbersplitting2005}, which are employed by the vast majority of practical QKD implementations\footnote{{After the first version of this manuscript was posted on arXiv, the approach of Ref.~\cite{curras-lorenzoRigorousPhaseerrorestimation2026} was extended to the decoy-state setting in Ref.~\cite{sixtoFinitekeySecurity2026}, incorporating further source and detector imperfections. That analysis, however, inherits the round partitioning and its associated statistical-efficiency loss, i.e., the second limitation discussed below. See also the Note added at the end of this paper.}}. In a decoy-state protocol, Alice emits phase-randomized {weak} coherent states, which are diagonal in the Fock basis and therefore mixed---even when conditioned on all of Alice's setting choices (bit, basis, and intensity).  Second, the approach of Ref.~\cite{curras-lorenzoRigorousPhaseerrorestimation2026} inherently requires dividing the protocol rounds into $(l_c+1)$ groups and estimating the phase-error rate separately within each group. Since only a fraction of the total data is used for parameter estimation in each group, this partitioning inevitably incurs a loss in statistical efficiency. One may therefore wonder whether it is possible to prove security without resorting to this division into groups.

In this work, we develop a new {and fully analytical} approach to handle encoding correlations in QKD that addresses both of these limitations. Rather than partitioning rounds into groups, our approach employs the quantum coin argument~\cite{gottesmanSecurityQuantum2004,loSecurityQuantum2007,curras-lorenzoSecurityFramework2025} to bound the phase-error rate of the full sifted key directly, using all available protocol data simultaneously for parameter estimation. A central component of our analysis is a new technique to bound the so-called bad coin events in the presence of encoding correlations. This technique exploits the fact that the relevant coin statistics depend only on Alice's local source-replacement state and are independent of Eve's attack, allowing us to partition the analysis of these events into groups of independent random variables and apply concentration inequalities, while the main phase-error estimation step uses the full, undivided data. 

Our approach is natively compatible with decoy-state protocols, and we use it to provide a complete finite-key security proof against general attacks for decoy-state BB84 in the presence of correlations introduced by the bit-and-basis encoder. {As a concrete example, we} apply our results to a 
linear time-invariant (LTI) correlations model~\cite{agulleiroModelingCharacterization2025} with exponentially decaying correlation strength, and we handle correlations of unbounded length by truncating at a finite correlation length and rigorously bounding the resulting approximation error via trace distance arguments. To our knowledge, this constitutes the first finite-key security proof of decoy-state QKD against general attacks that incorporates encoding correlations in the transmitter {(see also the Note at the end of this paper)}.

{The remainder of this paper is organized as follows. In \cref{sec:assumptions} we state our assumptions on Alice's source and Bob's receiver and define the protocol. \Cref{sec:security} presents our security analysis {and its application to a partially characterized LTI correlations model}, with all technical proofs deferred to \cref{app:proofs}. In \cref{sec:numerics} we numerically evaluate the resulting secret-key rates. Although the main text focuses on the decoy-state setting, our techniques apply equally to protocols employing ideal single-photon sources, and we provide the analysis for this case in \cref{app:single_photon}. Comparing the two implementations reveals a counterintuitive effect: when the correlations are strong, the decoy-state protocol can outperform the single-photon one, indicating that decoy-state implementations are inherently more robust to correlated bit-and-basis encoders. Finally, in \cref{sec:conclusion}, we conclude the paper by summarizing our findings.}

    \section{Device assumptions and protocol description} \label{sec:assumptions}


     \subsection{Assumptions on Alice's source}

    In each round $k$, Alice independently selects a bit $a_k \in \{0,1\}$, a basis $\alpha_k \in \{Z,X\}$ and an intensity $\mu_k \in \{\mu_1, \ldots, \mu_N\}$. The bit and basis are chosen uniformly at random, while the intensity is drawn according to a probability distribution $\{p_{\mu_k}\}$. Based on these choices, she prepares the state
    \begin{equation}
    \label{eq:rho_muk_a1k_alpha1k}
        \rho_{\mu_k, a_1^k, \alpha_1^k} = \sum_{m_k = 0}^{\infty} p_{m_k \vert \mu_k} \ketbra*{\psi_{a_1^k,\alpha_1^k}^{(m_k)}}_{T_k},
    \end{equation}
    where $m_k$ is the photon number of the pulse, {$p_{m_k|\mu_k}$ is the probability of a particular photon number given the intensity $\mu_k$, $T_k$ is the photonic system sent through the channel on round $k$ and where} we have introduced the shorthand notation $a_1^k = a_1 \ldots a_k$ and $\alpha_1^k = \alpha_1 \ldots \alpha_k$ for the histories of bit and basis choices up to round $k$. Note that, due to the correlations introduced by the bit-and-basis encoder, the eigenstates $\ket*{\psi_{a_1^k,\alpha_1^k}^{(m_k)}}$ depend not only on the bit and basis choices $a_k$ and $\alpha_k$ for the $k$-th round, but also on the histories $a_1^{k-1}$ and $\alpha_1^{k-1}$ of previous choices. We emphasize that we consider a decoy-state implementation with ideal intensity modulation and ideal {global} phase randomization; that is, only the bit and basis encoders are assumed to suffer from correlations.

{Equation~\eqref{eq:rho_muk_a1k_alpha1k} is the only assumption we make on the form of Alice's emitted states: we do not assume any specific functional dependence of the eigenstates $\ket*{\psi_{a_1^k,\alpha_1^k}^{(m_k)}}$ on the setting history. The reason such generality is possible is that, as will become clear in \cref{sec:security}, the encoding correlations affect our analysis through a single quantity that depends only on Alice's local source-replacement state and remains well defined regardless of its specific form.} 

{To apply our proof, one requires a characterization of the source sufficient to bound that single quantity. If the eigenstates are fully characterized, this can in principle be computed or bounded directly. However, since a perfect characterization of the correlations is unrealistic, our proof is also designed to be compatible with partial state characterization. {To illustrate this, in \cref{subsec:particular_correlations_model} we apply our analysis to a particular partially characterized correlations model, which we later use to evaluate the performance of our security proof (see \cref{sec:numerics}).} We remark that this model is introduced only as an example, and that our proof is equally compatible with other partial-characterization models; for instance, those that specify only a fidelity bound between the actual emitted states and some reference states (see, e.g., \cite[Supp.~Note C]{curras-lorenzoSecurityFramework2025}).}

\subsection{Assumptions on Bob's receiver}
\label{sec:assumptions_Bob}
We assume that, in each round, Bob randomly chooses between two POVMs, $\vec \Gamma_Z \coloneqq \{\Gamma_0^{(Z)}, \Gamma_1^{(Z)}, \Gamma_\bot^{(Z)}\}$ and $\vec \Gamma_X \coloneqq \{\Gamma_0^{(X)}, \Gamma_1^{(X)}, \Gamma_\bot^{(X)}\}$, and measures the incoming signal in the chosen POVM. Here, the outcomes $\{0,1\}$ correspond to bit values and $\bot$ indicates no detection. Moreover, since our goal is developing tools that incorporate encoding correlations, we consider the standard assumption made in phase-error-estimation security proofs of BB84-type protocols, namely, basis-independent detection efficiency. This means that the operator associated to an undetected round is the same for both POVMs, i.e., $\Gamma_\bot^{(Z)} = \Gamma_\bot^{(X)} \eqqcolon \Gamma_\bot$. {In a standard receiver where Bob actively selects the basis and measures with two threshold single-photon detectors, this assumption holds provided that both detectors have identical detection efficiencies and dark count rates, and that double-click events are assigned a uniformly random bit value rather than discarded.} We remark, however, that our proof can be directly extended to incorporate partially characterized detection efficiency mismatches in Bob's setup by simply applying the modular lifting theorem in Ref.~\cite{curras-lorenzoSecurityQuantum2026} (see also \cite{tupkaryPhaseError2025}). Our proof could also be extended to setups in which Bob makes a passive basis choice (using a beam splitter) by applying the results of Ref.~\cite[Sec.~VIII.A]{wangPhaseError2025}.

\subsection{Protocol description}

Consistently with all the above, we consider the following protocol definition:

\vspace{2pt}

\begin{mdframed}

\noindent\textbf{Actual Protocol} 

\begin{enumerate}
    \item \textit{State preparation:} For each round $k \in \{1,2,\ldots,N\}$, Alice chooses $a_k \in \{0,1\}$, $\alpha_k \in \{Z,X\}$ and $\mu_k$, prepares the state $\rho_{\mu_k, a_1^k, \alpha_1^k}$ in \cref{eq:rho_muk_a1k_alpha1k}, and sends it {to Bob} through the quantum channel.
    
    \item \textit{Eve's attack:} Eve performs the most general attack allowed by quantum mechanics on the transmitted systems $T_1^N$, and re-sends some output systems $B_1^N$ to Bob, while potentially keeping an ancillary system $E$.
    
    \item \textit{Detection announcement:} For each round $k$, Bob chooses $\beta_k \in \{Z,X\}$ with equal probability, performs the POVM $\vec \Gamma_{\beta_k}$, and records his outcome $b_k \in \{0,1,\bot\}$. Then, he publicly announces the set of detected rounds $\mathcal{D} \subseteq \{1, \ldots, N\}$ {i.e., those rounds in which his result is either 0 or 1}.
    
    \item \label{step:keep-trash} \textit{Keep/trash announcement:} For each detected round $k \in \mathcal{D}$, Alice announces "keep" with probability $p_{\mathrm{keep}}$ or "trash" otherwise. Let $\mathcal{D}_{\mathrm{keep}} \subseteq \mathcal{D}$ denote the set of keep rounds and $\mathcal{D}_{\mathrm{trash}} \subseteq \mathcal{D}$ the set of trash rounds.
    
    \item \textit{Basis and intensity announcement:} For each keep round $k \in \mathcal{D}_{\mathrm{keep}}$, Alice announces both $\alpha_k$ and $\mu_k$.
     
    \item \textit{Sifting:} For each $k \in \mathcal{D}_{\mathrm{keep}}$, Bob checks whether $\alpha_k = \beta_k$ and announces "sifted" if the condition is met or "unsifted" otherwise. For each $k \in \mathcal{D}_{\mathrm{trash}}$, Bob simply announces "sifted" or "unsifted" at random. Alice and Bob determine $\mathcal{D}_{\mathrm{keep},\mathrm{sifted},Z}$ and $\mathcal{D}_{\mathrm{keep},\mathrm{sifted},X}$, the sets of sifted keep rounds in which $\alpha_k = \beta_k = Z$ and $\alpha_k = \beta_k = X$, respectively.

    \item \textit{Test-round announcements:} For each round in $\mathcal{D}_{\mathrm{keep},\mathrm{sifted},X}$, Alice and Bob announce their bit outcomes $a_k$ and $b_k$.

    \item \textit{Sifted-key generation:} Alice and Bob define their sifted keys according to their bit outcomes $a_k$ and $b_k$ of the rounds in $\mathcal{D}_{\mathrm{keep},\mathrm{sifted},Z}$.

    \item  \label{step:postprocessing} \textit{Variable-length decision:} Let $\vec n$ denote all the data announced by Alice and Bob until this point. Using this data, Alice and Bob compute $\lambda_\mathrm{EC} (\vec n)$ (the number of bits to be revealed in one-way error correction) and $l(\vec n)$ (the length of the final key to be produced). Aborting corresponds to $l(\vec n) =0$.

    \item \textit{Error correction and error verification:} Alice and Bob implement a one-way error correction protocol that reveals $\lambda_\mathrm{EC} (\vec n)$ bits of information. They implement error verification by using a common two-universal hash function of output length $\log_2(2/\varepsilon_{\rm EV})$, having one of the parties announce the result and comparing their values. {Here, $\varepsilon_\mathrm{EV}$ upper bounds the probability that error verification passes while Alice's and Bob's keys differ, and thus quantifies the correctness of the protocol.} 

    \item \textit{Privacy amplification:} If error verification passes, Alice and Bob select a random hash function from a two-universal family and apply it to their sifted key to obtain a final key of length $l(\vec n)$.
\end{enumerate}

\end{mdframed}

\vspace{5pt}

Note that, unlike in the standard BB84 protocol, our protocol includes a keep/trash announcement step (step~\ref*{step:keep-trash}), in which Alice randomly designates each detected round as either ``keep'' or ``trash''. For trash rounds, Bob still announces whether the round is ``sifted'' or ``unsifted'', but he does so at random, independently of his actual basis choice, and the data from these rounds is never announced or used for key generation or parameter estimation. The trash rounds are introduced to enable the quantum coin argument, as will become clear in the next section.

\section{Security proof} \label{sec:security}

\subsection{Source replacement, phase-error estimation protocol and security statement}

To analyze the security of the above protocol, we employ the source-replacement technique together with a decomposition of Bob's measurement. We describe these two modifications in turn and then present the resulting equivalent protocol.

\paragraph*{Source replacement.} In the actual protocol, Alice classically selects $(a_k, \alpha_k, \mu_k)$ in each round $k$ and prepares the corresponding state $\rho_{\mu_k, a_1^k, \alpha_1^k}$ in \cref{eq:rho_muk_a1k_alpha1k}. Instead, we may equivalently imagine that Alice prepares, at the very beginning of the protocol, a global entangled state that purifies all of her classical and quantum registers simultaneously. Concretely, for each round $k$, we introduce a coin qubit register $C_k$ {(with orthonormal basis $\{\ket{Z}, \ket{X}\}$)} encoding the basis choice, a qubit register $A_k$ encoding the bit in the corresponding basis, an intensity register $I_k$ (with orthonormal basis $\{\ket{\hat{\mu}_1}, \ket{\hat{\mu}_2}, \ldots\}$), and a photon-number register $M_k$ (with orthonormal basis $\{\ket{0}, \ket{1}, \ket{2}, \ldots\}$). The global source-replacement state is then
\begin{equation}
    \label{eq:global_state_decoy_correlations}
     \ket{\Psi_N}_{C_1^N A_1^N I_1^N M_1^N T_1^N} = \sum_{\alpha_1^N \in \{Z,X\}^N} \sum_{a_1^N \in \{0,1\}^N} \bigotimes_{k=1}^{N} \frac{1}{2} \ket{\alpha_k}_{C_k}  \ket{(a_k)_{\alpha_k}}_{A_k} \ket*{\Psi_{a_1^k,\alpha_1^k}}_{I_k M_k T_k},
\end{equation}
where $T_k$ is the photonic system transmitted through the channel and
{\begin{equation}
\label{eq:round_state_decoy}
\ket*{\Psi_{a_1^k,\alpha_1^k}}_{I_k M_k T_k} = \sum_{\mu_k} \sqrt{p_{\mu_k}} \ket{\hat{\mu}_k}_{I_k} \sum_{m_k=0}^{\infty} \sqrt{p_{m_k|\mu_k}} \ket{m_k}_{M_k} \ket*{\psi_{a_1^k,\alpha_1^k}^{(m_k)}}_{T_k}.
\end{equation}}{Here, $\ket{Z}_{C_k}$ and $\ket{X}_{C_k}$ simply denote the computational basis states $\ket{0}_{C_k}$ and $\ket{1}_{C_k}$ of the coin qubit, relabeled to indicate the basis choice they are associated to, and the Hadamard-basis states employed below are defined in the usual way, i.e., $\ket{\pm}_{C_k} = \big(\ket{Z}_{C_k} \pm \ket{X}_{C_k}\big)/\sqrt{2}$.} By construction, if Alice measures $C_k$ in the $\{\ket{Z},\ket{X}\}$ basis, $A_k$ in the basis determined by the outcome $\alpha_k$, $I_k$ in its defining basis, and $M_k$ in the photon-number basis, the post-measurement state on the transmitted systems $T_1^N$ is precisely the state that the actual protocol would have produced for the same sequence of outcomes. Since these registers {remain in Alice's laboratory and never} enter the quantum channel, {their measurements commute with all operations performed by Eve and Bob on the transmitted systems, and} can be  {thus} postponed to any later stage of the protocol without affecting the state of the transmitted systems, Bob's measurement outcomes, or any statistics derived from them.

\paragraph*{Decomposition of Bob's measurement.} In the actual protocol, for each round $k$, Bob selects a basis $\beta_k \in \{Z,X\}$ and applies the three-outcome POVM $\vec{\Gamma}_{\beta_k} = \{\Gamma_0^{(\beta_k)}, \Gamma_1^{(\beta_k)}, \Gamma_\bot^{(\beta_k)}\}$. By the assumption of basis-independent detection efficiency {(\cref{sec:assumptions_Bob})},
the detection operator
\begin{equation}
    F \coloneqq \Gamma_0^{(Z)} + \Gamma_1^{(Z)} = \Gamma_0^{(X)} + \Gamma_1^{(X)} = \mathbb{I} - \Gamma_\bot
\end{equation}
is the same regardless of the basis. This means that the detection event (i.e., whether the round belongs to $\mathcal{D}$) is statistically independent of Bob's basis choice $\beta_k$. We can therefore decompose Bob's action into two sequential steps: (1) he applies the basis-independent binary filter $\{F, \mathbb{I} - F\}$, which determines whether a detection occurs; (2) only for detected rounds, he chooses $\beta_k \in \{Z,X\}$ and applies a two-outcome POVM $\{G_0^{(\beta_k)}, G_1^{(\beta_k)}\}$, where $G_{b_k}^{(\beta_k)} = F^{-1/2} \Gamma_{b_k}^{(\beta_k)} F^{-1/2}$ (defined on the support of $F$), to determine his bit value $b_k$. This decomposition reproduces the original measurement statistics exactly, but it allows us to defer Bob's bit-value measurement to a later stage in the protocol, which is convenient for the security analysis.

Incorporating both modifications, we define the following scenario:

\vspace{3pt}

\begin{mdframed}

\noindent\textbf{Actual protocol (source replaced)} 

\begin{enumerate}
    \item \textit{State preparation:} Alice prepares her global source replacement state $\ket{\Psi_N}_{C_1^N A_1^N I_1^N M_1^N T_1^N}$ and sends the photonic systems $T_1^N$ {to Bob} through the quantum channel.

    \item \textit{Eve's attack:} Eve performs the most general attack allowed by quantum mechanics on the transmitted systems $T_1^N$, and re-sends some output systems $B_1^N$ to Bob, while potentially keeping an ancillary system $E$.

    \item \textit{Detection announcement:} For each round $k$, Bob applies a basis-independent filter $\{F, \mathbb{I} - F\}$ to system $B_k$, which determines whether he obtains a detection. Bob publicly announces the set of detected rounds $\mathcal{D} \subseteq \{1, \ldots, N\}$.

    \item \textit{Keep/trash and sifting decisions:} For each round $k \in \{1, \ldots, N\}$, Alice decides "keep" with probability $p_{\mathrm{keep}}$ or "trash" otherwise, and Bob independently decides "sifted" or "unsifted" with equal probability $1/2$. For detected rounds ($k \in \mathcal{D}$), Alice announces her keep/trash decision.  Let $\mathcal{D}_{\mathrm{keep}} \subseteq \mathcal{D}$ denote the set of detected keep rounds and $\mathcal{D}_{\mathrm{trash}} \subseteq \mathcal{D}$ the set of detected trash rounds. For the rounds in $\mathcal{D}_{\mathrm{trash}}$, Bob announces whether he chose "sifted" or "unsifted".

    \item \label{step:coin-measurements-and-basis-announcements} \textit{Coin measurements and basis/intensity announcement:} 
    \begin{enumerate}
        \item \label{step:coin-measurements} For each round $k$ where Alice chose "trash" and Bob chose "sifted", Alice measures her coin register $C_k$ in the $\{\ket{+}, \ket{-}\}$ basis.
        \item For each detected keep round $k \in \mathcal{D}_{\mathrm{keep}}$, Alice measures her coin register $C_k$ in the $\{\ket{Z}, \ket{X}\}$ basis to obtain $\alpha_k$, measures her intensity register $I_k$ in the $\{\ket*{\hat{\mu}_{1}}, \ket*{\hat{\mu}_{2}}, \ldots\}$ basis to obtain $\mu_k$, and publicly announces both $\alpha_k$ and $\mu_k$. Bob announces his measurement basis $\beta_k$: if he chose "sifted" for round $k$, then $\beta_k = \alpha_k$; otherwise $\beta_k \neq \alpha_k$.
    \end{enumerate}

    \item \label{step:photon-number-measurement}  \textit{Round classification and photon-number measurements:} Alice and Bob determine $\mathcal{D}_{\mathrm{keep},\mathrm{sifted},Z}$ and $\mathcal{D}_{\mathrm{keep},\mathrm{sifted},X}$, the sets of detected keep rounds where Bob chose "sifted" and $\alpha_k = \beta_k = Z$ and $\alpha_k = \beta_k = X$, respectively. Also, for each round $k \in \{1,\ldots,N\}$, Alice measures her photon number register $M_k$ in the $\{\ket{0}, \ket{1}, \ket{2}, \ldots\}$ basis to obtain $m_k$.

    \item \textit{Test-round measurements:} For each round in $\mathcal{D}_{\mathrm{keep},\mathrm{sifted},X}$, Alice measures her qubit register $A_k$ in the $\{\ket{0_X},\ket{1_X}\}$ basis, and Bob performs a two-outcome measurement $\{G_0^{(X)}, G_1^{(X)}\}$. Then, Alice and Bob announce their bit outcomes $a_k$ and $b_k$.

    \item \textit{Sifted-key measurements:} For each round in $\mathcal{D}_{\mathrm{keep},\mathrm{sifted},Z}$, Alice measures her qubit register $A_k$ in the $\{\ket{0_Z},\ket{1_Z}\}$ basis, and Bob performs a two-outcome measurement $\{G_0^{(Z)}, G_1^{(Z)}\}$. Alice and Bob define their sifted keys according to their bit outcomes $a_k$ and $b_k$.

    \item[9-11.] Same as \textit{Actual Protocol}.
    
\end{enumerate}

\end{mdframed}

\vspace{5pt}

Four additional features of the source-replaced protocol merit comment:

\begin{enumerate}
    \item The measurement of coin registers $C_k$ in the Hadamard basis $\{\ket{+},\ket{-}\}$ for the trash-sifted rounds (step~\ref*{step:coin-measurements}) does not affect any announced information, since the outcome of this measurement is never revealed. From the perspective of all other systems, any measurement on Alice's local register is equivalent to a partial trace, which is basis-independent. This fictitious Hadamard-basis measurement is introduced solely for the purpose of the phase-error estimation argument, as will be clear later.

    \item The measurement of the photon-number registers $M_k$ (step~\ref*{step:photon-number-measurement}) is similarly a local operation on Alice's side whose outcome is not announced, and can therefore be placed at any point in the protocol without affecting observable statistics; it is included at this stage for notational convenience.

    \item In the actual protocol, Alice's keep/trash decision is made only for detected rounds $k \in \mathcal{D}$. In the source-replaced protocol, however, Alice makes this decision for \emph{all} rounds $k \in \{1, \ldots, N\}$, but only announces it for detected rounds. Since the keep/trash assignment for undetected rounds is never revealed, this modification does not affect any announced information or observable statistics. 

    \item {Similarly, in the actual protocol, Bob selects a basis $\beta_k$ in every round and, for the rounds in $\mathcal{D}_{\mathrm{keep}}$, determines whether the round is "sifted" by comparing it with Alice's announced $\alpha_k$. In the source-replaced protocol, this is replaced by an independent, uniformly random "sifted"/"unsifted" decision made for \emph{all} rounds $k \in \{1, \ldots, N\}$, with Bob's basis for the rounds in $\mathcal{D}_{\mathrm{keep}}$ determined a posteriori as $\beta_k = \alpha_k$ if he chose "sifted" and $\beta_k \neq \alpha_k$ otherwise. Since, in the actual protocol, $\beta_k$ is chosen uniformly at random and independently of $\alpha_k$, the condition $\alpha_k = \beta_k$ holds with probability $1/2$ independently of all other variables, and both descriptions result in identical statistics. Note that postponing Bob's basis choice in this manner is possible because his detection filter $\{F, \mathbb{I}-F\}$ is basis independent. Moreover, as with Alice's keep/trash assignment, Bob's sifted/unsifted decision for undetected rounds is never revealed, and therefore does not affect any announced information or observable statistics.}
\end{enumerate}

Due to all the above, it is clear that the \textit{Actual protocol (source replaced)} is completely equivalent to the \textit{Actual protocol} in terms of the joint statistical distribution of the sifted key and the announced information {$\bm{\vec n}$, where $\bm{\vec n}$ denotes the random vector whose realization is the announced data $\vec n$ introduced in Step~\ref{step:postprocessing}}.\footnote{Note that, throughout the paper, we use boldface symbols to denote random variables and their corresponding plain symbols to denote their realizations.} Thus, proving security of the \textit{Actual protocol (source replaced)} directly implies security of the \textit{Actual protocol}.

To analyze the security of the \textit{Actual protocol (source replaced)}, we define its associated \textit{Phase-error estimation protocol}:

\vspace{3pt}

\begin{mdframed}

\noindent\textbf{Phase-error estimation protocol} 

\begin{enumerate}
    \item[1-7.] Same as \textit{Actual protocol (source replaced)}.

    \item[8.] \textit{Phase-error measurements:} For each round in $\mathcal{D}_{\mathrm{keep},\mathrm{sifted},Z}$, Alice measures her qubit register $A_k$ in the $\{\ket{0_X},\ket{1_X}\}$ basis, and Bob performs a two-outcome measurement $\{G_0^{(X)}, G_1^{(X)}\}$. For each photon number $m \in \{0,1,2,\ldots\}$, we denote by $\bm{n_{K,(m)}}$ the number of events in $\mathcal{D}_{\mathrm{keep},\mathrm{sifted},Z}$ with photon number $m_k = m$, and by $\bm{\eph}^{(m)}$ the fraction of events in $\mathcal{D}_{\mathrm{keep},\mathrm{sifted},Z}$ with photon number $m_k = m$ in which $a_k \neq b_k$.
\end{enumerate}

\end{mdframed}

\vspace{5pt}

Then, we have the following result:

\begin{theorem}[Variable-length security via EUR]
\label{thm:variable_length_security_decoy}

Suppose that, for any initial global source replacement state  $\ket{\Psi_N}_{C_1^N A_1^N I_1^N M_1^N T_1^N}$, we have the guarantee that, in the \emph{Phase-error estimation protocol}, a bound of the form
\begin{equation}
\label{eq:statistical_guarantee_decoy}
        \Pr[\bm{n_{K,(1)}}< \mathcal{N}_{K,(1)} (\bm{\vec{n}})\quad \cup \quad \bm\ephone > \mathcal{E}_{\mathrm{ph},(1)} (\bm{\vec{n}}) ] \leq \varepsilon_\mathrm{PE}, ~
    {\text{for some $\epsilon_{\rm PE}>0$},}
    \end{equation}
holds, where $\mathcal{N}_{K,(1)}(\vec n)$ and $\mathcal{E}_{\mathrm{ph},(1)} (\vec{n})$ are functions of the observed data vector $\vec n$. Let $\lambda_\mathrm{EC}(\vec n)$ be a function that determines the number of bits used in error correction, and let
\begin{equation}
\label{eq:final_key_length}
    l(\vec n) = \max\Big(0,\mathcal{N}_{K,(1)} (\vec{n})\big[1-h\big(\mathcal{E}_{\mathrm{ph},(1)} (\vec n)\big)\big] - \lambda_\mathrm{EC}(\vec n) - 2 \log_2 \frac{1}{2\varepsilon_\mathrm{PA}}-\log_2 \frac{2}{\varepsilon_\mathrm{EV}}\Big),
\end{equation}
be a function that determines the length of the final key, where $h(x)$ is the binary entropy function for $x\leq 1/2$ and $h(x) = 1$ otherwise. Then, if Alice and Bob run the \emph{Actual protocol (source replaced)} using this choice of $\lambda_\mathrm{EC}(\vec n)$ and $l(\vec n)$, the output key is $(2 \sqrt{\varepsilon_\mathrm{PE}} +\varepsilon_\mathrm{PA} + \varepsilon_\mathrm{EV})$-secure. 
\end{theorem}

\begin{proof}
See Appendix~\ref{appsub:proof_thm:variable_length_security_decoy}.
\end{proof}

\subsection{Obtaining a bound on the single-photon phase-error rate}

To prove security, all that we need is to derive a statistical guarantee of the form in \cref{eq:statistical_guarantee_decoy}. Obtaining a bound on $\bm{n_{K,(1)}}$ is straightforward, since we can apply the standard decoy-state method, as we will shortly see. The main focus of the analysis below is therefore how to obtain a bound on the single-photon phase-error rate $\bm\ephone$ in the presence of encoding correlations.

\subsubsection*{Quantum coin inequality}

First, we establish the following result.

\begin{lemma}[Single-photon quantum coin inequality for decoy-state BB84]
\label{lem:quantum_coin_inequality_decoy}
    Consider the \emph{Phase-error estimation protocol} defined above. Let:
    \begin{itemize}
        \item $\bm{n}_{\mathrm{keep},\mathrm{sifted},Z,(1)}^{\mathrm{err}}$ = number of (phase) errors in detected {keep-sifted} rounds with $\alpha_k = \beta_k = Z$ and $m_k = 1$ (= number of single-photon phase errors $\bm{\nphone}$);
        \item $\bm{n}_{\mathrm{keep},\mathrm{sifted},Z,(1)}^{\mathrm{det}}$ = total number of detected {keep-sifted} rounds with $\alpha_k = \beta_k = Z$ and $m_k = 1$;
        \item $\bm{n}_{\mathrm{keep},\mathrm{sifted},X,(1)}^{\mathrm{err}}$ = number of errors in detected {keep-sifted} rounds with $\alpha_k = \beta_k = X$ and $m_k = 1$;
        \item $\bm{n}_{\mathrm{keep},\mathrm{sifted},X,(1)}^{\mathrm{det}}$ = total number of detected {keep-sifted} rounds with $\alpha_k = \beta_k = X$ and $m_k = 1$;
        \item $\bm{n}_\mathrm{sifted}^{\mathrm{det}}$ = total number of detected sifted rounds (both keep and trash);
        \item $\bm{n}_{\mathrm{trash},\mathrm{sifted},-,(1)}$ = number of {trash-sifted} rounds where Alice obtained $\ket{-}$ and $m_k = 1$.
    \end{itemize}
    
    Then, for any $\epsilon_A > 0$, with $\Delta_A = \sqrt{2 \bm{n}_\mathrm{sifted}^{\mathrm{det}} \ln(1/\epsilon_A)}$ and $\bm{n}_{\mathrm{keep},\mathrm{sifted},(1)}^{\mathrm{det}} = \bm{n}_{\mathrm{keep},\mathrm{sifted},Z,(1)}^{\mathrm{det}} + \bm{n}_{\mathrm{keep},\mathrm{sifted},X,(1)}^{\mathrm{det}}$, 
    \begin{equation}
        \begin{aligned}
        & \bm{n}_{\mathrm{keep},\mathrm{sifted},Z,(1)}^{\mathrm{err}} \underset{5\epsilon_{A}}{\leq} (\bm{n}_{\mathrm{keep},\mathrm{sifted},Z,(1)}^{\mathrm{det}} + \Delta_A) \times \\ 
        & G_+\left(\frac{\bm{n}_{\mathrm{keep},\mathrm{sifted},X,(1)}^{\mathrm{err}} + \Delta_A}{\bm{n}_{\mathrm{keep},\mathrm{sifted},X,(1)}^{\mathrm{det}} - \Delta_A}, 1 - \frac{2 p_{\mathrm{keep}}(\bm{n}_{\mathrm{trash},\mathrm{sifted},-,(1)} + \Delta_A)}{(1-p_{\mathrm{keep}}) (\bm{n}_{\mathrm{keep},\mathrm{sifted},Z,(1)}^{\mathrm{det}}+\bm{n}_{\mathrm{keep},\mathrm{sifted},X,(1)}^{\mathrm{det}})}\right) + \Delta_A \eqqcolon \bm{\nphoneU},
    \label{eq:finite_azuma_decoy}
    \end{aligned}
    \end{equation}
    {where $G_+(y,z) = y + (1-z^2)(1-2y) + 2\sqrt{z^2(1-z^2)y(1-y)}$ for {$\sqrt{y}<z$} and $G_+(y,z) = 1$ otherwise.} Thus, the single-photon phase-error rate is bounded as
    \begin{equation}
        \Pr[\bm{\ephone} > \frac{\bm{\nphoneU}}{\bm{n}_{\mathrm{keep},\mathrm{sifted},Z,(1)}^{\mathrm{det}}} ] \leq 5 \epsilon_A.
    \end{equation}
\end{lemma}

\begin{proof}
        The proof crucially relies on the fact that, throughout the protocol, at any moment, the coin systems $C_1^N$ are always qubits, and thus the Bloch sphere inequality can be applied to them. A full self-contained proof is provided in Appendix~\ref{appsub:proof_quantum_coin_inequality_decoy}.
\end{proof}

{\begin{remark}[On the choice of concentration inequality]
\label{rem:azuma_vs_kato}
The deviation terms $\Delta_A$ in \cref{eq:finite_azuma_decoy} originate from Azuma's inequality (\cref{lem:azuma}), which we use for its simplicity and because it requires no prior knowledge of the expected statistics. The price is that $\Delta_A$ grows with the square root of the total number of detected sifted rounds, regardless of how small the individual counts are. This is pessimistic for quantities such as $\bm{n}_{\mathrm{keep},\mathrm{sifted},X,(1)}^{\mathrm{err}}$, whose expected value is only a small fraction of $\bm{n}_\mathrm{sifted}^\mathrm{det}$. A sharper alternative is Kato's inequality~\cite{katoConcentrationInequality2020,curras-lorenzoTightFinitekey2021}, whose deviation term depends on the observed sum and is therefore much tighter when the latter is small. This inequality is also directly compatible with our analysis and would improve the finite-key performance, at the cost of a lengthier presentation: its parameters must be optimized before the protocol run from a prediction of the expected statistics, and the resulting deviation terms differ for each quantity, rather than being common to all of them as above. {We also note that Ref.~\cite{mannalathQuantumKey2026} has recently introduced concentration bounds that improve over Kato's inequality by being more robust to inaccurate predictions of the channel behavior, and which are equally compatible with our analysis.} We emphasize {that these refinements} are independent of the treatment of encoding correlations introduced here, which is the focus of this work; the same observation applies to the single-photon analysis of \cref{app:single_photon}.
\end{remark}}

\subsubsection*{Bounding the bad coin events}

As we will shortly see, the terms $\bm{n}_{\mathrm{keep},\mathrm{sifted},Z,(1)}^{\mathrm{det}}$, $\bm{n}_{\mathrm{keep},\mathrm{sifted},X,(1)}^{\mathrm{err}}$, and $\bm{n}_{\mathrm{keep},\mathrm{sifted},X,(1)}^{\mathrm{det}}$ appearing in \cref{eq:finite_azuma_decoy} can all be related to the observed protocol statistics via the standard decoy-state method. However, the remaining term, $\bm{n}_{\mathrm{trash},\mathrm{sifted},-,(1)}$, counts {the} outcomes of a Hadamard-basis measurement on Alice's coin registers during trash rounds, a fictitious measurement that is never performed in the actual protocol. Bounding this term using only publicly announced quantities is therefore the main challenge.

The key observation is that $\bm{n}_{\mathrm{trash},\mathrm{sifted},-,(1)}$ depends exclusively on Alice's local operations. To see this, note that the Hadamard measurement acts on Alice's coin registers {$C_1^N$}, which never leave Alice's lab. Thus, their outcomes are independent of both Eve's attack and Bob's measurement results. It does depend on Bob's decision whether each round is "sifted" or "unsifted", but since this assignment is independent of all other variables, we can without loss of generality assume that it is Alice who makes this decision. Consequently, the distribution of $\bm{n}_{\mathrm{trash},\mathrm{sifted},-,(1)}$ is fully determined by Alice's source-replacement state {$\ket{\Psi_N}_{C_1^N A_1^N I_1^N M_1^N T_1^N}$ in \cref{eq:global_state_decoy_correlations}}, and we can analyze it in the following simplified scenario that eliminates all protocol elements involving Bob and Eve:

 \vspace{3pt}

    \begin{mdframed}
    \textbf{Simple Coin Measuring Scenario}
    
    \begin{enumerate}
        \item \textit{State preparation:}  Alice prepares her global source replacement state $\ket{\Psi_N}_{C_1^N A_1^N I_1^N M_1^N T_1^N}$.
        \item \textit{Trash/sifted assignment and coin measurement:}  For each round $k \in \{1,\ldots,N\}$, Alice measures the photon number system $M_k$ to learn the photon number $m_k$. If $m_k = 1$, Alice assigns the round to "trash" with probability $(1-p_\mathrm{keep})$, and to "sifted" with probability $1/2$. Then, if the round is assigned to both "trash" and "sifted", Alice measures the coin system $C_k$ in the $\{\ket{+}, \ket{-}\}$ basis. We denote by $\bm{n}_{\mathrm{trash},\mathrm{sifted},-,(1)}$ the number of {trash-sifted} rounds with $m_k = 1$ where Alice obtained $\ket{-}$. 
    \end{enumerate}
    \end{mdframed}

    \vspace{5pt}

Within this scenario, we define the indicator random variables
    \begin{equation}
    \label{eq:X_k_def}
        \bm{X_k} \coloneqq \mathbf{1}\bigl\{\text{Alice obtains }m_k=1\text{, chooses "trash" and "sifted" and obtains }\ket{-}_{C_k}\text{ on round }k\bigr\}
    \end{equation}
so that
\begin{equation}
\label{eq:ntrashsiftedminusone_sum}
\bm{n}_{\mathrm{trash},\mathrm{sifted},-,(1)} = \sum_{k=1}^{N} \bm{X_k}.
\end{equation}

The problem thus reduces to bounding a sum of binary random variables derived from measurements on a known (or partially known) quantum state, with no adversary involved. The difficulty of this task depends on the structure of {$\ket{\Psi_N}_{C_1^N A_1^N I_1^N M_1^N T_1^N}$}, and in particular on whether the source exhibits correlations between rounds.

\textbf{Uncorrelated case.} If the source is memoryless, the state {$\ket{\Psi_N}_{C_1^N A_1^N I_1^N M_1^N T_1^N}$} factorizes across rounds, and so do the coin registers $C_1^N$. The random variables $\bm{X_k}$ are then independent Bernoulli variables. One can upper-bound $\mathbb{E}[\bm{X_k}]$ for each $k$ using whatever partial characterization of the source is available, and then apply a standard Chernoff-type bound to obtain an upper bound on $\bm{n}_{\mathrm{trash},\mathrm{sifted},-,(1)}$, as done in Ref.~\cite{curras-lorenzoSecurityFramework2025}.

\textbf{Correlated case.} When the source has memory, the global state {$\ket{\Psi_N}_{C_1^N A_1^N I_1^N M_1^N T_1^N}$} is entangled across rounds (see \cref{eq:global_state_decoy_correlations}), and the random variables $\bm{X_k}$ are no longer independent. However, since Eve has been removed from the picture, the distribution of $\sum_k \bm{X_k}$ is entirely determined by Alice's source-replacement state. In principle, if {$\ket{\Psi_N}_{C_1^N A_1^N I_1^N M_1^N T_1^N}$} is fully characterized---even with arbitrary long-range correlations---one could directly compute or bound this distribution without any further assumptions, {although} this could be computationally demanding.

In practice, however, one typically has only partial knowledge of the source correlations, so a general method to obtain a statistical bound on \cref{eq:ntrashsiftedminusone_sum} is useful. To achieve this, we introduce at this stage the assumption of a finite correlation length $l_c$: the state of round $k$ depends on at most $l_c$ preceding rounds. This allows us to partition the $N$ rounds into $l_c + 1$ groups of mutually independent variables and apply a Chernoff-type bound within each group, allowing us to obtain a broadly applicable result. We emphasize that this is one possible approach; other techniques for bounding sums of dependent random variables could perhaps be employed here as well, depending on the available characterization of the source. Moreover, we remark that the assumption of a strictly finite correlation length can be easily relaxed by applying \cref{lem:unbounded_correlations}, as explained later.

\begin{lemma}
\label{lem:bernstein_by_groups_decoy}
    Suppose that Alice runs the \emph{Phase-error estimation protocol},
    and that the source has a maximum correlation length~$l_c$.
    For each round~$k$, denote by
    $a_{\bar{k}} \coloneqq (a_{k-l_c}^{k-1}, a_{k+1}^{k+l_c})$
    and
    $\alpha_{\bar{k}} \coloneqq (\alpha_{k-l_c}^{k-1}, \alpha_{k+1}^{k+l_c})$
    the bit and basis choices in the $l_c$-neighbourhood of round~$k$
    (excluding round~$k$ itself).
    Then, for any $\epsilon_C \in (0,1)$,
    \begin{equation}
    \label{eq:bound_n_trash_decoy}
        \bm{n}_{\mathrm{trash},\mathrm{sifted},-,(1)}
        \;\underset{(l_c+1)\epsilon_C}{\leq}\;
        n_{\mathrm{trash},\mathrm{sifted},-,(1)}^{\mathrm{U}},
    \end{equation}
    where
    \begin{equation}
    \label{eq:n_trash_U_decoy}
    \begin{aligned}
        n_{\mathrm{trash},\mathrm{sifted},-,(1)}^{\mathrm{U}}
        &\coloneqq{} \frac{N p_1 (1-p_{\mathrm{keep}})}{2}\, \Delta_{\mathrm{coin}}^{\mathrm{U}}
        \\
        &~+ \sqrt{N p_1 (l_c+1)(1-p_{\mathrm{keep}})\,
              \Delta_{\mathrm{coin}}^{\mathrm{U}}\,
              \ln \tfrac{1}{\epsilon_C}}
        \;+\; \tfrac{2(l_c+1)}{3}\,\ln \tfrac{1}{\epsilon_C}\,,
    \end{aligned}
    \end{equation}
    with $p_1 = \sum_{\mu} p_{\mu}\, p_{1|\mu}$.
    Here, $\Delta_{\mathrm{coin}}^{\mathrm{U}}$ quantifies the
    worst-case deviation of Alice's single-photon states from an
    ideal BB84 source, and is defined as
    \begin{equation}
    \label{eq:Delta_coin_U_decoy}
        \Delta_{\mathrm{coin}}^{\mathrm{U}}
        \coloneqq
        \max_{k}\; \max_{a_{\bar{k}},\,\alpha_{\bar{k}}}\;
        \frac{1}{2}\,\Re\!\Bigg(
            1 - \frac{1}{2\sqrt{2}}
            \sum_{a,a' \in \{0,1\}}
            (-1)^{a a'}\,
            \braket*{\Xi_{a,X \mid a_{\bar{k}},\alpha_{\bar{k}}}}
                     {\Xi_{a',Z \mid a_{\bar{k}},\alpha_{\bar{k}}}}_{T_k\, I_{k+1}^{k+l_c}\, M_{k+1}^{k+l_c}\, T_{k+1}^{k+l_c}}
        \Bigg),
    \end{equation}
    where, for each round $k$, bit value $a_k \in \{0,1\}$, and
    basis $\alpha_k \in \{Z,X\}$,
    \begin{equation}
    \label{eq:Xi_decoy}
        \ket*{\Xi_{a_k,\alpha_k \mid a_{\bar{k}},\alpha_{\bar{k}}}}
            _{T_k\, I_{k+1}^{k+l_c}\, M_{k+1}^{k+l_c}\, T_{k+1}^{k+l_c}}
        \coloneqq
        \ket*{\psi_{a_{k-l_c}^{k},\,\alpha_{k-l_c}^{k}}^{(1)}}_{T_k}
        \;\otimes
        \bigotimes_{j=k+1}^{k+l_c}
            \ket*{\Psi_{a_{j-l_c}^{j},\,\alpha_{j-l_c}^{j}}}
                _{I_j M_j T_j}\,.
    \end{equation}
\end{lemma}

\begin{proof}
    See \cref{appsub:proof_bernstein_by_groups_decoy}.
\end{proof}

\subsubsection*{Decoy-state bounds and final security statement}

By combining \cref{lem:quantum_coin_inequality_decoy,lem:bernstein_by_groups_decoy}, one can obtain a bound on the single-photon phase-error rate as a function of $\bm{n}_{\mathrm{keep},\mathrm{sifted},Z,(1)}^{\mathrm{det}}$, $\bm{n}_{\mathrm{keep},\mathrm{sifted},X,(1)}^{\mathrm{det}}$ and $\bm{n}_{\mathrm{keep},\mathrm{sifted},X,(1)}^{\mathrm{err}}$. However, the value of these quantities is not contained in the announced data vector $\bm{\vec n}$, since Alice never announces (or even actually knows) the photon number of her signals. The vector $\bm{\vec n}$ only contains the values $\{\bm{n}_{\mathrm{keep},\mathrm{sifted},Z,\mu}^{\mathrm{det}}\}_\mu$, $\{\bm{n}_{\mathrm{keep},\mathrm{sifted},X,\mu}^{\mathrm{det}}\}_\mu$ and $\{\bm{n}_{\mathrm{keep},\mathrm{sifted},X,\mu}^{\mathrm{err}}\}_\mu$, where
\begin{itemize}
         \item $\bm{n}_{\mathrm{keep},\mathrm{sifted},Z,\mu}^{\mathrm{det}}$ = total number of detected {keep-sifted} rounds with $\alpha_k = \beta_k = Z$ and $I_k = \mu$;
        \item $\bm{n}_{\mathrm{keep},\mathrm{sifted},X,\mu}^{\mathrm{err}}$ = number of errors in detected {keep-sifted} rounds with $\alpha_k = \beta_k = X$ and $I_k = \mu$;
        \item $\bm{n}_{\mathrm{keep},\mathrm{sifted},X,\mu}^{\mathrm{det}}$ = total number of detected {keep-sifted} rounds with $\alpha_k = \beta_k = X$ and $I_k = \mu$.
\end{itemize}

{Importantly}, note that, conditional on the photon number $m_k$, the emitted state is independent of the intensity label $\mu_k$. Thus, one can simply apply the standard decoy-state method to obtain bounds on $\bm{n}_{\mathrm{keep},\mathrm{sifted},Z,(1)}^{\mathrm{det}}$, $\bm{n}_{\mathrm{keep},\mathrm{sifted},X,(1)}^{\mathrm{det}}$ and $\bm{n}_{\mathrm{keep},\mathrm{sifted},X,(1)}^{\mathrm{err}}$ using the observed $\{\bm{n}_{\mathrm{keep},\mathrm{sifted},Z,\mu}^{\mathrm{det}}\}_\mu$, $\{\bm{n}_{\mathrm{keep},\mathrm{sifted},X,\mu}^{\mathrm{det}}\}_\mu$ and $\{\bm{n}_{\mathrm{keep},\mathrm{sifted},X,\mu}^{\mathrm{err}}\}_\mu$. Once one has {derived} these bounds, one can directly prove security using the following theorem.
\begin{theorem}[Phase-error-rate bound for decoy-state BB84 with encoding correlations from quantum coin analysis]
\label{thm:phase_error_bound_quantum_coin_decoy}
    Consider the \emph{Phase-error estimation protocol} with encoding correlations up to maximum length $l_c$. Let  $\bm{n}_{\mathrm{keep},\mathrm{sifted},Z,(1)}^{\mathrm{det,L}}$, $\bm{n}_{\mathrm{keep},\mathrm{sifted},Z,(1)}^{\mathrm{det,U}}$, $\bm{n}_{\mathrm{keep},\mathrm{sifted},X,(1)}^{\mathrm{det,L}}$ and $\bm{n}_{\mathrm{keep},\mathrm{sifted},X,(1)}^{\mathrm{err,U}}$ be random variables that are functions of the announced data vector $\bm{\vec n}$. Consider the event
    \begin{equation}
    \label{eq:Omega_decfail}
        \begin{aligned}
        \Omega_\mathrm{dec,fail} &= \Big\{\bm{n}_{\mathrm{keep},\mathrm{sifted},Z,(1)}^{\mathrm{det}} < \bm{n}_{\mathrm{keep},\mathrm{sifted},Z,(1)}^{\mathrm{det,L}}  \cup \bm{n}_{\mathrm{keep},\mathrm{sifted},Z,(1)}^{\mathrm{det}} > \bm{n}_{\mathrm{keep},\mathrm{sifted},Z,(1)}^{\mathrm{det,U}} \\
        &~~~~~~~\cup \,\bm{n}_{\mathrm{keep},\mathrm{sifted},X,(1)}^{\mathrm{det}}< \bm{n}_{\mathrm{keep},\mathrm{sifted},X,(1)}^{\mathrm{det,L}} \cup \bm{n}_{\mathrm{keep},\mathrm{sifted},X,(1)}^{\mathrm{err}} > \bm{n}_{\mathrm{keep},\mathrm{sifted},X,(1)}^{\mathrm{err,U}} \Big\},
        \end{aligned}
    \end{equation}
    and suppose we have the guarantee that
    \begin{equation}
        \label{eq:pr_decfail}\Pr[\Omega_\mathrm{dec,fail}] \leq \epsilon_{\rm decoy}.
    \end{equation}
    Then, it follows that
    \begin{equation}
        \Pr[\bm{n}_{\mathrm{keep},\mathrm{sifted},Z,(1)}^{\mathrm{det}} < \bm{n}_{\mathrm{keep},\mathrm{sifted},Z,(1)}^{\mathrm{det,L}}\quad \cup \quad \bm\ephone > \mathcal{E}_{\mathrm{ph},(1)} (\bm{\vec{n}};\epsilon_A,\epsilon_C) ] \leq 5\epsilon_A + (l_c+1)\epsilon_C + \epsilon_\mathrm{decoy},
    \end{equation}
    where
    \begin{equation}
        \label{eq:E_ph_1_def}\mathcal{E}_{\mathrm{ph},(1)} (\bm{\vec{n}};\epsilon_A,\epsilon_C) \coloneqq \frac{(\bm{n}_{\mathrm{keep},\mathrm{sifted},Z,(1)}^{\mathrm{det,U}} + \Delta_A) \, G_+ \left(\frac{\bm{n}_{\mathrm{keep},\mathrm{sifted},X,(1)}^{\mathrm{err,U}} + \Delta_A}{\bm{n}_{\mathrm{keep},\mathrm{sifted},X,(1)}^{\mathrm{det,L}} - \Delta_A}, 1 - \frac{2p_{\mathrm{keep}} (n_{\mathrm{trash},\mathrm{sifted},-,(1)}^\mathrm{U} + \Delta_A)}{ (1-p_{\mathrm{keep}})(\bm{n}_{\mathrm{keep},\mathrm{sifted},Z,(1)}^{\mathrm{det,U}}+\bm{n}_{\mathrm{keep},\mathrm{sifted},X,(1)}^{\mathrm{det,L}})}\right) + \Delta_A}{\bm{n}_{\mathrm{keep},\mathrm{sifted},Z,(1)}^{\mathrm{det,L}}}.
    \end{equation}
\end{theorem}

\begin{proof}
See Appendix \ref{appsub:proof:thm:phase_error_bound_quantum_coin_decoy}.
\end{proof}

To obtain the guarantee in \cref{eq:Omega_decfail,eq:pr_decfail}, one can simply apply the standard decoy state analysis to obtain bounds on $\bm{n}_{\mathrm{keep},\mathrm{sifted},\gamma,(1)}^{\zeta}$ using the observed $\{\bm{n}_{\mathrm{keep},\mathrm{sifted},\gamma,\mu}^{\zeta}\}_\mu$ for each $(\gamma, \zeta) \in \{(Z, \mathrm{det}), (X, \mathrm{det}), (X, \mathrm{err})\}$. This is because the emitted signals are diagonal in the Fock basis with photon-number distribution $p_{m\vert \mu}$. Thus, one can consider the typical decoy-state equivalent counterfactual scenario in which Alice and Bob first determine the values of  $\bm{n}_{\mathrm{keep},\mathrm{sifted},\gamma,(m)}^{\zeta}$, and then assign each of these events to intensity $\mu$ with probability %
\begin{equation}
    p_{\mu \vert m} = \frac{p_{\mu} p_{m\vert \mu}}{\sum_{\nu}  p_{\nu} p_{m \vert \nu}},
\end{equation}
given by Bayes' rule. In general, the decoy bounds can be obtained by applying concentration inequalities for sums of independent Bernoulli random variables such as Chernoff bounds, and then solving a linear program. For some practical cases, there exist analytical bounds. Here we use analytical bounds from Ref.~\cite{mannalathSharpFinite2025} for the case of three intensities $s$ {(signal)}, $w$ {(weak)} and $v$ {(vacuum)}, {which lead to the following result.}

\begin{lemma}[Decoy-state bounds]
\label{lem:decoy_state_bounds}
For the three-intensity {decoy-state protocol} with intensities $s$, $w$, and $v$, with $s > w > v$, define
\begin{equation}
\label{eq:fL_fU}
\begin{gathered}   
    f_L(\{M_\mu\}_{\mu \in \{s,w,v\}}; \epsilon_D) = \frac{p_1 s}{s(w-v)-w^2+v^2}\left(\frac{e^{w}}{p_{w}}M_{w}^{-}-\frac{e^{v}}{p_{v}}M_{v}^{+}-\frac{w^2-v^2}{s^2} \frac{e^{s}}{p_{s}}M_{s}^{+}\right), \\
    f_U(\{M_\mu\}_{\mu \in \{s,w,v\}}; \epsilon_D) = \frac{p_1}{w-v} \left(\frac{e^{w}}{p_{w}} M_{w}^{+}-\frac{e^{v}}{p_{v}} M_{v}^{-}\right)
\end{gathered}
\end{equation}
where $p_1 = \sum_\mu p_{\mu} p_{1\vert \mu} = \sum_{\mu} e^{-\mu} \mu \,p_{\mu}$, {$M = \sum_\mu M_\mu$, $M_\mu^{\pm} = M\,\Gamma^{\pm}_{M,\epsilon_D}\!\left(M_\mu/M\right)$, and $\Gamma^{\pm}$ is given in \cref{lem:relaxed_chernoff} of \cref{app:concentration_ineqs} (taken from Ref.~\cite[Prop.\,1]{mannalathSharpFinite2025}).} Note that we have not explicitly stated the dependence of the functions $f_L$ and $f_U$ on $(s,w,v,p_s,p_w,p_v)$ for simplicity of notation. 

Then, for each $(\gamma, \zeta) \in \{(Z, \mathrm{det}), (X, \mathrm{det}),(X,\mathrm{err})\}$, the inequality
\begin{equation}
    \label{eq:n_L}
    \bm{n}_{\mathrm{keep},\mathrm{sifted},\gamma,(1)}^{\zeta} \geq f_L(\{\bm{n}_{\mathrm{keep},\mathrm{sifted},\gamma,\mu}^{\zeta}\}_{\mu \in \{s,w,v\}}; \epsilon_D) \eqqcolon \bm{n}_{\mathrm{keep},\mathrm{sifted},\gamma,(1)}^{\zeta,\mathrm{L}}
\end{equation}
holds with probability at least $1-3\epsilon_D$, and the inequality
\begin{equation}
    \label{eq:n_U}
    \bm{n}_{\mathrm{keep},\mathrm{sifted},\gamma,(1)}^{\zeta} \leq f_U(\{\bm{n}_{\mathrm{keep},\mathrm{sifted},\gamma,\mu}^{\zeta}\}_{\mu \in \{s,w,v\}}; \epsilon_D) \eqqcolon \bm{n}_{\mathrm{keep},\mathrm{sifted},\gamma,(1)}^{\zeta,\mathrm{U}}
\end{equation}
holds with probability at least $1-2\epsilon_D$. Thus, with these definitions, \cref{eq:Omega_decfail,eq:pr_decfail} are satisfied with $\epsilon_\mathrm{decoy} = 10 \epsilon_D$.
\end{lemma}

\begin{proof}
    Follows from the decoy-state analysis in Ref.~\cite{mannalathSharpFinite2025}. The total failure probability $\epsilon_\mathrm{decoy} = 10 \epsilon_D$ follows from a union bound: $\Omega_\mathrm{dec,fail}$ in \cref{eq:Omega_decfail} contains two lower-bound failure events (each with probability $\leq 3\epsilon_D$) and two upper-bound failure events (each with probability $\leq 2\epsilon_D$).
\end{proof}

Finally, combining \cref{thm:phase_error_bound_quantum_coin_decoy} with the decoy-state bounds in \cref{lem:decoy_state_bounds}, we obtain the following result.

\begin{corollary}[Single-photon phase-error rate bound for decoy-state BB84 with encoding correlations]
\label{cor:phase_error_bound_decoy}
Consider the \emph{Phase-error estimation protocol} with encoding correlations up to maximum length $l_c$ and three-intensity decoy states with intensities $s$, $w$, and $v$ such that $s > w > v$. Define $\bm{n}_{\mathrm{keep},\mathrm{sifted},\gamma,(1)}^{\zeta,\mathrm{L}}$ and $\bm{n}_{\mathrm{keep},\mathrm{sifted},\gamma,(1)}^{\zeta,\mathrm{U}}$ for $(\gamma, \zeta) \in \{(Z, \mathrm{det}), (X, \mathrm{det}),(X,\mathrm{err})\}$ as in \cref{eq:n_L,eq:n_U}. Then,
\begin{equation}
        \Pr[\bm{n}_{\mathrm{keep},\mathrm{sifted},Z,(1)}^{\mathrm{det}} < \bm{n}_{\mathrm{keep},\mathrm{sifted},Z,(1)}^{\mathrm{det,L}}\quad \cup \quad \bm\ephone > \mathcal{E}_{\mathrm{ph},(1)} (\bm{\vec{n}};\epsilon_A,\epsilon_C,\epsilon_D) ] \leq 5\epsilon_A + (l_c+1)\epsilon_C + 10 \epsilon_D,
    \end{equation}
where $\mathcal{E}_{\mathrm{ph},(1)} (\bm{\vec{n}};\epsilon_A,\epsilon_C,\epsilon_D)$ {denotes the function defined in \cref{eq:E_ph_1_def}, evaluated with the decoy-state bounds in \cref{eq:n_L,eq:n_U}.} {Consequently, by \cref{thm:variable_length_security_decoy}, if Alice and Bob set the final key length as in \cref{eq:final_key_length}, the protocol produces an $\varepsilon_{\mathrm{sec}}$-secure key with
    \begin{equation}
        \varepsilon_{\mathrm{sec}}
        = 2\sqrt{5\epsilon_A + (l_c+1)\epsilon_C + 10\,\epsilon_D}
        \;+\; \varepsilon_\mathrm{PA} + \varepsilon_\mathrm{EV}.
    \end{equation}}
\end{corollary}

\begin{proof}
Follows directly by combining \cref{thm:phase_error_bound_quantum_coin_decoy,lem:decoy_state_bounds}.
\end{proof}

Note that the statistical bound provided by \cref{cor:phase_error_bound_decoy} is of the form in \cref{eq:statistical_guarantee_decoy}, since $\bm{n}_{\mathrm{keep},\mathrm{sifted},Z,(1)}^{\mathrm{det}} \equiv \bm{n_{K,(1)}}$ and  $\bm{n}_{\mathrm{keep},\mathrm{sifted},Z,(1)}^{\mathrm{det,L}} \equiv \mathcal{N}_{K,(1)}(\bm{\vec n})$. Thus, one can directly prove security by combining \cref{thm:variable_length_security_decoy,cor:phase_error_bound_decoy}.

\subsection{{Application to a partially characterized LTI correlations model}}
\label{subsec:particular_correlations_model}

{We now illustrate how to apply the general framework developed above to a concrete, experimentally motivated correlations model, which we also employ in the numerical evaluations of \cref{sec:numerics}. We emphasize that none of the results derived so far rely on this model: as anticipated in \cref{sec:assumptions}, the specific form of the correlated states enters the analysis only through the coin parameter $\Delta_{\rm coin}^{\rm U}$ defined in \cref{eq:Delta_coin_U_decoy}, and the sole purpose of the model below is to enable an explicit bound on this parameter.}

\subsubsection*{{The LTI correlations model}}

{Concretely, we consider that the eigenstates in \cref{eq:rho_muk_a1k_alpha1k} have the form}
    \begin{equation}
    \label{eq:correlations_model_state}
    \ket*{\psi_{a_1^k,\alpha_1^k}^{(m)}}_{T_k} = \frac{1}{\sqrt{m!}}\left(\cos\big(\theta_{a_1^k,\alpha_1^k}\big) \, a_0^\dagger + \sin\big(\theta_{a_1^k,\alpha_1^k}\big) \, a_1^\dagger\right)^m \ket{\rm vac}_{T_k},
\end{equation}
    where $a_0^\dagger$ and $a_1^\dagger$ are the creation operators of two orthogonal modes (e.g., horizontal and vertical polarizations). This state represents $m$ identical photons encoded in the $XZ$ plane of the Bloch sphere, with an encoding phase $\theta_{a_1^k,\alpha_1^k}$ that depends on the history of bit-and-basis choices, which is a reasonable model for a BB84 encoder suffering from correlations.

    {Next, we specify what we assume to be known about the phases $\theta_{a_1^k,\alpha_1^k}$. We consider} a linear time-invariant (LTI) correlations model, such that the encoding phase can be expressed as
    \begin{equation}
\label{eq:correlations_model_phase}
\theta_{a_1^k,\alpha_1^k} = \hat\theta_{a_k,\alpha_k} + \sum_{l=1}^{k-1} \delta_{a_{k-l},\alpha_{k-l}}^{(l)},
    \end{equation}
    where $\hat\theta_{a_k,\alpha_k}$ is the ideal BB84 encoding phase for bit $a_k$ and basis $\alpha_k$ (i.e., $\hat\theta_{0,Z} = 0$, $\hat\theta_{1,Z} = \pi/2$, $\hat\theta_{0,X} = \pi/4$, $\hat\theta_{1,X} = -\pi/4$), and $\delta_{a_{k-l},\alpha_{k-l}}^{(l)}$ quantifies the residual contribution from the bit-and-basis choices made $l$ rounds earlier. We assume that the exact value of the coefficients $\delta_{a_{k-l},\alpha_{k-l}}^{(l)}$ is unknown, but that their pairwise differences are bounded, i.e.,
    \begin{equation}
        \label{eq:correlations_model_delta}
        \abs{\delta_{a_{k-l},\alpha_{k-l}}^{(l)}-\delta_{a'_{k-l},\alpha'_{k-l}}^{(l)}} \leq \Delta_l,
    \end{equation}
    with $\Delta_l$ decaying exponentially in $l$ as
\begin{equation}
    \label{eq:correlations_model_Delta_exp}
    \Delta_l = \Delta_1 e^{-\lambda (l-1)},
\end{equation}
with $0\leq \Delta_l \leq \pi$ and $\lambda > 0$. We note that experimental studies \cite{agulleiroModelingCharacterization2025} have shown that this model is a good approximation for the correlations introduced by realistic encoders\footnote{{Note that our $\Delta_l$ is precisely the angular deviation that the LTI model of Ref.~\cite{agulleiroModelingCharacterization2025} bounds, and $\lambda$ is its decay rate per round. That work then converts this into a fidelity deficit $\epsilon_l \leq \Delta_l^2$, obtaining an exponential bound with decay constant $C = 2\lambda$; the same fidelity-based convention is used in Ref.~\cite{curras-lorenzoRigorousPhaseerrorestimation2026}. Here we prefer to work with $\Delta_l$ directly, since it is the quantity that enters our state overlaps. In particular, the value $\lambda = 6.35$ used in our simulations {in \cref{sec:numerics}} corresponds to the decay constant $C \approx 12.7$ characterized in Ref.~\cite {agulleiroModelingCharacterization2025}.}}. {Note that \cref{eq:correlations_model_delta} bounds only the differences between residual contributions, so the model does not require knowing the absolute value of any $\delta^{(l)}$, nor the encoder's transfer function.} 

{Since $\Delta_l$ never strictly vanishes, this model contains correlations of unbounded length.} Our security analysis, however, requires a source with a strictly finite maximum correlation length~$l_c$ in order to apply \cref{cor:phase_error_bound_decoy}. We bridge this gap through the following strategy introduced {in Ref.~\cite{curras-lorenzoRigorousPhaseerrorestimation2026} (see also Ref.~\cite{pereiraQuantumKey2024a})}:
\begin{enumerate}
    \item {we} define a fictitious source with correlations truncated at a finite length~$l_c^\mathrm{eff}${;}
    \item {we} prove security for the fictitious source using the tools established {in the previous section;}
    \item {we} bound the trace distance between the actual and fictitious source-replacement states, and lift the security guarantee to the actual source via \cref{lem:unbounded_correlations} below.
\end{enumerate}

\subsubsection*{Lifting security from bounded to unbounded correlations}

The following lemma, which is model-independent, provides the mechanism for step~3.

\begin{lemma}[Unbounded correlations~\cite{curras-lorenzoRigorousPhaseerrorestimation2026}]
    \label{lem:unbounded_correlations}
    Let $\ket*{\Psi_N^{(\infty)}}_{C_1^N A_1^N I_1^N M_1^N T_1^N}$ be the source-replacement state for a source with correlations of unbounded length, and let $\ket*{\Psi_N^{(l_c^\mathrm{eff})}}_{C_1^N A_1^N I_1^N M_1^N T_1^N}$ be the source-replacement state for a fictitious source with correlations truncated at length~$l_c^\mathrm{eff}$. Suppose that the trace distance between these two states satisfies
    \begin{equation}
        T\Big(\ketbra*{\Psi_N^{(\infty)}},\ketbra*{\Psi_N^{(l_c^\mathrm{eff})}}\Big) \leq d,
    \end{equation}
    {for some $d \geq 0$,} and that for the fictitious source $\ket*{\Psi_N^{(l_c^\mathrm{eff})}}$, the following holds for any eavesdropping attack:
    \begin{equation}
        \label{eqapp:eph_bound_fic}
        \Pr_{(l_c^\mathrm{eff})}\!\Big[\bm{n_{K,(1)}} < \mathcal{N}_{K,(1)}(\bm{\vec{n}}) \;\cup\; \bm\ephone > \mathcal{E}_{\mathrm{ph},(1)}(\bm{\vec{n}})\Big] \leq \varepsilon_\mathrm{PE},
    \end{equation}
    {where the subscript in $\Pr_{(l_c^\mathrm{eff})}$ (respectively, $\Pr_{(\infty)}$ below) indicates that the probability is evaluated for the \emph{Phase-error estimation protocol} run with initial state $\ket*{\Psi_N^{(l_c^\mathrm{eff})}}$ (respectively, $\ket*{\Psi_N^{(\infty)}}$).} Then, for the actual source $\ket*{\Psi_N^{(\infty)}}$, the following holds for any eavesdropping attack:
    \begin{equation}
        \label{eqapp:eph_bound_act}
        \Pr_{(\infty)}\!\Big[\bm{n_{K,(1)}} < \mathcal{N}_{K,(1)}(\bm{\vec{n}}) \;\cup\; \bm\ephone > \mathcal{E}_{\mathrm{ph},(1)}(\bm{\vec{n}})\Big] \leq \varepsilon_\mathrm{PE} + d.
    \end{equation}
\end{lemma}

\begin{proof}
    Follows almost directly from the operational interpretation of the trace distance as the maximum distinguishing probability. See full proof in Ref.~\cite[Lemma~1]{curras-lorenzoRigorousPhaseerrorestimation2026}.
\end{proof}

\subsubsection*{The truncated correlations model}

Following the strategy above, we define a fictitious source identical to that in \cref{eq:correlations_model_state,eq:correlations_model_phase,eq:correlations_model_delta,eq:correlations_model_Delta_exp}, but with the actual bit and basis choices beyond length $l_c^\mathrm{eff}$ replaced by fixed reference choices $a^\ast,\alpha^\ast$. Concretely, the $m$-photon state for round~$k$ in the truncated model is
\begin{equation}
    \label{eq:correlations_model_state_truncated}
    \ket*{\psi_{a_{k-l_c^\mathrm{eff}}^k,\alpha_{k-l_c^\mathrm{eff}}^k}^{(m)}}_{T_k}
    = \frac{1}{\sqrt{m!}}
    \left(\cos\!\left(\theta_{a_{k-l_c^\mathrm{eff}}^k,\alpha_{k-l_c^\mathrm{eff}}^k}\right) a_0^\dagger
    + \sin\!\left(\theta_{a_{k-l_c^\mathrm{eff}}^k,\alpha_{k-l_c^\mathrm{eff}}^k}\right) a_1^\dagger\right)^m
    \ket{\mathrm{vac}}_{T_k},
\end{equation}
with the encoding phase now given by
\begin{equation}
    \label{eq:correlations_model_phase_truncated}
    \theta_{a_{k-l_c^\mathrm{eff}}^k,\alpha_{k-l_c^\mathrm{eff}}^k}
    = \hat\theta_{a_k,\alpha_k}
    + \sum_{l=1}^{l_c^\mathrm{eff}} \delta_{a_{k-l},\alpha_{k-l}}^{(l)} + \sum_{l=l_c^\mathrm{eff}+1}^{k-1} \delta_{a^\ast,\alpha^\ast}^{(l)},
\end{equation}
where the residual contributions satisfy the same bounds as before,
\begin{equation}
    \label{eq:correlations_model_delta_Delta_exp_truncated}
    \abs{\delta_{a_{k-l},\alpha_{k-l}}^{(l)}-\delta_{a'_{k-l},\alpha'_{k-l}}^{(l)}} \leq \Delta_l,
    \qquad
    \Delta_l = \Delta_1 e^{-\lambda(l-1)}.
\end{equation}
The key difference from the original model in \cref{eq:correlations_model_phase} is that the actual bit and basis choices for $l > l_c^\mathrm{eff}$ have been substituted by some fixed choices $a^\ast,\alpha^\ast$. Now, by construction, the state of round~$k$ depends on at most the $l_c^\mathrm{eff}$ preceding rounds, so \cref{cor:phase_error_bound_decoy} is directly applicable to the truncated source with $l_c = l_c^\mathrm{eff}$. To apply it, we need a bound on the parameter $\Delta_{\rm coin}^{\rm U}$ defined in \cref{eq:Delta_coin_U_decoy}, which the following result provides.

\begin{lemma}[Coin parameter for the truncated LTI model]
    \label{lem:bound_Delta_coin_decoy}
    Consider the \emph{Phase-error estimation protocol} with the truncated correlations model of \cref{eq:correlations_model_state_truncated,eq:correlations_model_phase_truncated,eq:correlations_model_delta_Delta_exp_truncated}. Then the coin parameter $\Delta_{\rm coin}^{\rm U}$ in \cref{eq:Delta_coin_U_decoy} satisfies
    \begin{equation}
        \label{eq:Delta_coin_U_decoy_bound}
        {\Delta_{\rm coin}^{\rm U}
        \leq \frac{1}{2}\left[
            1 - \prod_{l=1}^{l_c^\mathrm{eff}}
            \sum_{\mu} p_{\mu}\,
            \exp\!\Big[-\mu\,\Big(1-\cos\big(\Delta_l\big)\Big)\Big]
        \right].}
    \end{equation}
\end{lemma}

\begin{proof}
See Appendix~\ref{appsub:proof:lem:bound_Delta_coin_decoy}.

\end{proof}

\subsubsection*{Bounding the truncation error}

The following lemma quantifies the cost of truncating the correlations at length~$l_c^\mathrm{eff}$.

\begin{lemma}[Trace distance for exponentially decaying correlations]
    \label{lem:trace_distance_exponential_decay_decoy}
    Let $\ket*{\Psi_N^{(\infty)}}$ denote the source-replacement state in \cref{eq:global_state_decoy_correlations} for the unbounded model of \cref{eq:correlations_model_state,eq:correlations_model_phase,eq:correlations_model_delta,eq:correlations_model_Delta_exp}, and let $\ket*{\Psi_N^{(l_c^\mathrm{eff})}}$ denote the corresponding state for the truncated model of \cref{eq:correlations_model_state_truncated,eq:correlations_model_phase_truncated,eq:correlations_model_delta_Delta_exp_truncated}. Then
\begin{equation}
        \label{eq:trace_distance_asymp_decoy}
        T\Big(
            \ketbra*{\Psi_N^{(\infty)}},
            \ketbra*{\Psi_N^{(l_c^\mathrm{eff})}}
        \Big)
        \leq
\sqrt{N\bar{\mu}}
\sum_{l=l_c^\mathrm{eff}+1}^{\infty}\Delta_l
=
\frac{
    \sqrt{N\bar{\mu}}\,\Delta_1\,
    e^{-\lambda l_c^\mathrm{eff}}
}{
    1-e^{-\lambda}
},
    \end{equation}
    where $\bar{\mu} = \sum_{\mu} p_{\mu}\, \mu$ is the average intensity. To ensure {that the bound holds} 
    for a given tolerance~$d > 0$, it suffices to choose
\begin{equation}
\label{eq:l_c_formula_stable_decoy}
l_c^\mathrm{eff}
= \Bigg\lceil
\frac{1}{\lambda}
\ln\!\left(
    \frac{
        \sqrt{N\bar{\mu}}\,\Delta_1
    }{
        d(1-e^{-\lambda})
    }
\right)\Bigg\rceil.
\end{equation}

\end{lemma}

\begin{proof}
See Appendix~\ref{appsub:proof:lem:trace_distance_exponential_decay_decoy}.
\end{proof}

Note that $l_c^\mathrm{eff}$ grows only logarithmically in the block size~$N$, so the truncation introduces a very mild overhead.

\subsubsection*{Combined security statement}

Combining the results above, we obtain a complete security guarantee for the LTI correlations model.

\begin{corollary}[Security of decoy-state BB84 with LTI encoding correlations]
    \label{cor:security_LTI}
    Consider the \emph{Actual protocol} with a source described by the LTI correlations model of \cref{eq:correlations_model_state,eq:correlations_model_phase,eq:correlations_model_delta,eq:correlations_model_Delta_exp}, three-intensity decoy states with intensities $s > w > v$, and a truncation length~$l_c^\mathrm{eff}$ satisfying \cref{eq:l_c_formula_stable_decoy} for some tolerance~$d > 0$. Define the decoy-state bounds $\bm{n}_{\mathrm{keep},\mathrm{sifted},\gamma,(1)}^{\zeta,\mathrm{L}}$ and $\bm{n}_{\mathrm{keep},\mathrm{sifted},\gamma,(1)}^{\zeta,\mathrm{U}}$ as in \cref{eq:n_L,eq:n_U}, and the phase-error rate bound $\mathcal{E}_{\mathrm{ph},(1)}$ as in \cref{eq:E_ph_1_def}, with $\Delta_{\rm coin}^{\rm U}$ evaluated via \cref{eq:Delta_coin_U_decoy_bound}. Then
    \begin{equation}
        \Pr\!\Big[
            \bm{n_{K,(1)}} < \bm{n}_{\mathrm{keep},\mathrm{sifted},Z,(1)}^{\mathrm{det,L}}
            \;\cup\;
            \bm\ephone > \mathcal{E}_{\mathrm{ph},(1)}(\bm{\vec{n}};\epsilon_A,\epsilon_C,\epsilon_D)
        \Big]
        \leq 5\epsilon_A + (l_c^\mathrm{eff}+1)\epsilon_C + 10\,\epsilon_D + d.
    \end{equation}
    Consequently, by \cref{thm:variable_length_security_decoy}, if Alice and Bob set the final key length as in \cref{eq:final_key_length}, the protocol produces an $\varepsilon_{\mathrm{sec}}$-secure key with
    \begin{equation}
        \varepsilon_{\mathrm{sec}}
        = 2\sqrt{5\epsilon_A + (l_c^\mathrm{eff}+1)\epsilon_C + 10\,\epsilon_D + d}
        \;+\; \varepsilon_\mathrm{PA} + \varepsilon_\mathrm{EV}.
    \end{equation}
\end{corollary}

\begin{proof}
    By \cref{cor:phase_error_bound_decoy}, the phase-error rate bound holds for the truncated source $\ket*{\Psi_N^{(l_c^\mathrm{eff})}}$ with failure probability $\varepsilon_\mathrm{PE} = 5\epsilon_A + (l_c+1)\epsilon_C + 10\,\epsilon_D$. By \cref{lem:trace_distance_exponential_decay_decoy}, the trace distance between the actual and truncated source-replacement states is at most~$d$ for the chosen~$l_c^\mathrm{eff}$. Applying \cref{lem:unbounded_correlations} then lifts the guarantee to the actual source, adding~$d$ to the failure probability. The final security parameter follows from \cref{thm:variable_length_security_decoy}.
\end{proof}

{ \section{Numerical results} \label{sec:numerics}

In this section, we numerically evaluate the secret-key rate achievable using our security analysis. We consider a decoy-state BB84 protocol whose source suffers from exponentially decaying LTI correlations, as described {in \cref{subsec:particular_correlations_model}}. {The simulations are based on a standard channel and detection model for decoy-state BB84 with an active-basis-choice receiver with two threshold single-photon detectors; the explicit expressions for the expected detection and error statistics are taken from Ref.~\cite{panayiMemoryassistedMeasurementdeviceindependent2014} and provided in \cref{app:channel_model} for completeness.} We consider an error-correction inefficiency $f = 1.16$, a dark-count probability $p_d = 10^{-6}$ per detector, a detector efficiency $\eta_d = 0.73$ {\cite{pittaluga600kmRepeaterlike2021}}, a fiber attenuation coefficient of $\SI{0.2}{dB/km}$, no misalignment, and a block size of $N = 10^{12}$ transmitted signals. For the correlations model, we set $\lambda = 6.35$~\cite{agulleiroModelingCharacterization2025} and consider two representative values of the nearest-neighbor correlation strength, $\Delta_1 \in \{0.01, 0.001\}$.

{In the simulations, we evaluate} two scenarios that differ in how the correlations are modeled:
\begin{itemize}
    \item \textit{Finite correlation length} ($l_c = 1$ and $l_c = 5$): we assume that the correlations are artificially truncated at length $l_c$, i.e.,~$\Delta_l = \Delta_1 \, e^{-\lambda(l-1)}$ for $l \leq l_c$ and $\Delta_l = 0$ for $l > l_c$. Security is established via \cref{cor:phase_error_bound_decoy}. We consider these only for comparison purposes.
    \item \textit{Unbounded correlation length} ($l_c = \infty$): we consider a natural scenario in which the correlations decay exponentially without any truncation, i.e.,~$\Delta_l = \Delta_1 \, e^{-\lambda(l-1)}$ for all $l$. Security is rigorously established via \cref{cor:security_LTI}.
\end{itemize}

In both cases, we target an overall security parameter of $\varepsilon_\mathrm{sec} = 2 \cdot 10^{-10}$, and fix $\varepsilon_\mathrm{EV} = 10^{-10}$ and $\varepsilon_\mathrm{PA} = \varepsilon_\mathrm{sec}/6$. {Since, by \cref{cor:phase_error_bound_decoy,cor:security_LTI,thm:variable_length_security_decoy}, the overall security parameter is given by $\varepsilon_\mathrm{sec} = 2\sqrt{\varepsilon_\mathrm{PE}} + \varepsilon_\mathrm{PA} + \varepsilon_\mathrm{EV}$, where $\varepsilon_\mathrm{PE}$ is the total failure probability of the parameter estimation, these choices fix $2\sqrt{\varepsilon_\mathrm{PE}} = \varepsilon_\mathrm{sec} - \varepsilon_\mathrm{PA} - \varepsilon_\mathrm{EV} = \tfrac{2}{3}\cdot 10^{-10}$, i.e.,~$\varepsilon_\mathrm{PE} = 10^{-20}/9$. The remaining failure parameters are then allocated by distributing this overall budget evenly among the individual failure events, as follows.} For the scenario with finite correlation length, we set
\begin{equation}
\epsilon_A = \epsilon_C = \epsilon_D = \frac{10^{-20}}{9\,(16 + l_c)},
\end{equation}
{which achieves the required $\varepsilon_\mathrm{sec}$. Here, the factor $16 + l_c$ is simply the total number of individual failure events in \cref{cor:phase_error_bound_decoy}: five applications of Azuma's inequality, $l_c+1$ applications of Bernstein's inequality (one per group), and ten decoy-state bounds, so that $\varepsilon_\mathrm{PE} = (16+l_c)\,\epsilon_A = 10^{-20}/9$. We remark that an even split is not necessarily optimal, but since these parameters enter the bounds only through the deviation terms of the concentration inequalities, which depend only logarithmically on the inverse of the failure probabilities, we expect the potential gain from optimizing this allocation to be minor.} 

{For the scenario with unbounded correlation length, the parameter estimation additionally fails with probability at most $d$ due to the truncation, i.e.,~$\varepsilon_\mathrm{PE} = 5\epsilon_A + (l_c^\mathrm{eff}+1)\epsilon_C + 10\epsilon_D + d$. In this case, the allocation cannot be done in a single step, since the number of Bernstein terms depends on $l_c^\mathrm{eff}$, which is itself determined by $d$ through \cref{eq:l_c_formula_stable_decoy}. We therefore proceed in two steps. First, we assign to $d$ one seventeenth of the overall failure probability budget by setting $d = \tfrac{1}{17}\cdot\tfrac{10^{-20}}{9} = 10^{-20}/153$, and compute the resulting effective truncation length $l_c^\mathrm{eff}$ via \cref{eq:l_c_formula_stable_decoy}. Then, we distribute the remaining budget of $\tfrac{10^{-20}}{9} - d = \tfrac{16}{153}\cdot 10^{-20}$ evenly among the $16 + l_c^\mathrm{eff}$ remaining failure events, setting
\begin{equation}
\epsilon_A = \epsilon_C = \epsilon_D = \frac{16 \cdot 10^{-20}}{153\,(16 + l_c^\mathrm{eff})},
\end{equation}
thus achieving the required $\varepsilon_\mathrm{sec}$.}

The results are shown in \cref{fig:results}. As expected, the secret-key rate decreases as $\Delta_1$ increases, reflecting the larger deviations from the ideal BB84 states. For the exponential-decay model considered here, the residual contributions $\Delta_l$ vanish rapidly with $l$, so the overall impact of the correlations is dominated by the first few terms. Consequently, treating the correlations as unbounded ($l_c = \infty$) incurs essentially no penalty relative to the truncated cases with finite $l_c$, demonstrating that our framework can rigorously handle realistic decaying correlations at virtually no cost in performance. 

{In \cref{fig:results_SP} of \cref{app:single_photon}, we also show the results for an implementation with a single-photon source under the same parameter regime. A similar behavior is observed, although the performance gap between different values of $\Delta_1$ is larger than in \cref{fig:results}. More interestingly, we find that, for high $N$ and $\Delta_1$, the decoy-state implementation actually outperforms the single-photon one, as shown in \cref{fig:results_comparison}. This reverses the behavior of the uncorrelated case, in which a single-photon source is always superior to a weak-coherent-pulse source. The explanation lies in the photon-number statistics of the two sources. Encoding correlations effectively constitute a side channel in which information about the setting choice of a given round is leaked through the pulses emitted in later rounds. However, this leakage only occurs if the latter pulses actually contain photons, since a vacuum emission is identical regardless of the previous settings and therefore leaks nothing. Given that weak coherent pulses are a mixture of photon-number states with a large vacuum component, the impact of correlated bit-and-basis encoders is diminished for such sources.}

\begin{figure}
\includegraphics[width=10cm]{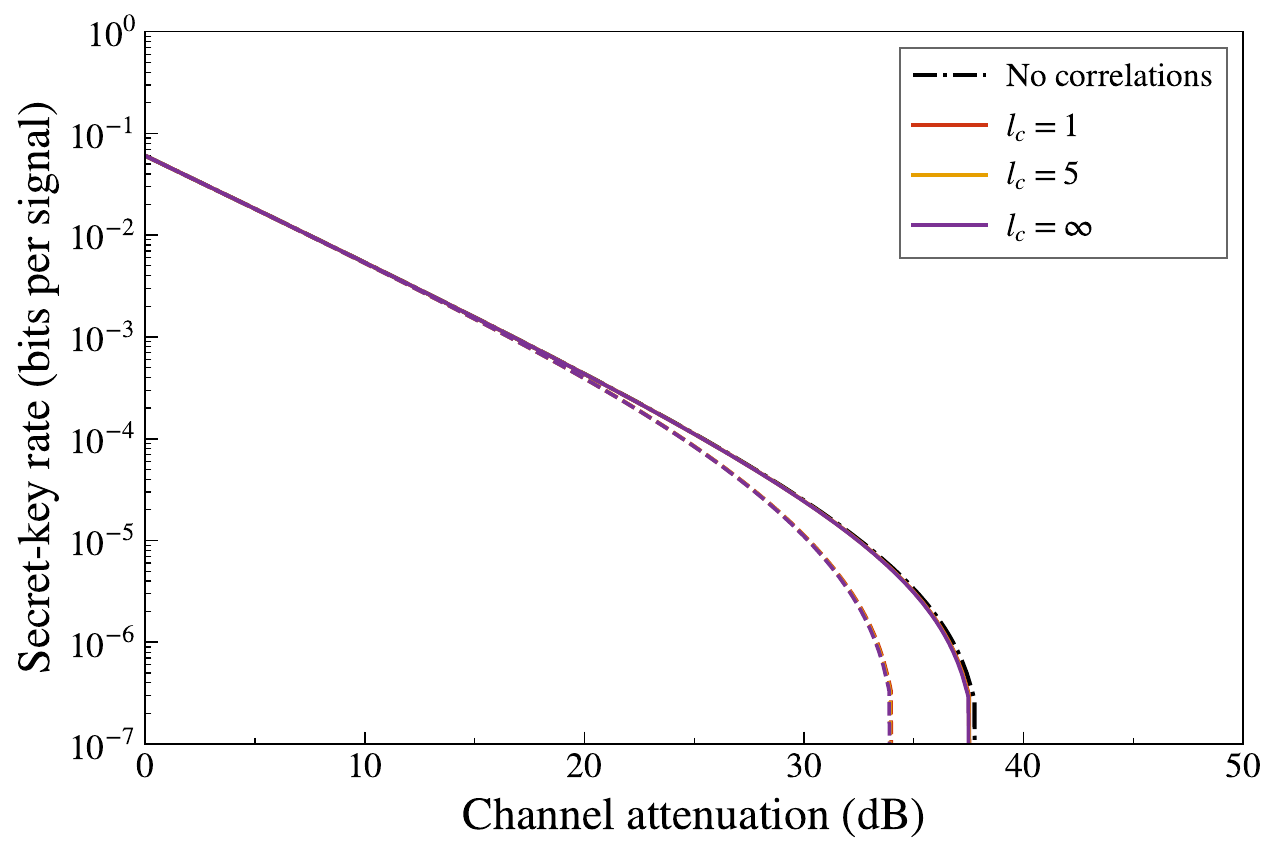}
    \caption{{Secret-key rate as a function of the channel attenuation for a decoy-state BB84 protocol with a source suffering from LTI encoding correlations following the exponential-decay model of {\cref{subsec:particular_correlations_model}}. We employ three intensity settings, fix the weakest to vacuum, and numerically optimize over the remaining two intensities, their associated probabilities{, and $p_\mathrm{keep}$}. Solid lines correspond to $\Delta_1 = 0.001$ and dashed lines to $\Delta_1 = 0.01$. The curves labeled $l_c = 1$ and $l_c = 5$ assume correlations artificially truncated at length $l_c$ (i.e.,~$\Delta_l = 0$ for $l > l_c$), while $l_c = \infty$ corresponds to correlations of unbounded length. {Note that the lines corresponding to different values of $l_c$ are nearly indistinguishable.} The dash-dotted line shows the ideal uncorrelated case for reference. }}
    \label{fig:results}
\end{figure}

}

\section{Conclusion} \label{sec:conclusion}

We have presented a complete, {fully analytical} finite-key security proof for decoy-state BB84 with correlated bit-and-basis encoders, valid against general coherent attacks. To our knowledge, this is the first such result: prior proofs that rigorously accommodate encoding correlations were restricted to sources emitting pure states conditional on Alice's setting choices \cite{curras-lorenzoRigorousPhaseerrorestimation2026}, and thus did not directly cover the phase-randomized coherent-state sources used in most practical QKD implementations{, whose emissions remain mixed even after conditioning on the bit, basis and intensity choices. Moreover, our analysis requires only a partial characterization of the correlations, and it can be combined with existing lifting theorems to accommodate detection efficiency mismatches~\cite{curras-lorenzoSecurityQuantum2026,tupkaryPhaseError2025}}.

Beyond this extension to the decoy-state setting, our approach introduces an important methodological improvement. While the framework of Ref.~\cite{curras-lorenzoRigorousPhaseerrorestimation2026} partitions the protocol rounds into $l_c+1$ groups and estimates the phase-error rate separately within each group, our proof successfully bounds the single-photon phase-error rate of the full sifted key directly, using all the available data simultaneously. This avoids the loss of statistical efficiency inherent to the partitioning step and should translate into tighter finite-key rates for a given block size. More broadly, this shows that the division into independent groups, while a powerful tool, is not a necessary ingredient for incorporating encoding correlations into security proofs of QKD. We anticipate that this observation will prove useful in incorporating source correlations into the security analyses of other QKD protocols.

{To illustrate how to apply our proof, we have considered the LTI correlations model with exponentially decaying correlation strength---the only experimentally characterized model available to date~\cite{agulleiroModelingCharacterization2025}---and handled its unbounded correlation length rigorously via a trace-distance argument, at an essentially negligible cost that grows only logarithmically with the block size. Our results show that the secret-key rate drops as the magnitude of the correlations increases, as expected. Moreover, for the experimentally characterized decay rate, considering correlations of unbounded length is virtually indistinguishable from assuming nearest-neighbour correlations only. Besides, we have also applied our analysis to the case of single-photon sources, and found that although the secret-key rate has a similar behavior to the decoy-state scenario, the latter can actually achieve a better performance when the magnitude of the correlations is large. These results highlight that decoy-state implementations are inherently more robust to correlated bit-and-basis encoders.  }

{Our analysis assumes ideal intensity modulation and ideal phase randomization, i.e., only the bit-and-basis encoder is assumed to suffer from memory effects. Since in practice intensity modulators are also expected to suffer from bandwidth limitations, and gain-switched lasers can suffer from correlated phases between consecutive pulses due to imperfect depletion of the cavity between emissions, extending the proof to incorporate intensity correlations~\cite{yoshinoQuantumKey2018,zapateroSecurityQuantum2021,sixtoSecurityDecoyState2022,trefilovIntensityCorrelations2025} and phase randomization correlations \cite{curras-lorenzoSecurityQuantum2023,marcominiCharacterisingHigherorder2025} is an important next step. However, to our knowledge, no general tools are yet available to incorporate these imperfections into finite-key proofs against general coherent attacks under realistic partially-characterized correlation models, and doing so will likely require the development of novel techniques {(see the Note added below for recent progress on intensity correlations)}.}

{\paragraph*{Note added.}
After the first version of this manuscript was uploaded to the arXiv, two related works involving some of the present authors have appeared. Ref.~\cite{navarreteNumericalSecurity2026} introduces a numerical finite-key security framework against coherent attacks that accommodates partially characterized, correlated bit-and-basis encoders, as well as intensity correlations under a fully characterized photon-number distribution. To handle both types of correlations, it builds on the key technique introduced in the present work: performing the main phase-error estimation on the full, undivided protocol data, and restricting the round partitioning to a fictitious step used only to bound an adversary-independent quantity determined by Alice's local state. However, unlike the present analysis, the encoding correlations model of Ref.~\cite{navarreteNumericalSecurity2026} considers a fixed correlation length and assigns the same worst-case strength to every round within the correlation window, which is pessimistic when the correlations decay slowly with the round separation{, although it is likely the analysis could be adapted to use the same correlations model as in the present work}. Ref.~\cite{sixtoFinitekeySecurity2026} provides an analytical finite-key security proof for decoy-state QKD that additionally incorporates state-preparation flaws, side-channel leakage and detector imperfections, by applying the round-partitioning technique of Ref.~\cite{curras-lorenzoRigorousPhaseerrorestimation2026} to the decoy-state setting, and thus incurs the statistical-efficiency loss that the present analysis avoids. These works are complementary to ours, which remains, to our knowledge, the only fully analytical security proof for decoy-state QKD with correlated encoders that uses all the protocol data simultaneously, without partitioning the rounds into groups. Moreover, unlike Refs.~\cite{navarreteNumericalSecurity2026,sixtoFinitekeySecurity2026}, which require solving semidefinite or linear programs, all the bounds in the present analysis are given in closed form.}

	{\section{Acknowledgements}
    We thank Vaisakh Mannalath and Álvaro Navarrete for valuable discussions. This work was supported by the Galician Regional Government (consolidation of research units: AtlanTTic), the Spanish Ministry of Economy and Competitiveness (MINECO), the Fondo Europeo de Desarrollo Regional (FEDER) through the grant No.~PID2024-162270OB-I00, the “Hub Nacional de Excelencia en Comunicaciones Cuanticas” funded by the Spanish Ministry for Digital Transformation and the Public Service and the European Union NextGenerationEU, the European Union’s Horizon Europe Framework Programme under the Marie Sklodowska-Curie Grant No.~101072637 (Project {"Quantum Safe Internet"}), the project “Quantum Secure Networks Partnership” (QSNP, grant agreement No 101114043) and the Project IberianQCI (grant 101249593), as well as the Programa de Cooperación Interreg VI-A España--Portugal (POCTEP) 2021--2027 through the project QUANTUM IBERIA. G.C.-L. acknowledges funding from the European Union's Horizon Europe research and innovation programme under the Marie Skłodowska-Curie Postdoctoral Fellowship grant agreement No.\ 101149523. K.T. acknowledges support from JSPS KAKENHI Grant Numbers 23K25793 and 23H01096. }


\appendix

\section{Proof of technical results} \label{app:proofs}

\subsection{Proof of \texorpdfstring{\cref{thm:variable_length_security_decoy}}{Theorem 1}}
\label{appsub:proof_thm:variable_length_security_decoy}

\begin{proof}
Let $W$ be the classical register storing the realization of $\bm{\vec n}$; let $\bm{\vec m}$ be a random vector containing the outcomes of Alice's measurements on the photon number registers $M_1^N$ (which are not announced); let $\Omega(\vec n,\vec m)$ be the event in which $\bm{\vec n} = \vec n$ {and} $\bm{\vec m} = \vec m $ is observed, let $\rho_{\vert \Omega(\vec n,\vec m)}$ be the state shared by Alice, Bob and Eve before the beginning of the postprocessing (step \ref*{step:postprocessing}) in the \textit{Actual protocol (source replaced)} conditional on $\Omega(\vec n, \vec m)$; let  $\rho_{\vert \Omega(\vec n, \vec m)}^\mathrm{virt}$ be the state shared by Alice, Bob and Eve at the end of the \textit{Phase-error estimation protocol} conditional on $\Omega(\vec n, \vec m)$; and let
\begin{equation}
\label{eq:kappa_nm}
    \kappa(\vec n, \vec m) \coloneqq \Pr[\bm{n_{K,(1)}}< \mathcal{N}_{K,(1)} (\bm{\vec{n}})\quad \cup \quad \bm\ephone > \mathcal{E}_{\mathrm{ph},(1)} (\bm{\vec{n}})\mid \Omega(\vec n, \vec m)].
\end{equation}
Also, let $Z_A^{n_K}$ be the register in $\rho_{\vert \Omega(\vec n, \vec m)}$ containing Alice's sifted key in the \textit{Actual protocol (source replaced)}, and let $X_A^{n_K}$ ($X_B^{n_K}$) be the register in $\rho_{\vert \Omega(\vec n, \vec m)}^\mathrm{virt}$ containing Alice's (Bob's) bit outcomes for the rounds in $\mathcal{D}_{\mathrm{keep},\mathrm{sifted},Z}$ in the \textit{Phase-error estimation protocol}. Note that we can divide $Z_A^{n_K} = Z_A^{n_{K,(1)}} Z_A^{n_{K,(\mathrm{rest})}}$, $X_A^{n_K} = X_A^{n_{K,(1)}} X_A^{n_{K,(\mathrm{rest})}}$ and $X_B^{n_K} = X_B^{n_{K,(1)}} X_B^{n_{K,(\mathrm{rest})}}$ based on the photon number. Using this division and applying the EUR on the $n_{K,(1)}$ rounds conditional on $\Omega(\vec n, \vec m)$ and with smoothing parameter $\kappa(\vec n, \vec m)$ as in \cite[Theorem 3]{tupkaryPhaseError2025}, we obtain
\begin{equation}
    \begin{gathered}
        H^{\sqrt{\kappa(\vec n, \vec m)}}_\mathrm{min}\big(Z_A^{n_K} \mid W E\big)_{\rho_{\vert\Omega(\vec n, \vec m)}}
        \;\geq\;
        H^{\sqrt{\kappa(\vec n, \vec m)}}_\mathrm{min}\big(Z_A^{n_{K,(1)}} \mid W E\big)_{\rho_{\vert\Omega(\vec n, \vec m)}}
        \\[4pt]
        \geq\;
        n_{K,(1)}
        - H^{\sqrt{\kappa(\vec n, \vec m)}}_\mathrm{max}\big(X_A^{n_{K,(1)}} \mid X_B^{n_{K,(1)}}\big)_{\rho^\mathrm{virt}_{\vert\Omega(\vec n, \vec m)}}
        \;\geq\;
        n_{K,(1)}\big[1-h\big(\mathcal{E}_{\mathrm{ph},(1)} (\vec{n})\big)\big],
        \quad \forall \vec n, \vec m,
    \end{gathered}
\end{equation}
Now, let $S_{\vec n} \coloneqq \{\vec m : n_{K,(1)} \geq \mathcal{N}_{K,(1)}(\vec n)\}$ be the set of photon-number outcomes for which the number of single photons in $\mathcal{D}_{\mathrm{keep},\mathrm{sifted},Z}$ is at least $\mathcal{N}_{K,(1)}(\vec n)$. For $\vec m \in S_{\vec n}$, we have $n_{K,(1)} \geq \mathcal{N}_{K,(1)}(\vec n)$, and therefore
\begin{equation}
H^{\sqrt{\kappa(\vec n, \vec m)}} _\mathrm{min}(Z_A^{n_K} \mid W E)_{\rho_{\vert\Omega(\vec n, \vec m)}} \geq \mathcal{N}_{K,(1)}(\vec n)\big[1-h\big(\mathcal{E}_{\mathrm{ph},(1)} (\vec{n})\big)\big], \quad \forall \vec n, \forall \vec m \in S_{\vec n}.
\end{equation}
Furthermore, by definition of $\kappa(\vec n, \vec m)$ in \cref{eq:kappa_nm} and the assumed guarantee in \cref{eq:statistical_guarantee_decoy}, we have
\begin{equation}
\begin{aligned}
&\sum_{\vec n} \sum_{\vec m \in S_{\vec n}} \Pr(\Omega(\vec n, \vec m)) \kappa(\vec n, \vec m) + \sum_{\vec n} \sum_{\vec m \notin S_{\vec n}} \Pr(\Omega(\vec n, \vec m)) \\
&= \Pr[\bm{n_{K,(1)}} \geq \mathcal{N}_{K,(1)}(\bm{\vec n}) \text{ and } \bm\ephone > \mathcal{E}_{\mathrm{ph},(1)}(\bm{\vec n})] + \Pr[\bm{n_{K,(1)}} < \mathcal{N}_{K,(1)}(\bm{\vec n})] \\
&= \Pr[\bm{n_{K,(1)}} < \mathcal{N}_{K,(1)}(\bm{\vec n}) \cup \bm\ephone > \mathcal{E}_{\mathrm{ph},(1)}(\bm{\vec n})] \leq \varepsilon_\mathrm{PE}.
\end{aligned}
\end{equation}
Then, security follows directly from \cite[Theorem 4]{tupkaryPhaseError2025} after identifying $i \mapsto \vec n$, $j \mapsto \vec m$, $\beta_i \mapsto \mathcal{N}_{K,(1)}(\vec n)\big[1-h\big(\mathcal{E}_{\mathrm{ph},(1)} (\vec n)\big)\big]$, $\kappa_{(i,j)} \mapsto \kappa(\vec n, \vec m)$, $\varepsilon_\mathrm{AT}^2 \mapsto \varepsilon_\mathrm{PE}$, $\vec Z \mapsto Z_A^{n_K}$, $\vec C \mapsto W$ and $\vec E \mapsto E$. 
\end{proof}

\subsection{Proof of \texorpdfstring{\cref{lem:quantum_coin_inequality_decoy}}{Lemma 1}}    \label{appsub:proof_quantum_coin_inequality_decoy}
\begin{proof}
    The proof is fairly similar to that in Ref.~\cite[SE1.A]{curras-lorenzoSecurityFramework2025}. To prove this, we consider a scenario that yields identical statistics for the single-photon sifted rounds as the \textit{Phase-error estimation protocol}, but in which Alice and Bob do things in a different order. Namely:

\begin{enumerate}[label=\arabic*'., ref=\arabic*']
\item Alice and Bob determine the set of detected sifted rounds with $m_k=1$, $\mathcal{D}_\mathrm{sifted}^{(1)}$. 
    \item For each round  $k \in \mathcal{D}_\mathrm{sifted}^{(1)}$:
    \begin{enumerate}
        \item Alice and Bob perform the joint POVM $\{\hat m_\mathrm{err}, \mathbb{I}-\hat m_\mathrm{err}\}$, where $\hat m_\mathrm{err} = \ketbra{0_X}_A \otimes G_1^{(X)} + \ketbra{1_X}_A \otimes G_0^{(X)}$. Let $\bm e_k \in \{\mathrm{err},\overline{\mathrm{err}}\}$ be the outcome of this measurement.
    \end{enumerate}
    
    \item Alice chooses $\bm t_k = \mathrm{keep}$ with probability $p_\mathrm{keep}$ or $\bm t_k = \mathrm{trash}$ otherwise. Then, if $\bm t_k = \mathrm{keep}$, Alice measures system $C_k$ in the $\{\ket{Z},\ket{X}\} \equiv \{\ket{0},\ket{1}\}$ basis, obtaining an outcome $\bm c_k=\{Z,X\}$. Conversely, if $\bm t_k = \mathrm{trash}$, Alice measures system $C_k$ in the $\{\ket{+},\ket{-}\}$ basis, obtaining an outcome $\bm c_k=\{+,-\}$.

    \item Alice and Bob perform all their other measurements.
    \end{enumerate}

{The equivalence follows from three observations. First, all classical tags (keep/trash, sifted/unsifted) are generated independently of everything else, so we may assume they are drawn at the very beginning; likewise, the photon numbers $m_k$ come from a measurement on Alice's registers $M_1^N$ whose outcome is never announced, so it too may be placed at any point. This means that we can assume that $\mathcal{D}_{\mathrm{sifted}}^{(1)}$ is fixed at step~1'.

Second, in the \emph{Phase-error estimation protocol} the measurement applied to $A_kB_k$ on a detected keep-sifted round is the \emph{same} in both bases: Alice measures $A_k$ in $\{\ket{0_X},\ket{1_X}\}$ and Bob applies $\{G_0^{(X)},G_1^{(X)}\}$, both when $\alpha_k=\beta_k=X$ (Step~7) and when $\alpha_k=\beta_k=Z$ (Step~8). Grouping the outcomes into error/no-error, this is precisely the joint POVM $\{\hat m_{\mathrm{err}},\mathbb{I}-\hat m_{\mathrm{err}}\}$ on $A_kB_k$, and it does \emph{not} depend on $\alpha_k$. Since it acts on $A_kB_k$ while the coin measurement acts on the disjoint system $C_k$, the two commute, and the former may be applied \emph{before} the latter. (This is a feature of the virtual protocol: in the \emph{Actual protocol} Bob's POVM depends on $\alpha_k$, and no such reordering would be possible.)

Third, on a detected trash-sifted round the protocol applies no measurement to $A_kB_k$, whereas the reordered scenario applies $\{\hat m_{\mathrm{err}},\mathbb{I}-\hat m_{\mathrm{err}}\}$ there as well. This changes nothing, because on such a round the outcome $\bm e_k$ is never announced, never used in any of the counts of the lemma, and the systems $A_k$ and $B_k$ are never touched again. Measuring a system and discarding the outcome is a local operation on that system alone, so it leaves the state of all other systems---in particular $C_k$, the registers of the other rounds, and Eve's system $E$---unchanged. The joint distribution of the announced data, the sifted key and the coin outcomes is therefore the same in both descriptions, with the sole difference that $\bm e_k$ is now defined on trash-sifted rounds too, which is what {let} us condition on it below.

Combining the three observations, we may (i) first determine $\mathcal{D}_{\mathrm{sifted}}^{(1)}$, (ii) apply $\{\hat m_{\mathrm{err}},\mathbb{I}-\hat m_{\mathrm{err}}\}$ on all $k\in\mathcal{D}_{\mathrm{sifted}}^{(1)}$, and (iii) only afterwards measure the coin systems $C_k$ in the basis fixed by the keep/trash choice, without changing the joint distribution of the classical outcomes. This is exactly the reordered scenario above.}

In our proof, we condition on a particular realization of the set of detected sifted rounds $\mathcal{D}_{\mathrm{sifted}}^{(1)}$, and work in the conditional probability space given this realization. All probabilities below should be understood as conditional on this fixed $\mathcal{D}_{\mathrm{sifted}}^{(1)}$, though we suppress this notation for brevity. At the end of the proof, we will argue that since the derived bound holds for any realization of $\mathcal{D}_{\mathrm{sifted}}^{(1)}$, it also holds unconditionally.

Let $\bm v_k = (\bm e_k,\bm t_k, \bm{c}_k)$ be a random vector holding all the outcomes for round $k \in \mathcal{D}_{\mathrm{sifted}}^{(1)}$, and let $\bm v_1^k \coloneqq (\bm v_1,\ldots,\bm v_k)$ be the random history until and including round $k$. Consider the state of system $C_k$ for a given round $k\in\mathcal{D}_{\mathrm{sifted}}^{(1)}$ conditional on a particular outcome $\bm v_1^{k-1} = v_1^{k-1}$ for all measurements on the previous rounds, and on an outcome $\bm e_k = \mathrm{err}$ for this round. Since no operation has acted directly on system $C_k$, it remains a qubit state $\rho_k$ in the two-dimensional space spanned by $\{\ket{0}_{C_k}, \ket{1}_{C_k}\}$. The Bloch sphere constraint implies that the length of the Bloch vector is bounded by 1, which yields the following inequality for measurement probabilities in different bases \cite{loSecurityQuantum2007}:
\begin{equation}
\label{eq:bloch_sphere_general}
    1- 2 P_{-} \leq 2\sqrt{P_0 \cdot P_1},
\end{equation}
where $P_0$ and $P_1$ denote the probabilities of obtaining outcomes $\ket{0}$ and $\ket{1}$ when measuring $\rho_k$ in the computational basis $\{\ket{0}, \ket{1}\}$, and $P_{-}$ denotes the probability of obtaining outcome $\ket{-}$ when measuring $\rho_k$ in the Hadamard basis $\{\ket{+}, \ket{-}\}$. Applying this bound to our conditional state, and noting that Alice's choice $\bm t_k \in \{\mathrm{keep}, \mathrm{trash}\}$ only determines which measurement basis is applied to the qubit state $\rho_k$, we obtain
\begin{equation}
\label{eq:m-Bloch-average-ineq}
\begin{aligned}
    &1- 2\Pr[\bm c_k = - \vert  v_1^{k-1}, \bm t_k = \mathrm{trash},\bm e_k = \mathrm{err}] \\ 
    &\leq 2\sqrt{\Pr[\bm c_k = Z \vert  v_1^{k-1}, \bm t_k = \mathrm{keep},\bm e_k = \mathrm{err}]\Pr[\bm c_k = X \vert  v_1^{k-1}, \bm t_k = \mathrm{keep},\bm e_k = \mathrm{err}]}.
\end{aligned}
\end{equation}
    On the other hand, we have that
\begin{equation}
	\label{eq:PXcresult}
	\begin{aligned}
		&\Pr(\bm e_k = \mathrm{err}\vert  v_1^{k-1}) \Pr (\bm c_k=-|v_1^{k-1},\bm t_k = \mathrm{trash},\bm e_k = \mathrm{err}) \\
        &= \Pr(\bm e_k = \mathrm{err} \vert  v_1^{k-1}) \frac{\Pr(\bm c_k=-,\bm t_k = \mathrm{trash},\bm e_k = \mathrm{err}|  v_1^{k-1})}{\Pr(\bm t_k = \mathrm{trash},\bm e_k = \mathrm{err}|  v_1^{k-1})} \\
		&= \Pr(\bm e_k = \mathrm{err}\vert  v_1^{k-1}) \frac{\Pr(\bm c_k=-,\bm t_k = \mathrm{trash},\bm e_k = \mathrm{err}|  v_1^{k-1})}{\Pr(\bm t_k = \mathrm{trash}|\bm e_k = \mathrm{err},  v_1^{k-1})\Pr(\bm e_k = \mathrm{err}|  v_1^{k-1})} \\
		&= \Pr(\bm e_k = \mathrm{err}\vert  v_1^{k-1}) \frac{\Pr(\bm c_k=-,\bm t_k = \mathrm{trash},\bm e_k = \mathrm{err}|  v_1^{k-1})}{(1-p_\mathrm{keep})\Pr(\bm e_k = \mathrm{err}|  v_1^{k-1})} \\
		&= \frac{\Pr(\bm c_k=-,\bm t_k = \mathrm{trash},\bm e_k = \mathrm{err}| v_1^{k-1})}{1-p_\mathrm{keep}},
	\end{aligned}
\end{equation}
where in the third equality we have used the fact that, in a sifted round with an error outcome (i.e., a round with outcome $\bm e_k = \mathrm{err}$), the probability to choose trash mode is independent of obtaining this error outcome, and equal to $1-p_\mathrm{keep}$. In a similar way, we have that
\begin{equation}
	\label{eq:PZcresult}
	\begin{aligned}
		&\Pr(\bm e_k = \mathrm{err}\vert  v_1^{k-1})\Pr(\bm c_k=Z|\bm t_k = \mathrm{keep},\bm e_k = \mathrm{err},  v_1^{k-1}) 
		= \frac{\Pr(\bm c_k=Z,\bm t_k = \mathrm{keep},\bm e_k = \mathrm{err}| v_1^{k-1})}{p_\mathrm{keep}}, \\
		&\Pr(\bm e_k = \mathrm{err}\vert  v_1^{k-1})\Pr(\bm c_k=X|\bm t_k = \mathrm{keep},\bm e_k = \mathrm{err},  v_1^{k-1}) 
		= \frac{\Pr(\bm c_k=X,\bm t_k = \mathrm{keep},\bm e_k = \mathrm{err}| v_1^{k-1})}{p_\mathrm{keep}}.
	\end{aligned}
\end{equation}
Multiplying by $\Pr(\bm e_k = \mathrm{err}\vert  v_1^{k-1})$ on both sides of \cref{eq:m-Bloch-average-ineq}, and using \cref{eq:PXcresult,eq:PZcresult} we obtain
\begin{equation}
	\label{eq:1st-ineq-bloch}
	\begin{aligned}
	&\Pr(\bm e_k = \mathrm{err}\vert  v_1^{k-1})-2\frac{\Pr(\bm c_k=-,\bm t_k = \mathrm{trash},\bm e_k = \mathrm{err}| v_1^{k-1})}{1-p_\mathrm{keep}}\\
	&\le \frac{2}{p_\mathrm{keep}}\sqrt{\Pr(\bm c_k=Z,\bm t_k = \mathrm{keep},\bm e_k = \mathrm{err}| v_1^{k-1})\Pr(\bm c_k=X,\bm t_k = \mathrm{keep},\bm e_k = \mathrm{err}| v_1^{k-1})}\,.\\
\end{aligned}
\end{equation}
Replacing $\bm e_k = \mathrm{err}$ by $\bm e_k = \overline{\mathrm{err}}$ in {Eqs.~(\ref{eq:m-Bloch-average-ineq}) to (\ref{eq:1st-ineq-bloch}}), we obtain
\begin{equation}
	\label{eq:2st-ineq-bloch}
	\begin{aligned}
		&\Pr(\bm e_k = \overline{\mathrm{err}}\vert  v_1^{k-1})-2\frac{\Pr(\bm c_k=-,\bm t_k = \mathrm{trash},\bm e_k = \overline{\mathrm{err}}| v_1^{k-1})}{1-p_\mathrm{keep}}\\
		&\le \frac{2}{p_\mathrm{keep}}\sqrt{\Pr(\bm c_k=Z,\bm t_k = \mathrm{keep},\bm e_k = \overline{\mathrm{err}}| v_1^{k-1})\Pr(\bm c_k=X,\bm t_k = \mathrm{keep},\bm e_k = \overline{\mathrm{err}}| v_1^{k-1})}\,\\
		&= {\frac{2}{p_\mathrm{keep}}}\sqrt{\Pr(\bm c_k=Z,\bm t_k = \mathrm{keep}| v_1^{k-1})-\Pr(\bm c_k=Z,\bm t_k = \mathrm{keep},\bm e_k = \mathrm{err}| v_1^{k-1})}\\
		&\times\sqrt{\Pr(\bm c_k=X,\bm t_k = \mathrm{keep}| v_1^{k-1})-\Pr(\bm c_k=X,\bm t_k = \mathrm{keep},\bm e_k = \mathrm{err}| v_1^{k-1})}\,.
	\end{aligned}
\end{equation}
Taking the sum of \cref{eq:1st-ineq-bloch,eq:2st-ineq-bloch}, we obtain
\begin{align}
	&1-2\frac{\Pr(\bm c_k=-,\bm t_k = \mathrm{trash}| v_1^{k-1})}{1-p_\mathrm{keep}}\nonumber\\
	&\le \frac{2}{p_\mathrm{keep}}\Big(\sqrt{\Pr(\bm c_k=Z,\bm t_k = \mathrm{keep},\bm e_k = \mathrm{err}| v_1^{k-1})\Pr(\bm c_k=X,\bm t_k = \mathrm{keep},\bm e_k = \mathrm{err}| v_1^{k-1})}\nonumber\\
	&+\sqrt{\Pr(\bm c_k=Z,\bm t_k = \mathrm{keep}| v_1^{k-1})-\Pr(\bm c_k=Z,\bm t_k = \mathrm{keep},\bm e_k = \mathrm{err}| v_1^{k-1})}\nonumber\\
	&\times\sqrt{\Pr(\bm c_k=X,\bm t_k = \mathrm{keep}| v_1^{k-1})-\Pr(\bm c_k=X,\bm t_k = \mathrm{keep},\bm e_k = \mathrm{err}| v_1^{k-1})}\Big)\,. \label{eq:expresion_cond_detection1}
\end{align}

Also, we have that,
\begin{equation}
	1 = \frac{p_\mathrm{keep}}{p_\mathrm{keep}} = \frac{\Pr(\bm t_k = \mathrm{keep}\vert v_1^{k-1})}{p_\mathrm{keep}}{.}\label{eq:Zc_trick}
\end{equation}
Substituting \cref{eq:Zc_trick} into \cref{eq:expresion_cond_detection1}, taking the sum over all rounds $k \in \mathcal{D}_{\mathrm{sifted}}^{(1)}$ on both sides of the inequality, and using the fact that $\sqrt{xy}$ is a concave function of both $x$ and $y$ and thus by Cauchy-Schwarz $\sum_k \sqrt{x_k y_k}\le \sqrt{\sum_k x_k \sum_k y_k}$, we obtain
\begin{equation}
	\begin{aligned}
	&\frac{\sum_k \Pr(\bm t_k = \mathrm{keep}\vert  v_1^{k-1})}{p_\mathrm{keep}}-2\frac{\sum_k \Pr(\bm c_k=-,\bm t_k = \mathrm{trash}| v_1^{k-1})}{1-p_\mathrm{keep}}\\
	&\le \frac{2}{p_\mathrm{keep}}\Bigg(\sqrt{\sum_k \Pr(\bm c_k=Z,\bm t_k = \mathrm{keep},\bm e_k = \mathrm{err}| v_1^{k-1})\sum_k \Pr(\bm c_k=X,\bm t_k = \mathrm{keep},\bm e_k = \mathrm{err}| v_1^{k-1})}\\
	&+\sqrt{\sum_k \Pr(\bm c_k=Z,\bm t_k = \mathrm{keep}| v_1^{k-1})-\sum_k \Pr(\bm c_k=Z,\bm t_k = \mathrm{keep},\bm e_k = \mathrm{err}| v_1^{k-1})}\\
	&\times\sqrt{\sum_k \Pr(\bm c_k=X,\bm t_k = \mathrm{keep}| v_1^{k-1})-\sum_k \Pr(\bm c_k=X,\bm t_k = \mathrm{keep},\bm e_k = \mathrm{err}| v_1^{k-1})}\Bigg)\,. \label{eq:expresion_cond_detection2}
	\end{aligned}
\end{equation}
Now, note that the inequality in \cref{eq:expresion_cond_detection2} holds for any particular realization $v_1^{k-1}$ of the measurement history up to round $k-1$. Since the inequality holds for all possible histories, it also holds when we consider the random variable $\bm v_1^{k-1}$. That is, we can rewrite \cref{eq:expresion_cond_detection2} by replacing the conditioning on fixed outcomes $v_1^{k-1}$ with conditioning on the random history $\bm v_1^{k-1}$. By doing so, we obtain
\begin{equation}
\begin{aligned}
	&\frac{\bm{\tilde{n}}_{\mathrm{keep},\mathrm{sifted},(1)}^{\mathrm{det}}}{p_\mathrm{keep}} - \frac{2 \bm{\tilde{n}}_{\mathrm{trash},\mathrm{sifted},-,(1)}^{\mathrm{det}}}{1-p_\mathrm{keep}} \\ &\leq \frac{2}{p_\mathrm{keep}} \bigg[\sqrt{\bm{\tilde{n}}_{\mathrm{keep},\mathrm{sifted},Z,(1)}^{\mathrm{err}} \bm{\tilde{n}}_{\mathrm{keep},\mathrm{sifted},X,(1)}^{\mathrm{err}}} + \sqrt{\big(\bm{\tilde{n}}_{\mathrm{keep},\mathrm{sifted},Z,(1)}^{\mathrm{det}}-\bm{\tilde{n}}_{\mathrm{keep},\mathrm{sifted},Z,(1)}^{\mathrm{err}}\big) \big(\bm{\tilde{n}}_{\mathrm{keep},\mathrm{sifted},X,(1)}^{\mathrm{det}}-\bm{\tilde{n}}_{\mathrm{keep},\mathrm{sifted},X,(1)}^{\mathrm{err}}\big)} \bigg],
 \label{eq:N_tilde}
 \end{aligned}
\end{equation}
where we have defined $\bm{\tilde{n}}_{\mathrm{keep},\mathrm{sifted},(1)}^{\mathrm{det}}:=\sum_k \Pr(\bm t_k = \mathrm{keep}\vert  \bm v_1^{k-1})$, $\bm{\tilde{n}}_{\mathrm{trash},\mathrm{sifted},-,(1)}^{\mathrm{det}} := \sum_k \Pr(\bm c_k=-,\bm t_k = \mathrm{trash}| \bm v_1^{k-1})$, $\bm{\tilde{n}}_{\mathrm{keep},\mathrm{sifted},b,(1)}^{\mathrm{err}}:=\sum_k \Pr(\bm c_k=b,\bm t_k = \mathrm{keep},\bm e_k = \mathrm{err}| \bm v_1^{k-1})$ and $\bm{\tilde{n}}_{\mathrm{keep},\mathrm{sifted},b,(1)}^{\mathrm{det}}:=\sum_k \Pr(\bm c_k=b,\bm t_k = \mathrm{keep}| \bm v_1^{k-1})$ for $b \in \{Z,X\}$, where the conditioning is now on the random history $\bm v_1^{k-1}$ for each round $k$. Note that these $\bm{\tilde{n}}$ quantities are random variables representing sums of conditional expectations. Then, {dividing by $\bm{\tilde{n}}_{\mathrm{keep},\mathrm{sifted},(1)}^{\mathrm{det}}$ on both sides and} using the fact that $\bm{\tilde{n}}_{\mathrm{keep},\mathrm{sifted},(1)}^{\mathrm{det}} = 
\bm{\tilde{n}}_{\mathrm{keep},\mathrm{sifted},Z,(1)}^{\mathrm{det}} + \bm{\tilde{n}}_{\mathrm{keep},\mathrm{sifted},X,(1)}^{\mathrm{det}} \geq 2\sqrt{ \bm{\tilde{n}}_{\mathrm{keep},\mathrm{sifted},Z,(1)}^{\mathrm{det}} \bm{\tilde{n}}_{\mathrm{keep},\mathrm{sifted},X,(1)}^{\mathrm{det}}}$, we obtain
\begin{equation}
	\label{eq:finite_formula_2}
	\begin{aligned}
		1 - \frac{2 p_\mathrm{keep} \bm{\tilde{n}}_{\mathrm{trash},\mathrm{sifted},-,(1)}^{\mathrm{det}}}{(1-p_\mathrm{keep}) \bm{\tilde{n}}_{\mathrm{keep},\mathrm{sifted},(1)}^{\mathrm{det}}} &\leq \frac{1}{\sqrt{ \bm{\tilde{n}}_{\mathrm{keep},\mathrm{sifted},Z,(1)}^{\mathrm{det}} \bm{\tilde{n}}_{\mathrm{keep},\mathrm{sifted},X,(1)}^{\mathrm{det}}}} \bigg[\sqrt{\bm{\tilde{n}}_{\mathrm{keep},\mathrm{sifted},Z,(1)}^{\mathrm{err}} \bm{\tilde{n}}_{\mathrm{keep},\mathrm{sifted},X,(1)}^{\mathrm{err}}} \\
        &+ \sqrt{\big(\bm{\tilde{n}}_{\mathrm{keep},\mathrm{sifted},Z,(1)}^{\mathrm{det}}-\bm{\tilde{n}}_{\mathrm{keep},\mathrm{sifted},Z,(1)}^{\mathrm{err}}\big) \big(\bm{\tilde{n}}_{\mathrm{keep},\mathrm{sifted},X,(1)}^{\mathrm{det}}-\bm{\tilde{n}}_{\mathrm{keep},\mathrm{sifted},X,(1)}^{\mathrm{err}}\big)} \bigg] \\
		&= \sqrt{\frac{\bm{\tilde{n}}_{\mathrm{keep},\mathrm{sifted},Z,(1)}^{\mathrm{err}}}{\bm{\tilde{n}}_{\mathrm{keep},\mathrm{sifted},Z,(1)}^{\mathrm{det}}} \frac{\bm{\tilde{n}}_{\mathrm{keep},\mathrm{sifted},X,(1)}^{\mathrm{err}}}{\bm{\tilde{n}}_{\mathrm{keep},\mathrm{sifted},X,(1)}^{\mathrm{det}}}} + \sqrt{\bigg(1-\frac{\bm{\tilde{n}}_{\mathrm{keep},\mathrm{sifted},Z,(1)}^{\mathrm{err}}}{\bm{\tilde{n}}_{\mathrm{keep},\mathrm{sifted},Z,(1)}^{\mathrm{det}}}\bigg) \bigg(1-\frac{\bm{\tilde{n}}_{\mathrm{keep},\mathrm{sifted},X,(1)}^{\mathrm{err}}}{\bm{\tilde{n}}_{\mathrm{keep},\mathrm{sifted},X,(1)}^{\mathrm{det}}}\bigg)}.
	\end{aligned}
\end{equation}
The inequality $z \leq \sqrt{y' y} + \sqrt{1-y'} \sqrt{1-y}$, with $z,y,y' \in [0,1]$, is equivalent to the inequality $G_{-} (y,z) \leq y' \leq G_{+}(y,z)$, where \cite{pereiraQuantumKey2020}
\begin{equation}
	G_-(y,z) =
	\begin{cases}
		g_-(y,z)  & \quad \text{if } y > 1 -z^{2} \\
		0  & \quad \text{otherwise}
	\end{cases}
	\quad \quad
	\textrm{and}
	\quad \quad
	G_+(y,z) =
	\begin{cases}
		g_+(y,z)  & \quad \text{if } y < z^{2} \\
		1  & \quad \text{otherwise}{,}
	\end{cases}
\end{equation}
and
\begin{equation}
\label{eq:g}
	g_{\pm} (y,z) = y + (1-z^2)(1-2y) \pm 2\sqrt{z^2(1-z^2)y(1-y)}.
\end{equation}
Therefore, \cref{eq:finite_formula_2} implies that
\begin{equation}
\label{eq:finite_formula_3}
	\frac{\bm{\tilde{n}}_{\mathrm{keep},\mathrm{sifted},Z,(1)}^{\mathrm{err}}}{\bm{\tilde{n}}_{\mathrm{keep},\mathrm{sifted},Z,(1)}^{\mathrm{det}}} \leq G_+ \bigg(\frac{\bm{\tilde{n}}_{\mathrm{keep},\mathrm{sifted},X,(1)}^{\mathrm{err}}}{\bm{\tilde{n}}_{\mathrm{keep},\mathrm{sifted},X,(1)}^{\mathrm{det}}},	1 - \frac{2 p_\mathrm{keep} \bm{\tilde{n}}_{\mathrm{trash},\mathrm{sifted},-,(1)}^{\mathrm{det}}}{(1-p_\mathrm{keep}) (\bm{\tilde{n}}_{\mathrm{keep},\mathrm{sifted},Z,(1)}^{\mathrm{det}} + \bm{\tilde{n}}_{\mathrm{keep},\mathrm{sifted},X,(1)}^{\mathrm{det}})}\bigg),
\end{equation}
where we have also substituted $\bm{\tilde{n}}_{\mathrm{keep},\mathrm{sifted},(1)}^{\mathrm{det}} = \bm{\tilde{n}}_{\mathrm{keep},\mathrm{sifted},Z,(1)}^{\mathrm{det}} + \bm{\tilde{n}}_{\mathrm{keep},\mathrm{sifted},X,(1)}^{\mathrm{det}}$. We can rewrite this as
\begin{equation}
\label{eq:finite_formula_4}
	\bm{\tilde{n}}_{\mathrm{keep},\mathrm{sifted},Z,(1)}^{\mathrm{err}} \leq \bm{\tilde{n}}_{\mathrm{keep},\mathrm{sifted},Z,(1)}^{\mathrm{det}} G_+ \bigg(\frac{\bm{\tilde{n}}_{\mathrm{keep},\mathrm{sifted},X,(1)}^{\mathrm{err}}}{\bm{\tilde{n}}_{\mathrm{keep},\mathrm{sifted},X,(1)}^{\mathrm{det}}},	1 - \frac{2 p_\mathrm{keep} \bm{\tilde{n}}_{\mathrm{trash},\mathrm{sifted},-,(1)}^{\mathrm{det}}}{(1-p_\mathrm{keep}) (\bm{\tilde{n}}_{\mathrm{keep},\mathrm{sifted},Z,(1)}^{\mathrm{det}} + \bm{\tilde{n}}_{\mathrm{keep},\mathrm{sifted},X,(1)}^{\mathrm{det}})}\bigg).
\end{equation}

Now, note that the difference $\bm{n}_{\mathrm{keep},\mathrm{sifted},b,(1)}^{\mathrm{det}} - \bm{\tilde{n}}_{\mathrm{keep},\mathrm{sifted},b,(1)}^{\mathrm{det}}$ forms a martingale difference sequence with bounded increments in the conditional probability space, where $\bm{n}_{\mathrm{keep},\mathrm{sifted},b,(1)}^{\mathrm{det}}$ is the actual (random) count of detected keep sifted rounds where both users employed {the} basis $b$. The same is true for the other quantities. Thus, since the number of elements in the set of detected sifted single-photon rounds $\vert\mathcal{D}_{\mathrm{sifted}}^{(1)}\vert$ has been fixed, we can apply Azuma's inequality\footnote{{As discussed in \cref{rem:azuma_vs_kato}, a tighter martingale inequality known as Kato's inequality~\cite{katoConcentrationInequality2020,curras-lorenzoTightFinitekey2021} {(see also recent improvement in Ref.~\cite{mannalathQuantumKey2026})} could be applied to replace Azuma's inequality at this step to improve the finite-key performance.}} to these martingales to obtain the following relations
\begin{equation}
\label{eq:azuma_relations}
\begin{gathered}
\bm{\tilde{n}}_{\mathrm{keep},\mathrm{sifted},Z,(1)}^{\mathrm{err}} \underset{\epsilon_A}{\geq} \bm{n}_{\mathrm{keep},\mathrm{sifted},Z,(1)}^{\mathrm{err}} - \Delta'_A, \quad \bm{\tilde{n}}_{\mathrm{keep},\mathrm{sifted},Z,(1)}^{\mathrm{det}} \underset{\epsilon_A}{\leq} \bm{n}_{\mathrm{keep},\mathrm{sifted},Z,(1)}^{\mathrm{det}} + \Delta'_A, \\
\bm{\tilde{n}}_{\mathrm{keep},\mathrm{sifted},X,(1)}^{\mathrm{err}} \underset{\epsilon_A}{\leq} \bm{n}_{\mathrm{keep},\mathrm{sifted},X,(1)}^{\mathrm{err}} + \Delta'_A, \quad \bm{\tilde{n}}_{\mathrm{keep},\mathrm{sifted},X,(1)}^{\mathrm{det}} \underset{\epsilon_A}{\geq} \bm{n}_{\mathrm{keep},\mathrm{sifted},X,(1)}^{\mathrm{det}} - \Delta'_A, \\
\bm{\tilde{n}}_{\mathrm{trash},\mathrm{sifted},-,(1)}^{\mathrm{det}}  \underset{\epsilon_A}{\leq} \bm{n}_{\mathrm{trash},\mathrm{sifted},-,(1)}^{\mathrm{det}} + \Delta'_A,
\end{gathered}
\end{equation}
where $\Delta'_A := \Delta'_A (n_\mathrm{sifted,(1)}^{\mathrm{det}}) = \sqrt{2\bm{n}_\mathrm{sifted,(1)}^{\mathrm{det}} \ln 1/\epsilon_A}$ with $n_\mathrm{sifted,(1)}^{\mathrm{det}} = \vert \mathcal{D}_\mathrm{sifted}^{(1)} \vert$ being the total number of detected sifted rounds, which is fixed in the conditional probability space; $\bm{n}_{\mathrm{trash},\mathrm{sifted},-,(1)}^{\mathrm{det}}$ is the total number of detected sifted rounds in which Alice chose trash mode and obtained outcome $-$; $\bm{n}_{\mathrm{keep},\mathrm{sifted},(1)}^{\mathrm{det}} := \bm{n}_{\mathrm{keep},\mathrm{sifted},Z,(1)}^{\mathrm{det}} + \bm{n}_{\mathrm{keep},\mathrm{sifted},X,(1)}^{\mathrm{det}}$ is the total number of detected keep sifted rounds; and $\bm{n}_{\mathrm{keep},\mathrm{sifted},Z,(1)}^{\mathrm{err}}$, $\bm{n}_{\mathrm{keep},\mathrm{sifted},X,(1)}^{\mathrm{err}}$, $\bm{n}_{\mathrm{keep},\mathrm{sifted},Z,(1)}^{\mathrm{det}}$ and $\bm{n}_{\mathrm{keep},\mathrm{sifted},X,(1)}^{\mathrm{det}}$ are defined above.

Then, using the fact that the RHS of \cref{eq:finite_formula_4} is monotonically increasing with respect to $\bm{\tilde{n}}_{\mathrm{keep},\mathrm{sifted},Z,(1)}^{\mathrm{det}}$ (see proof in Ref.~\cite[Supplementary Note D]{curras-lorenzoSecurityQuantum2026}), and the fact that $G_+(y,z) \leq G_+(\overline{y},\underline{z})$ with $\overline{y} \geq y$ and $\underline{z} \leq z$, and applying the relations in \cref{eq:azuma_relations}, we obtain
\begin{equation}
\begin{aligned}
	&\bm{n}_{\mathrm{keep},\mathrm{sifted},Z,(1)}^{\mathrm{err}} - \Delta'_A \underset{5\epsilon_A}{\leq} (\bm{n}_{\mathrm{keep},\mathrm{sifted},Z,(1)}^{\mathrm{det}} + \Delta'_A) \, \times \\
    &G_+ \bigg(\frac{\bm{n}_{\mathrm{keep},\mathrm{sifted},X,(1)}^{\mathrm{err}} + \Delta'_A}{\bm{n}_{\mathrm{keep},\mathrm{sifted},X,(1)}^{\mathrm{det}} - \Delta'_A},	1 - \frac{2 p_\mathrm{keep} (\bm{n}_{\mathrm{trash},\mathrm{sifted},-,(1)}^{\mathrm{det}} + \Delta'_A)}{(1-p_\mathrm{keep}) (\bm{n}_{\mathrm{keep},\mathrm{sifted},Z,(1)}^{\mathrm{det}}+\bm{n}_{\mathrm{keep},\mathrm{sifted},X,(1)}^{\mathrm{det}})}\bigg).
\label{eq:finite_azuma_p}
\end{aligned}
\end{equation}
Note that when applying {Azuma's} bounds to the sum $\bm{\tilde{n}}_{\mathrm{keep},\mathrm{sifted},Z,(1)}^{\mathrm{det}} + \bm{\tilde{n}}_{\mathrm{keep},\mathrm{sifted},X,(1)}^{\mathrm{det}}$ appearing in the denominator of the second argument of $G_+$, the $\Delta'_A$ terms cancel: $(\bm{n}_{\mathrm{keep},\mathrm{sifted},Z,(1)}^{\mathrm{det}} + \Delta'_A) + (\bm{n}_{\mathrm{keep},\mathrm{sifted},X,(1)}^{\mathrm{det}} - \Delta'_A) = \bm{n}_{\mathrm{keep},\mathrm{sifted},Z,(1)}^{\mathrm{det}} + \bm{n}_{\mathrm{keep},\mathrm{sifted},X,(1)}^{\mathrm{det}}$.

Since the bound in \cref{eq:finite_azuma_p} holds with probability at least $1-5\epsilon_A$ in the conditional probability space for any particular realization of $\mathcal{D}^{(1)}_\mathrm{sifted}$, by the law of total probability, it also holds with probability at least $1-5\epsilon_A$ in the unconditional probability space where $\mathcal{D}^{(1)}_\mathrm{sifted}$ itself is random. Importantly, since $\bm{n}_\mathrm{sifted,(1)}^{\mathrm{det}} = |\mathcal{D}^{(1)}_\mathrm{sifted}|$ is random, the {Azuma's} deviation term $\Delta'_A = \sqrt{2\bm{n}_\mathrm{sifted,(1)}^{\mathrm{det}} \ln 1/\epsilon_A}$ is now also a random variable that depends on the number of detected sifted rounds.

Finally, we note that $\bm{n}_{\mathrm{trash},\mathrm{sifted},-,(1)}^{\mathrm{det}} \leq \bm{n}_{\mathrm{trash},\mathrm{sifted},-,(1)}$, where $\bm{n}_{\mathrm{trash},\mathrm{sifted},-,(1)}$ is the total number of sifted rounds in which Alice chose trash mode and obtained outcome $-$, irrespective of whether they are detected or not. Also,  $\bm{n}_\mathrm{sifted,(1)}^{\mathrm{det}} \leq \bm{n}_\mathrm{sifted}^{\mathrm{det}}$. Using both of these bounds and the fact that $G_+(y,z)$ is monotonically decreasing in $z$, we obtain
\begin{equation}
\begin{aligned}
	&\bm{n}_{\mathrm{keep},\mathrm{sifted},Z,(1)}^{\mathrm{err}} - \Delta_A \underset{5\epsilon_A}{\leq} (\bm{n}_{\mathrm{keep},\mathrm{sifted},Z,(1)}^{\mathrm{det}} + \Delta_A) \, \times\\ 
    &G_+ \bigg(\frac{\bm{n}_{\mathrm{keep},\mathrm{sifted},X,(1)}^{\mathrm{err}} + \Delta_A}{\bm{n}_{\mathrm{keep},\mathrm{sifted},X,(1)}^{\mathrm{det}} - \Delta_A},	1 - \frac{2 p_\mathrm{keep} (\bm{n}_{\mathrm{trash},\mathrm{sifted},-,(1)} + \Delta_A)}{(1-p_\mathrm{keep}) (\bm{n}_{\mathrm{keep},\mathrm{sifted},Z,(1)}^{\mathrm{det}}+\bm{n}_{\mathrm{keep},\mathrm{sifted},X,(1)}^{\mathrm{det}})}\bigg),
\label{eq:finite_azuma_final}
\end{aligned}
\end{equation}
where $\Delta_A = \sqrt{2\bm{n}_\mathrm{sifted}^{\mathrm{det}} \ln 1/\epsilon_A}$, as we wanted to prove. We note that this bound is meaningful only when the arguments of $G_+$ are well-defined and in {the interval} $[0,1]$, in other cases, one could simply trivially replace it by the deterministic bound $\bm{n}_{\mathrm{keep},\mathrm{sifted},Z,(1)}^{\mathrm{err}} \leq \bm{n}_{\mathrm{keep},\mathrm{sifted},Z,(1)}^{\mathrm{det}}$.

\end{proof}

\subsection{Proof of \texorpdfstring{\cref{lem:bernstein_by_groups_decoy}}{Lemma 2}}    \label{appsub:proof_bernstein_by_groups_decoy}

\begin{proof}
For simplicity, we assume that the total number of rounds $N$ is a multiple of $(l_c+1)$, where $l_c$ is the source correlation length. If this does not hold, a very similar result can be derived by making trivial modifications.  Partition the rounds into $(l_c+1)$ sets according to $I_w = \{k : k \equiv w\mod{l_c+1}\}$, and express
\begin{equation}
\bm{n}_{\mathrm{trash},\mathrm{sifted},-,(1)} = \sum_{w=0}^{l_c} \bm{n}_{\mathrm{trash},\mathrm{sifted},-,(1)}^{(I_w)},
\end{equation}
where
\begin{equation}
\label{eq:n_trash_Iw_decoy}
\bm{n}_{\mathrm{trash},\mathrm{sifted},-,(1)}^{(I_w)} = \sum_{k \in I_w} \bm{X_k},
\end{equation}
and $\bm{X_k}$ is defined in \cref{eq:X_k_def}. Our objective is to show that
    \begin{equation}
        \label{eq:bound_n_trash_Iw_decoy}\Pr[\bm{n}_{\mathrm{trash},\mathrm{sifted},-,(1)}^{(I_w)} > n^\mathrm{U}_{\mathrm{trash},\mathrm{sifted},-,(1)}/(l_c+1)] \leq \epsilon_C,
    \end{equation}
    where $n^\mathrm{U}_{\mathrm{trash},\mathrm{sifted},-,(1)}$ is defined in \cref{eq:n_trash_U_decoy}. 
Then, we directly obtain \cref{eq:bound_n_trash_decoy} via the union bound.

 Let us now prove \cref{eq:bound_n_trash_Iw_decoy} for a particular $w \in \{0,1,\ldots,l_c\}$. For this, we define the following scenario:

    \vspace{3pt}

    \begin{mdframed}

    \noindent\textbf{$w$-th Modified Coin Measuring Scenario}

    \begin{enumerate}
        \item Alice prepares her global quantum coin state $\ket{\Psi_{N}}_{C_1^N A_1^N I_1^N M_1^N T_1^N}$.
        \item For each round $k \not\in I_w$, Alice measures the coin system $C_k$ in the $\{\ket{Z},\ket{X}\}$ basis to obtain $\alpha_k$, and then measures system $A_k$ in the $\{\ket{0_{\alpha_k}},\ket{1_{\alpha_k}}\}$ basis to obtain $a_k$.

        \item For each round $k \in I_w$, Alice measures the photon number system $M_k$ to learn the photon number $m$. If $m = 1$, Alice assigns the round to "trash" with probability $(1-p_\mathrm{keep})$, and to "sifted" with probability $1/2$. Then, if the round is assigned to both "trash" and "sifted", Alice measures the coin system $C_k$ in the $\{\ket{+}, \ket{-}\}$ basis.
    \end{enumerate}

    \end{mdframed}

    \vspace{5pt}

    Note that the measurements on the rounds $k \notin I_w$ act on different subsystems than the measurements on the rounds $k \in I_w$. Because of this, these measurements commute, and the \textit{a priori} distribution of $\bm{n}_{\mathrm{trash},\mathrm{sifted},-,(1)}^{(I_w)}$ must be the same in the Simple Coin Measuring Scenario {and} in the $w$-th Modified Coin Measuring Scenario. Thus, if we prove that \cref{eq:bound_n_trash_Iw_decoy} holds in the latter scenario, it must also hold in the former. From this point on, we consider the $w$-th Modified Coin Measuring Scenario.

Let $a_{I_{\bar w}}$ and $\alpha_{I_{\bar w}}$ denote the bit and basis outcomes of all the measurements on the rounds not in $I_w$. Once we condition on a particular outcome $\bm{a_{I_{\bar w}}} = a_{I_{\bar w}}$ and $\bm{\alpha_{I_{\bar w}}} = \alpha_{I_{\bar w}}$, the postmeasurement global state (after tracing out the systems already measured) factorizes as
    \begin{equation}
    \label{eq:post_meas_decoy}    \ket{\Psi_{a_{I_{\bar w}}, \alpha_{I_{\bar w}}}}_{C_{I_w} A_{I_w} I_1^N M_1^N T_1^N} = \bigotimes_{k \in I_w} \ket{\hat \Psi_{\mathrm{coin}|a_{I_{\bar w}}, \alpha_{I_{\bar w}}}^{(k)}}_{C_k A_k I_k^{k+l_c} M_k^{k+l_c} T_k^{k+l_c}},
    \end{equation}
    where we have used that, since consecutive elements of $I_w$ are separated by exactly $l_c+1$ rounds, the blocks $\{k,k+1,\ldots,k+l_c\}$ with $k \in I_w$ partition $\{1,\ldots,N\}$, so that every register $I_j M_j T_j$ belongs to exactly one factor. Here, for $k \in I_w$,
    \begin{equation}
    \label{eq:non_recursive_def_decoy}
    \begin{gathered}
        \ket{\hat \Psi_{\mathrm{coin}|a_{I_{\bar w}}, \alpha_{I_{\bar w}}}^{(k)}}_{C_k A_k I_k^{k+l_c} M_k^{k+l_c} T_k^{k+l_c}} = \frac{1}{\sqrt 2} \sum_{\alpha_k \in \{Z,X\}} \ket{\alpha_k}_{C_k} \ket{\hat \Psi_{\alpha_k|a_{I_{\bar w}}, \alpha_{I_{\bar w}}}^{(k)}}_{A_k I_k^{k+l_c} M_k^{k+l_c} T_k^{k+l_c}},\\[4pt] \ket{\hat \Psi_{\alpha_k|a_{I_{\bar w}}, \alpha_{I_{\bar w}}}^{(k)}}_{A_k I_k^{k+l_c} M_k^{k+l_c} T_k^{k+l_c}} = \frac{1}{\sqrt 2} \sum_{a_k \in \{0,1\}} \ket{(a_k)_{\alpha_k}}_{A_k} \ket*{\Psi_{a_{k-l_c}^k,\alpha_{k-l_c}^k}}_{I_k M_k T_k} \otimes \ket{\vec{\Psi}_{a_{\bar k},\alpha_{\bar k},a_k,\alpha_k}}_{I_{k+1}^{k+l_c} M_{k+1}^{k+l_c} T_{k+1}^{k+l_c}},
    \end{gathered}
    \end{equation}
    where $\ket*{\Psi_{a_{k-l_c}^k,\alpha_{k-l_c}^k}}_{I_k M_k T_k}$ is given by \cref{eq:round_state_decoy} and
    \begin{equation}
    \label{eq:vecPsi_decoy}
        \ket{\vec{\Psi}_{a_{\bar k},\alpha_{\bar k},a_k,\alpha_k}}_{I_{k+1}^{k+l_c} M_{k+1}^{k+l_c} T_{k+1}^{k+l_c}} = \bigotimes_{j=k+1}^{k+l_c} \ket*{\Psi_{a_{j-l_c}^j,\alpha_{j-l_c}^j}}_{I_j M_j T_j}.
    \end{equation}

    Thus, crucially, the measurement events in $I_w$ are now independent Bernoulli trials. Formally, under the conditional probability measure
    \begin{equation}
    \label{eq:cond_prob_meas_decoy}
        \Pr[\,\cdot \mid a_{I_{\bar w}}, \alpha_{I_{\bar w}}],
    \end{equation}
    $\{\bm{X_k}\}_{k \in I_w}$ is a sequence of independent Bernoulli trials, since {it corresponds} to the outcome of measurements on systems $M_k$ and $C_k$ for each $k \in I_w$ in \cref{eq:post_meas_decoy}, and the state of these systems factorizes.

    For each $k \in I_w$, we have $E[\bm{X_k} \mid a_{I_{\bar w}}, \alpha_{I_{\bar w}}] = p_1 \cdot \frac{1}{2}(1-p_\mathrm{keep}) \cdot p_{C_k = - \vert m_k=1,\mathrm{trash},\mathrm{sifted}}^{(k)}$, where $p_1 = \sum_\mu p_\mu p_{1|\mu}$ is the probability of obtaining $m_k=1$ when measuring system $M_k$. {To evaluate the last factor, note that projecting $M_k$ onto $\ket{1}$ in \cref{eq:non_recursive_def_decoy} leaves, on register $I_k$, the normalized state $\sum_\mu \sqrt{p_\mu p_{1|\mu}/p_1}\,\ket{\hat\mu}_{I_k}$, which is independent of $(a_k,\alpha_k)$ because the intensity modulation is assumed ideal. This factor therefore decouples and contributes an overlap of unity, and the normalized postmeasurement state is
    \begin{equation}
    \label{eq:postselected_state_decoy}
    \begin{gathered}
        \ket{\hat \Psi_{\mathrm{coin}|a_{I_{\bar w}}, \alpha_{I_{\bar w}},m_k=1}^{(k)}}_{C_k A_k T_k I_{k+1}^{k+l_c} M_{k+1}^{k+l_c} T_{k+1}^{k+l_c}} = \frac{1}{\sqrt 2} \sum_{\alpha_k \in \{Z,X\}} \ket{\alpha_k}_{C_k} \ket{\hat \Psi_{\alpha_k|a_{I_{\bar w}}, \alpha_{I_{\bar w}},m_k=1}^{(k)}}_{A_k T_k I_{k+1}^{k+l_c} M_{k+1}^{k+l_c} T_{k+1}^{k+l_c}},\\[4pt] \ket{\hat \Psi_{\alpha_k|a_{I_{\bar w}}, \alpha_{I_{\bar w}},m_k=1}^{(k)}}_{A_k T_k I_{k+1}^{k+l_c} M_{k+1}^{k+l_c} T_{k+1}^{k+l_c}} = \frac{1}{\sqrt 2} \sum_{a_k \in \{0,1\}} \ket{(a_k)_{\alpha_k}}_{A_k} \ket*{\Xi_{a_k,\alpha_k \vert a_{\bar k},\alpha_{\bar k}}}_{T_k I_{k+1}^{k+l_c} M_{k+1}^{k+l_c} T_{k+1}^{k+l_c}},
    \end{gathered}
    \end{equation}
    with $\ket{\Xi_{a_k,\alpha_k\vert a_{\bar k},\alpha_{\bar k}}}$ defined in \cref{eq:Xi_decoy}. Using $\braket*{(a_k)_{\alpha_k}}{(a'_k)_{\alpha'_k}} = (-1)^{a_k a'_k}/\sqrt{2}$ for $\alpha_k \neq \alpha'_k$, we then obtain}
    \begin{equation}
    \label{eq:p_coin_minus_decoy}
    \begin{aligned}
     p_{C_k = - \vert m_k=1,\mathrm{trash},\mathrm{sifted}}^{(k)}
     &= \norm{{}_{C_k}\!\braket{-} {\hat \Psi_{\mathrm{coin}|a_{I_{\bar w}}, \alpha_{I_{\bar w}},m_k=1}^{(k)}}_{C_k A_k T_k I_{k+1}^{k+l_c} M_{k+1}^{k+l_c} T_{k+1}^{k+l_c}}}^2 \\
     &= \frac{1}{2} \bigg(1-\Re \braket{\hat \Psi_{X|a_{I_{\bar w}}, \alpha_{I_{\bar w}},m_k=1}^{(k)}}{\hat \Psi_{Z|a_{I_{\bar w}}, \alpha_{I_{\bar w}},m_k=1}^{(k)}}\bigg) \\
     &= \frac{1}{2} \Re \Bigg(1 - \frac{1}{2\sqrt{2}} \Big[\braket{\Xi_{0,X\vert a_{\bar k},\alpha_{\bar k}}}{\Xi_{0,Z\vert a_{\bar k},\alpha_{\bar k}}} + \braket{\Xi_{0,X\vert a_{\bar k},\alpha_{\bar k}}}{\Xi_{1,Z\vert a_{\bar k},\alpha_{\bar k}}} \\
     &\qquad\qquad\qquad + \braket{\Xi_{1,X\vert a_{\bar k},\alpha_{\bar k}}}{\Xi_{0,Z\vert a_{\bar k},\alpha_{\bar k}}} - \braket{\Xi_{1,X\vert a_{\bar k},\alpha_{\bar k}}}{\Xi_{1,Z\vert a_{\bar k},\alpha_{\bar k}}}\Big] \Bigg),
    \end{aligned}
    \end{equation}
    where the subscript $m_k=1$ indicates that we have postselected on the single-photon component. Therefore, we have that
    \begin{equation}
    \label{eq:EXk_bound_decoy}
        E[\bm{X_k} \mid a_{I_{\bar w}}, \alpha_{I_{\bar w}}] \leq \frac{1}{2} p_1 (1-p_\mathrm{keep}) \Delta_{\rm coin}^{\rm U},
    \end{equation}
    where $\Delta_{\rm coin}^{\rm U}$ is defined in \cref{eq:Delta_coin_U_decoy}.

    As already argued, under the conditional probability measure in \cref{eq:cond_prob_meas_decoy}, the RV $\bm{n}_{\mathrm{trash},\mathrm{sifted},-,(1)}^{(I_w)}$ defined in \cref{eq:n_trash_Iw_decoy} is a sum of independent Bernoulli RVs. Thus, by Bernstein's inequality (Lemma \ref{lem:bernstein} in \cref{app:concentration_ineqs}), we have that
    \begin{equation}
    \label{eq:bernstein_applied_decoy}
    \begin{aligned}
\Pr\bigg[\bm{n}_{\mathrm{trash},\mathrm{sifted},-,(1)}^{(I_w)} &> E[\bm{n}_{\mathrm{trash},\mathrm{sifted},-,(1)}^{(I_w)} \vert a_{I_{\bar w}}, \alpha_{I_{\bar w}}] \\
        &+ \sqrt{2\, E[\bm{n}_{\mathrm{trash},\mathrm{sifted},-,(1)}^{(I_w)} \vert a_{I_{\bar w}}, \alpha_{I_{\bar w}}]\, \ln \tfrac{1}{\epsilon_C}} + \tfrac{2}{3}\ln\tfrac{1}{\epsilon_C} \;\Big\vert\; a_{I_{\bar w}}, \alpha_{I_{\bar w}}\bigg] \leq \epsilon_C.
    \end{aligned}
    \end{equation}

    Note that, from \cref{eq:EXk_bound_decoy},
    \begin{equation}
    \label{eq:E_n_trash_Iw_decoy}
        E[\bm{n}_{\mathrm{trash},\mathrm{sifted},-,(1)}^{(I_w)}\vert a_{I_{\bar w}}, \alpha_{I_{\bar w}}] = \sum_{k \in I_w} E[\bm{X_k} \vert a_{I_{\bar w}}, \alpha_{I_{\bar w}}] \leq \frac{N}{l_c+1} \frac{1}{2} p_1 (1-p_\mathrm{keep}) \Delta_{\rm coin}^{\rm U}.
    \end{equation}
    Combining \cref{eq:bernstein_applied_decoy} with \cref{eq:E_n_trash_Iw_decoy}, and using that $\mu \mapsto \mu + \sqrt{2\mu \ln(1/\epsilon_C)} + \tfrac{2}{3}\ln(1/\epsilon_C)$ is monotonically increasing in $\mu$, we have that
    \begin{equation}
    \begin{aligned}
        \Pr\Big[\bm{n}_{\mathrm{trash},\mathrm{sifted},-,(1)}^{(I_w)} &> \frac{N}{l_c+1} \frac{1}{2} p_1 (1-p_\mathrm{keep}) \Delta_{\rm coin}^{\rm U} + \sqrt{\frac{N\, p_1 (1-p_\mathrm{keep}) \Delta_{\rm coin}^{\rm U}}{l_c+1}\, \ln \tfrac{1}{\epsilon_C}} + \tfrac{2}{3}\ln\tfrac{1}{\epsilon_C} \;\Big\vert\; a_{I_{\bar w}}, \alpha_{I_{\bar w}}\Big] \leq \epsilon_C.
    \end{aligned}
    \end{equation}

    Using that $\abs{I_w} = \frac{N}{l_c+1}$ and the definition of $n_{\mathrm{trash},\mathrm{sifted},-,(1)}^\mathrm{U}$ in \cref{eq:n_trash_U_decoy}, we obtain
    \begin{equation}
        \Pr[\bm{n}_{\mathrm{trash},\mathrm{sifted},-,(1)}^{(I_w)} > n_{\mathrm{trash},\mathrm{sifted},-,(1)}^\mathrm{U}/(l_c+1) \Big\vert  a_{I_{\bar w}}, \alpha_{I_{\bar w}}] \leq \epsilon_C.
    \end{equation}

    Then, \cref{eq:bound_n_trash_Iw_decoy} directly follows, since
    \begin{equation}
    \begin{aligned}
     &\Pr[\bm{n}_{\mathrm{trash},\mathrm{sifted},-,(1)}^{(I_w)} > n_{\mathrm{trash},\mathrm{sifted},-,(1)}^\mathrm{U}/(l_c+1)] \\
     &= \sum_{a_{I_{\bar w}}, \alpha_{I_{\bar w}}} \Pr[a_{I_{\bar w}}, \alpha_{I_{\bar w}}]\Pr[\bm{n}_{\mathrm{trash},\mathrm{sifted},-,(1)}^{(I_w)} > n_{\mathrm{trash},\mathrm{sifted},-,(1)}^\mathrm{U}/(l_c+1) \Big\vert a_{I_{\bar w}}, \alpha_{I_{\bar w}}] \leq \epsilon_C.
     \end{aligned}
    \end{equation}

\end{proof}

\subsection{Proof of \texorpdfstring{\cref{thm:phase_error_bound_quantum_coin_decoy}}{Theorem 2}}
\label{appsub:proof:thm:phase_error_bound_quantum_coin_decoy}
\begin{proof}

We prove this by considering the complement of the failure events and applying a union bound. First, we define the "good" events:
\begin{itemize}
    \item $\Omega_A$: the event where the bound from Lemma~\ref{lem:quantum_coin_inequality_decoy} holds;
    \item $\Omega_C$: the event where the bound from Lemma~\ref{lem:bernstein_by_groups_decoy} holds;
    \item $\Omega_{\rm dec}$: the complement of $\Omega_{\rm dec,fail}$, i.e., the event where the decoy-state bounds hold.
\end{itemize}

From the lemmas, we have that (where the superscript $c$ denotes the complement of an event)
\begin{align}
    \Pr[\Omega_A^c] \leq 5\epsilon_A, \qquad
    \Pr[\Omega_C^c] \leq (l_c+1)\epsilon_C, \qquad
    \Pr[\Omega_{\rm dec}^c] \leq \epsilon_{\rm decoy}.
\end{align}

On the event $\Omega_A \cap \Omega_C \cap \Omega_{\rm dec}$, we have:
\begin{itemize}
    \item From $\Omega_A$: the inequality from \eqref{eq:finite_azuma_decoy} holds.

    \item From $\Omega_C$: $\bm{n}_{\mathrm{trash},\mathrm{sifted},-,(1)} \leq n_{\mathrm{trash},\mathrm{sifted},-,(1)}^\mathrm{U}$;

    \item From $\Omega_{\rm dec}$: 
    \begin{align}
        \bm{n}_{\mathrm{keep},\mathrm{sifted},Z,(1)}^{\mathrm{det,L}} &\leq \bm{n}_{\mathrm{keep},\mathrm{sifted},Z,(1)}^{\mathrm{det}} \leq \bm{n}_{\mathrm{keep},\mathrm{sifted},Z,(1)}^{\mathrm{det,U}}, \\
        \bm{n}_{\mathrm{keep},\mathrm{sifted},X,(1)}^{\mathrm{det,L}} &\leq \bm{n}_{\mathrm{keep},\mathrm{sifted},X,(1)}^{\mathrm{det}}, \\
        \bm{n}_{\mathrm{keep},\mathrm{sifted},X,(1)}^{\mathrm{err}} &\leq \bm{n}_{\mathrm{keep},\mathrm{sifted},X,(1)}^{\mathrm{err,U}};
    \end{align}    
\end{itemize}

Since $G_+(x,y)$ is increasing in $x$ and decreasing in $y$, it is clear that the RHS of \cref{eq:finite_azuma_decoy} is increasing in $\bm{n}_{\mathrm{keep},\mathrm{sifted},X,(1)}^{\mathrm{err}}$ and $\bm{n}_{\mathrm{trash},\mathrm{sifted},-,(1)}$, and decreasing in $\bm{n}_{\mathrm{keep},\mathrm{sifted},X,(1)}^{\mathrm{det}}$. Moreover, as proven in Ref.~\cite[Supplementary Note D]{curras-lorenzoSecurityQuantum2026}, the RHS of \cref{eq:finite_azuma_decoy} is increasing in $\bm{n}_{\mathrm{keep},\mathrm{sifted},Z,(1)}^{\mathrm{det}}$. Using all these monotonicity conditions, on the event $\Omega_A \cap \Omega_C \cap \Omega_{\rm dec}$, we can upper bound the right-hand side of \eqref{eq:finite_azuma_decoy} by substituting:
\begin{itemize}
    \item $\bm{n}_{\mathrm{keep},\mathrm{sifted},Z,(1)}^{\mathrm{det}} \to \bm{n}_{\mathrm{keep},\mathrm{sifted},Z,(1)}^{\mathrm{det,U}};$ 
    \item $\bm{n}_{\mathrm{keep},\mathrm{sifted},X,(1)}^{\mathrm{det}} \to \bm{n}_{\mathrm{keep},\mathrm{sifted},X,(1)}^{\mathrm{det,L}};$ 
    \item $\bm{n}_{\mathrm{keep},\mathrm{sifted},X,(1)}^{\mathrm{err}} \to \bm{n}_{\mathrm{keep},\mathrm{sifted},X,(1)}^{\mathrm{err,U}};$ 
    \item $\bm{n}_{\mathrm{trash},\mathrm{sifted},-,(1)} \to n_{\mathrm{trash},\mathrm{sifted},-,(1)}^\mathrm{U}.$
\end{itemize}

This results in 
\begin{equation}
\begin{aligned}
    &\bm{n}_{\mathrm{keep},\mathrm{sifted},Z,(1)}^{\mathrm{err}} \leq (\bm{n}_{\mathrm{keep},\mathrm{sifted},Z,(1)}^{\mathrm{det,U}} + \Delta_A) \, \times \\
   & G_+ \left(\frac{\bm{n}_{\mathrm{keep},\mathrm{sifted},X,(1)}^{\mathrm{err,U}} + \Delta_A}{\bm{n}_{\mathrm{keep},\mathrm{sifted},X,(1)}^{\mathrm{det,L}} - \Delta_A}, 1 - \frac{2 p_{\mathrm{keep}}(n_{\mathrm{trash},\mathrm{sifted},-,(1)}^\mathrm{U} + \Delta_A)}{(1-p_{\mathrm{keep}}) (\bm{n}_{\mathrm{keep},\mathrm{sifted},Z,(1)}^{\mathrm{det,U}}+\bm{n}_{\mathrm{keep},\mathrm{sifted},X,(1)}^{\mathrm{det,L}})}\right) + \Delta_A
   \end{aligned}
\end{equation}
on the event $\Omega_A \cap \Omega_C \cap \Omega_{\rm dec}$. Note that on $\Omega_{\rm dec}$, we have in particular that $\bm{n}_{\mathrm{keep},\mathrm{sifted},Z,(1)}^{\mathrm{det}} \geq \bm{n}_{\mathrm{keep},\mathrm{sifted},Z,(1)}^{\mathrm{det,L}}$. Dividing both sides of the above inequality by $\bm{n}_{\mathrm{keep},\mathrm{sifted},Z,(1)}^{\mathrm{det}}$ and using the fact that $\bm{n}_{\mathrm{keep},\mathrm{sifted},Z,(1)}^{\mathrm{det}} \geq \bm{n}_{\mathrm{keep},\mathrm{sifted},Z,(1)}^{\mathrm{det,L}}$, we obtain
\begin{equation}
\begin{aligned}
    &\bm{\ephone} = \frac{\bm{n}_{\mathrm{keep},\mathrm{sifted},Z,(1)}^{\mathrm{err}}}{\bm{n}_{\mathrm{keep},\mathrm{sifted},Z,(1)}^{\mathrm{det}}} \leq \frac{\bm{n}_{\mathrm{keep},\mathrm{sifted},Z,(1)}^{\mathrm{det,U}} + \Delta_A}{\bm{n}_{\mathrm{keep},\mathrm{sifted},Z,(1)}^{\mathrm{det,L}}} \, \times\\
    &G_+ \left(\frac{\bm{n}_{\mathrm{keep},\mathrm{sifted},X,(1)}^{\mathrm{err,U}} + \Delta_A}{\bm{n}_{\mathrm{keep},\mathrm{sifted},X,(1)}^{\mathrm{det,L}} - \Delta_A}, 1 - \frac{2 p_{\mathrm{keep}}(n_{\mathrm{trash},\mathrm{sifted},-,(1)}^\mathrm{U} + \Delta_A)}{(1-p_{\mathrm{keep}}) (\bm{n}_{\mathrm{keep},\mathrm{sifted},Z,(1)}^{\mathrm{det,U}}+\bm{n}_{\mathrm{keep},\mathrm{sifted},X,(1)}^{\mathrm{det,L}})}\right) + \frac{\Delta_A}{\bm{n}_{\mathrm{keep},\mathrm{sifted},Z,(1)}^{\mathrm{det,L}}} = \mathcal{E}_{\mathrm{ph},(1)}(\bm{\vec{n}};\epsilon_A,\epsilon_C)
\end{aligned}
\end{equation}
on $\Omega_A \cap \Omega_C \cap \Omega_{\rm dec}$. Therefore, we have shown that
\begin{equation}
    \Omega_A \cap \Omega_C \cap \Omega_{\rm dec} \subseteq \{\bm{\ephone} \leq \mathcal{E}_{\mathrm{ph},(1)}(\bm{\vec{n}};\epsilon_A,\epsilon_C)\}.
\end{equation}
Taking complements and intersecting both sides with $\Omega_{\rm dec}$, we have
\begin{equation}
\label{eq:set_containment}
    \{\bm{\ephone} > \mathcal{E}_{\mathrm{ph},(1)}(\bm{\vec{n}};\epsilon_A,\epsilon_C)\} \cap \Omega_{\rm dec} \subseteq \Omega_A^c \cup \Omega_C^c.
\end{equation}

Finally, applying the union bound:
\begin{equation}
\begin{aligned}
    &\Pr[\bm{n}_{\mathrm{keep},\mathrm{sifted},Z,(1)}^{\mathrm{det}} < \bm{n}_{\mathrm{keep},\mathrm{sifted},Z,(1)}^{\mathrm{det,L}} \cup \bm{\ephone} > \mathcal{E}_{\mathrm{ph},(1)}(\bm{\vec{n}};\epsilon_A,\epsilon_C)] \\
    &\leq \Pr[\Omega_{\rm dec}^c] + \Pr[\bm{\ephone} > \mathcal{E}_{\mathrm{ph},(1)}(\bm{\vec{n}};\epsilon_A,\epsilon_C) \cap \Omega_{\rm dec}] \\
    &\leq \Pr[\Omega_{\rm dec}^c] + \Pr[\Omega_A^c \cup \Omega_C^c] \\
    &\leq \epsilon_{\rm decoy} + \Pr[\Omega_A^c] + \Pr[\Omega_C^c] \\
    &= \epsilon_{\rm decoy} + 5\epsilon_A + (l_c+1)\epsilon_C,
\end{aligned}
\end{equation}
where the first inequality uses $\{\bm{n}_{\mathrm{keep},\mathrm{sifted},Z,(1)}^{\mathrm{det}} < \bm{n}_{\mathrm{keep},\mathrm{sifted},Z,(1)}^{\mathrm{det,L}}\} \subseteq \Omega_{\rm dec}^c$ and the second inequality uses \cref{eq:set_containment}.
\end{proof}

\subsection{Proof of \texorpdfstring{\cref{lem:bound_Delta_coin_decoy}}{Lemma 5}}
\label{appsub:proof:lem:bound_Delta_coin_decoy}

\begin{proof}
We bound $\Delta_{\rm coin}^{\rm U}$ by analyzing the overlaps appearing in its definition in \cref{eq:Delta_coin_U_decoy}. Recall from \cref{eq:Xi_decoy} that
\begin{equation}
    \ket*{\Xi_{a_k,\alpha_k \mid a_{\bar{k}},\alpha_{\bar{k}}}}_{T_k\, I_{k+1}^{k+l_c}\, M_{k+1}^{k+l_c}\, T_{k+1}^{k+l_c}}
    =
    \ket*{\psi_{a_{k-l_c}^{k},\,\alpha_{k-l_c}^{k}}^{(1)}}_{T_k}
    \;\otimes
    \bigotimes_{j=k+1}^{k+l_c}
        \ket*{\Psi_{a_{j-l_c}^{j},\,\alpha_{j-l_c}^{j}}}_{I_j M_j T_j},
\end{equation}
where $a_{\bar{k}} = (a_{k-l_c}^{k-1}, a_{k+1}^{k+l_c})$ and $\alpha_{\bar{k}} = (\alpha_{k-l_c}^{k-1}, \alpha_{k+1}^{k+l_c})$. This state factorizes into the single-photon state of round~$k$ and the full purification states of the subsequent $l_c$ rounds. Accordingly, overlaps between $\Xi$~states that differ only in $(a_k,\alpha_k)$ factorize as
\begin{equation}
\label{eq:Xi_overlap_factorization}
    \braket*{\Xi_{a'_k,\alpha'_k \mid a_{\bar{k}},\alpha_{\bar{k}}}}{\Xi_{a_k,\alpha_k \mid a_{\bar{k}},\alpha_{\bar{k}}}}_{T_k\, I_{k+1}^{k+l_c}\, M_{k+1}^{k+l_c}\, T_{k+1}^{k+l_c}}
    =
    \braket*{\psi_{a_{k-l_c}^{k-1}a'_k,\,\alpha_{k-l_c}^{k-1}\alpha'_k}^{(1)}}{\psi_{a_{k-l_c}^{k-1}a_k,\,\alpha_{k-l_c}^{k-1}\alpha_k}^{(1)}}_{T_k}
    \;\times\;
    \prod_{j=k+1}^{k+l_c} O_j(a_k,\alpha_k;\,a'_k,\alpha'_k),
\end{equation}
where $O_j$ denotes the overlap of the purifications for round $j$ when varying the $k$-th bit and basis choices.

Consider first the round-$k$ factor. Using the model in \cref{eq:correlations_model_state_truncated,eq:correlations_model_phase_truncated}, the single-photon overlap for round~$k$ depends only on the difference of encoding phases. Since both states share the same history $a_{k-l_c}^{k-1},\alpha_{k-l_c}^{k-1}$, the correlation contributions cancel so that
\begin{equation}
\label{eq:round_k_overlap}
    \braket*{\psi_{a_{k-l_c}^{k-1}a'_k,\,\alpha_{k-l_c}^{k-1}\alpha'_k}^{(1)}}{\psi_{a_{k-l_c}^{k-1}a_k,\,\alpha_{k-l_c}^{k-1}\alpha_k}^{(1)}}_{T_k}
    = \cos\!\big(\hat\theta_{a_k,\alpha_k} - \hat\theta_{a'_k,\alpha'_k}\big)
    = \braket*{(a'_k)_{\alpha'_k}}{(a_k)_{\alpha_k}},
\end{equation}
recovering the ideal BB84 overlap.

We now turn to the later-round overlaps. For each subsequent round $j = k+1,\ldots,k+l_c$, the two purifications being compared share all history entries except $(a_k,\alpha_k)$ versus $(a'_k,\alpha'_k)$. Define the overlap of their single-photon components as
\begin{equation}
\label{eq:gamma_def}
    \gamma_{a_k,\alpha_k,a'_k,\alpha'_k}^{(j-k)}
    \coloneqq
    \cos\!\big(\delta_{a_k,\alpha_k}^{(j-k)} - \delta_{a'_k,\alpha'_k}^{(j-k)}\big)
    \geq \cos\!\big(\Delta_{j-k}\big),
\end{equation}
where the inequality follows directly from the assumption in \cref{eq:correlations_model_delta_Delta_exp_truncated}. For an $m$-photon state of the form in \cref{eq:correlations_model_state_truncated}, which consists of $m$ identical photons, the overlap is $\big(\gamma_{a_k,\alpha_k,a'_k,\alpha'_k}^{(j-k)}\big)^m$. 

Combining this with the definition of the purifications in \cref{eq:round_state_decoy} and using the Poisson probability generating function $\sum_{m=0}^{\infty} p_{m|\mu}\, t^m = e^{-\mu(1-t)}$, we obtain
\begin{equation}
\label{eq:full_round_overlap}
    O_j(a_k,\alpha_k;\,a'_k,\alpha'_k)
    = \sum_{\mu} p_{\mu} \sum_{m=0}^{\infty} p_{m|\mu}\, \big(\gamma_{a_k,\alpha_k,a'_k,\alpha'_k}^{(j-k)}\big)^m
    = \sum_{\mu} p_{\mu}\, \exp\!\Big[-\mu\,\big(1 - \gamma_{a_k,\alpha_k,a'_k,\alpha'_k}^{(j-k)}\big)\Big].
\end{equation}
Since $\gamma_{a_k,\alpha_k,a'_k,\alpha'_k}^{(j-k)} \geq \cos(\Delta_{j-k})$ and $t \mapsto e^{-\mu(1-t)}$ is monotonically increasing, we have
\begin{equation}
\label{eq:full_round_overlap_bound}
    O_j(a_k,\alpha_k;\,a'_k,\alpha'_k)
    \geq \sum_{\mu} p_{\mu}\, \exp\!\Big[-\mu\,\big(1 - \cos(\Delta_{j-k})\big)\Big].
\end{equation}

It remains to evaluate the coin parameter. We want to compute the sum appearing in \cref{eq:Delta_coin_U_decoy}, i.e.,
\begin{equation}
    S \coloneqq \frac{1}{2\sqrt{2}} \sum_{a,a' \in \{0,1\}} (-1)^{aa'}\, \braket*{\Xi_{a,X \mid a_{\bar{k}},\alpha_{\bar{k}}}}{\Xi_{a',Z \mid a_{\bar{k}},\alpha_{\bar{k}}}}.
\end{equation}
Using the factorization in \cref{eq:Xi_overlap_factorization}, the ideal BB84 overlaps in \cref{eq:round_k_overlap}, and noting that $\braket{0_X}{0_Z} = \braket{0_X}{1_Z} = \braket{1_X}{0_Z} = 1/\sqrt{2}$ and $\braket{1_X}{1_Z} = -1/\sqrt{2}$, the factor of $(-1)^{aa'}$ exactly compensates for the sign of each overlap, obtaining
\begin{equation}
    S = \frac{1}{4} \sum_{a,a' \in \{0,1\}} \prod_{j=k+1}^{k+l_c} O_j\big(a,X;\,a',Z\big).
\end{equation}
Since every term in the sum is bounded below by \cref{eq:full_round_overlap_bound}, we obtain
\begin{equation}
    S \geq \prod_{l=1}^{l_c} \sum_{\mu} p_{\mu}\, \exp\!\Big[-\mu\,\big(1 - \cos(\Delta_{l})\big)\Big].
\end{equation}
Finally, since $\Delta_{\rm coin}^{\rm U} = \max_{k,\,a_{\bar{k}},\alpha_{\bar{k}}} \frac{1}{2}\Re(1 - S)$ and $S$ is real, the bound in \cref{eq:Delta_coin_U_decoy_bound} follows.
\end{proof}

\subsection{Proof of \texorpdfstring{\cref{lem:trace_distance_exponential_decay_decoy}}{Lemma 6}}
\label{appsub:proof:lem:trace_distance_exponential_decay_decoy}

\begin{proof}
Since both states are pure, the trace distance equals $T = \sqrt{1 - F^2}$, where $F = \big\lvert\!\braket*{\Psi_N^{(l_c)}}{\Psi_N^{(\infty)}}_{C_1^N A_1^N I_1^N M_1^N T_1^N}\!\big\rvert$ is the fidelity. The objective is thus to bound $F$ from below.

We begin by computing the inner product of the global states. Both $\ket*{\Psi_N^{(\infty)}}_{C_1^N A_1^N I_1^N M_1^N T_1^N}$ and $\ket*{\Psi_N^{(l_c)}}_{C_1^N A_1^N I_1^N M_1^N T_1^N}$ have the structure given in \cref{eq:global_state_decoy_correlations}, with identical coin, qubit, and intensity/photon-number registers $C_k$, $A_k$, $I_k$, $M_k$---they differ only in the photonic states $\ket*{\psi_{a_1^k,\alpha_1^k}^{(m)}}_{T_k}$ versus $\ket*{\psi_{a_{k-l_c}^k,\alpha_{k-l_c}^k}^{(m)}}_{T_k}$. Since the states $\{\ket{\alpha_k}_{C_k}\ket{(a_k)_{\alpha_k}}_{A_k}\}_{a_k,\alpha_k}$ are orthonormal, the inner product collapses to a single sum,
\begin{equation}
\label{eq:global_overlap}
\braket*{\Psi_N^{(l_c)}}{\Psi_N^{(\infty)}}_{C_1^N A_1^N I_1^N M_1^N T_1^N} = \sum_{\alpha_1^N,\, a_1^N} \frac{1}{4^N} \prod_{k=1}^{N} \braket*{\Psi_{a_{k-l_c}^k,\,\alpha_{k-l_c}^k}}{\Psi_{a_1^k,\,\alpha_1^k}}_{I_k M_k T_k},
\end{equation}
where the factor $1/4^N$ arises because each round contributes a coefficient $1/2$ in the global state \cref{eq:global_state_decoy_correlations}, giving $(1/2)^2 = 1/4$ per round in the squared amplitude.

We now evaluate each per-round overlap. Using the definition of the purifications in \cref{eq:round_state_decoy} and the orthonormality of the intensity and photon-number bases, the overlap for round~$k$ factorizes as
\begin{equation}
\label{eq:round_overlap}
\braket*{\Psi_{a_{k-l_c}^k,\,\alpha_{k-l_c}^k}}{\Psi_{a_1^k,\,\alpha_1^k}}_{I_k M_k T_k} = \sum_{\mu} p_{\mu} \sum_{m=0}^{\infty} p_{m|\mu}\, \braket*{\psi_{a_{k-l_c}^k,\,\alpha_{k-l_c}^k}^{(m)}}{\psi_{a_1^k,\,\alpha_1^k}^{(m)}}_{T_k}.
\end{equation}
The two $m$-photon states share the same bit and basis values $(a_k,\alpha_k)$ and the same recent history $(a_{k-l_c}^{k-1},\alpha_{k-l_c}^{k-1})$; they differ only in the contributions from rounds more than $l_c$ steps in the past. In the truncated model, these long-range contributions are evaluated at the fixed reference choices $(a^\ast,\alpha^\ast)$. From \cref{eq:correlations_model_phase,eq:correlations_model_phase_truncated}, the resulting phase difference is
\begin{equation}
\label{eq:phase_difference}
\Delta\theta_k \coloneqq \theta_{a_1^k,\alpha_1^k} - \theta_{a_{k-l_c}^k,\alpha_{k-l_c}^k} = \sum_{l=l_c+1}^{k-1} \left( \delta_{a_{k-l},\alpha_{k-l}}^{(l)} - \delta_{a^\ast,\alpha^\ast}^{(l)} \right).
\end{equation}
Using the pairwise correlation bound in \cref{eq:correlations_model_delta_Delta_exp_truncated}, we obtain
\begin{equation}
\label{eq:xi_tail_decoy}
\abs{\Delta\theta_k} \leq \sum_{l=l_c+1}^{\infty}\Delta_l \eqqcolon \Delta_{\mathrm{tail}}(l_c).
\end{equation}

To proceed, we need the $m$-photon overlap. By the correlations model in \cref{eq:correlations_model_state,eq:correlations_model_state_truncated}, both states have the symmetric Fock-space form $\ket*{\psi^{(m)}_\theta} = \frac{1}{\sqrt{m!}}(\cos\theta\, a_0^\dagger + \sin\theta\, a_1^\dagger)^m\ket{\mathrm{vac}}$, with angles $\theta$ and $\theta + \Delta\theta_k$ respectively. Expanding in the two-mode Fock basis $\ket{m-j,\,j}$, one has
\begin{equation}
\ket*{\psi^{(m)}_\theta} = \sum_{j=0}^{m} \sqrt{\binom{m}{j}}\, \cos^{m-j}\!\theta\;\sin^j\!\theta\;\ket{m-j,\,j},
\end{equation}
so the overlap between states with angles $\theta$ and $\theta' = \theta + \Delta\theta_k$ is
\begin{equation}
\label{eq:m_photon_overlap}
\braket*{\psi^{(m)}_{\theta'}}{\psi^{(m)}_\theta} = \sum_{j=0}^{m} \binom{m}{j} \big(\cos\theta'\cos\theta\big)^{m-j} \big(\sin\theta'\sin\theta\big)^{j} = \big(\cos\theta'\cos\theta + \sin\theta'\sin\theta\big)^m = \cos^m(\Delta\theta_k),
\end{equation}
where the second step uses the binomial theorem and the third uses the angle-subtraction identity for cosine.

Substituting \cref{eq:m_photon_overlap} into \cref{eq:round_overlap} and using the Poisson probability generating function $\sum_{m=0}^{\infty} p_{m|\mu} t^m = e^{-\mu(1-t)}$, we obtain
\begin{equation}
\braket*{\Psi_{a_{k-l_c}^k,\,\alpha_{k-l_c}^k}}{\Psi_{a_1^k,\,\alpha_1^k}}_{I_k M_k T_k} = \sum_{\mu} p_{\mu}\, e^{-\mu(1-\cos\Delta\theta_k)} .
\end{equation}
Since $1-\cos x \leq x^2/2$ for all real $x$, and since $\abs{\Delta\theta_k}\leq \Delta_{\mathrm{tail}}(l_c)$, this is bounded below by
\begin{equation}
\braket*{\Psi_{a_{k-l_c}^k,\,\alpha_{k-l_c}^k}}{\Psi_{a_1^k,\,\alpha_1^k}}_{I_k M_k T_k} \geq c\!\left(\Delta_{\mathrm{tail}}(l_c)\right), \qquad c(x) \coloneqq \sum_{\mu} p_{\mu}\, e^{-\mu x^2/2}.
\end{equation}

Since every per-round overlap is bounded below by $c(\Delta_{\mathrm{tail}}(l_c))$, \cref{eq:global_overlap} gives
\begin{equation}
F \geq \sum_{\alpha_1^N,\, a_1^N} \frac{1}{4^N} \prod_{k=1}^{N} c\!\left(\Delta_{\mathrm{tail}}(l_c)\right) = c\!\left(\Delta_{\mathrm{tail}}(l_c)\right)^N ,
\end{equation}
where the sum evaluates to unity by normalization (there are $4^N$ terms, each weighted by $1/4^N$).

We now convert this fidelity bound into a trace distance bound. Using $1 - t^2 \leq 2(1-t)$ for $t \in [0,1]$, we have $T = \sqrt{1 - F^2} \leq \sqrt{2(1-F)}$. To bound $1-F$, we write $c = c(\Delta_{\mathrm{tail}}(l_c))$ and apply Bernoulli's inequality $c^N = (1-(1-c))^N \geq 1-N(1-c)$, which gives $1-F \leq 1-c^N \leq N(1-c)$. Therefore
\begin{equation}
T \leq \sqrt{2N\left(1-c\!\left(\Delta_{\mathrm{tail}}(l_c)\right)\right)}.
\end{equation}
It remains to simplify $1-c(x)$. Using $1-e^{-y}\leq y$ for $y\geq0$, we have
\begin{equation}
1-c(x) = \sum_{\mu} p_{\mu} \left(1-e^{-\mu x^2/2}\right) \leq \sum_{\mu}p_{\mu}\frac{\mu x^2}{2} = \frac{\bar{\mu}x^2}{2}.
\end{equation}
Substituting $x=\Delta_{\mathrm{tail}}(l_c)$, we obtain
\begin{equation}
T \leq \sqrt{2N\cdot\frac{\bar{\mu}\left(\Delta_{\mathrm{tail}}(l_c)\right)^2}{2}} = \sqrt{N\bar{\mu}}\, \Delta_{\mathrm{tail}}(l_c),
\end{equation}
which establishes \cref{eq:trace_distance_asymp_decoy}.

Finally, we evaluate the tail sum explicitly. The tail sum in \cref{eq:xi_tail_decoy} is a geometric series:
\begin{equation}
\Delta_{\mathrm{tail}}(l_c) = \sum_{l=l_c+1}^{\infty}\Delta_l = \sum_{l=l_c+1}^{\infty} \Delta_1\,e^{-\lambda(l-1)} = \frac{\Delta_1\,e^{-\lambda l_c}}{1-e^{-\lambda}}.
\end{equation}
Substituting this into the trace-distance bound gives
\begin{equation}
\label{eq:trace_distance_asymp_decoy_proof}
T\Big(\ketbra*{\Psi_N^{(\infty)}}, \ketbra*{\Psi_N^{(l_c)}}\Big) \leq \frac{\sqrt{N\bar{\mu}}\,\Delta_1\, e^{-\lambda l_c}}{1-e^{-\lambda}}.
\end{equation}
Thus, to ensure $T\leq d$, it suffices to choose
\begin{equation}
\label{eq:l_c_formula_stable_decoy_proof}
l_c \geq \frac{1}{\lambda} \ln\!\left(\frac{\sqrt{N\bar{\mu}}\,\Delta_1}{d(1-e^{-\lambda})}\right).
\end{equation}
\end{proof}

{ \section{Reduction to single-photon sources}
\label{app:single_photon}

The main text presents the security proof for a decoy-state implementation, in which Alice emits phase-randomized {weak} coherent states that are diagonal in the Fock basis and therefore mixed even after conditioning on all of her setting choices. In this appendix we show how to apply the analysis to the conceptually simpler case of an ideal single-photon source, in which Alice emits exactly one photon per round.

In this case, three simplifications occur. First, the intensity register $I_k$ and the photon-number register $M_k$ are no longer needed, since the emitted state is pure conditional on the bit-and-basis history. Second, {now Alice sends exactly one photon in all rounds}, so the single-photon detection and error counts (which in the decoy-state case must be estimated via the {decoy-state} method) become directly observable, and the decoy-state estimation step (\cref{lem:decoy_state_bounds} of the main text) is eliminated. Third, the quantum-coin and bad-coin analyses simplify, with $p_1 \to 1$ and the Poisson average over intensities collapsing to a simple single-photon overlap. We emphasize that the resulting proof still uses the new quantum-coin approach of the main text, and in particular does \emph{not} partition the rounds into $l_c+1$ groups for the phase-error estimation step; it therefore also improves, in the single-photon setting, upon the group-partitioning analysis of Ref.~\cite{curras-lorenzoRigorousPhaseerrorestimation2026}.

We organize {this} appendix to follow the structure of the main text, indicating in each case the modifications relative to the decoy-state analysis. Throughout, single-photon counterparts of the quantities defined in the main text are written without the photon-number subscript $(1)$, since every round is now a single-photon round.

\subsection{Alice's source}
\label{appsub:sp_source}

In each round $k$, Alice independently selects a bit $a_k \in \{0,1\}$ and a basis $\alpha_k \in \{Z,X\}$, both uniformly at random. Based on these choices, she prepares the single-photon pure state
\begin{equation}
\label{eq:sp_state}
    \ket*{\psi_{a_1^k,\alpha_1^k}}_{T_k},
\end{equation}
which, due to the correlations introduced by the bit-and-basis encoder, depends not only on the current choices $a_k$ {and} $\alpha_k$ but also on the histories $a_1^{k-1}$ and $\alpha_1^{k-1}$ of previous choices.

Introducing, for each round $k$, a coin register $C_k$ encoding the basis choice and a qubit register $A_k$ encoding the bit in the corresponding basis, the global source-replacement state has the same structure as \cref{eq:global_state_decoy_correlations}, but without the intensity register $I_k$ and the photon-number register $M_k$:
\begin{equation}
\label{eq:sp_global_state}
    \ket{\Psi_N}_{C_1^N A_1^N T_1^N}
    = \sum_{\alpha_1^N \in \{Z,X\}^N}\;\sum_{a_1^N \in \{0,1\}^N}
    \bigotimes_{k=1}^{N} \frac{1}{2}\,
    \ket{\alpha_k}_{C_k}\,
    \ket{(a_k)_{\alpha_k}}_{A_k}\,
    \ket*{\psi_{a_1^k,\alpha_1^k}}_{T_k}.
\end{equation}
If Alice measures {$C_1^N$} in the $\{\ket{Z},\ket{X}\}$ basis and {$A_1^N$} in the basis determined by the outcome $\alpha_k$, the post-measurement state on $T_1^N$ is exactly the state the actual protocol would have produced for the same sequence of outcomes.

Bob's receiver is unchanged: he randomly selects between the two three-outcome POVMs $\vec\Gamma_Z$ and $\vec\Gamma_X$, and we retain the standard assumption of basis-independent detection efficiency, $\Gamma_\bot^{(Z)} = \Gamma_\bot^{(X)} \eqqcolon \Gamma_\bot$.

\subsection{Protocol and equivalent scenarios}
\label{appsub:sp_protocol}

The \emph{Actual protocol} is similar to that in the main text, except that the intensity choice $\mu_k$ and its announcement are removed throughout. Crucially, the keep/trash announcement step is retained, to enable the quantum-coin argument. Applying the source-replacement technique of \cref{eq:sp_global_state} together with the same decomposition of Bob's measurement into a basis-independent filter $\{F,\mathbb{I}-F\}$ followed by a two-outcome POVM $\{G_0^{(\beta_k)},G_1^{(\beta_k)}\}$ on detected rounds, we obtain the following equivalent protocol. Relative to the {\emph{Actual protocol (source-replaced)}} of the main text, the intensity-register measurement and announcement are removed, and the photon-number-register measurement disappears.

\vspace{3pt}
\begin{mdframed}
{\noindent\textbf{Actual protocol (source replaced) --- single photon}
\begin{enumerate}
    \item \textit{State preparation:} Alice prepares $\ket{\Psi_N}_{C_1^N A_1^N T_1^N}$ in \cref{eq:sp_global_state} and sends the photonic systems $T_1^N$ {to Bob} through the quantum channel.

    \item \textit{Eve's attack:} Eve performs the most general attack on $T_1^N$, re-sends systems $B_1^N$ to Bob, and keeps an ancilla $E$.

    \item \textit{Detection announcement:} For each round $k$, Bob applies the basis-independent filter $\{F,\mathbb{I}-F\}$ to $B_k$ and publicly announces the set of detected rounds $\mathcal{D}\subseteq\{1,\ldots,N\}$.

    \item \textit{Keep/trash and sifting decisions:} For each round $k$, Alice decides ``keep'' with probability $p_{\mathrm{keep}}$ or ``trash'' otherwise, and Bob independently decides ``sifted'' or ``unsifted'' with equal probability. For {the} detected rounds, Alice announces her keep/trash decision; for rounds in $\mathcal{D}_\mathrm{trash}$, Bob announces his sifted/unsifted decision.

    \item \textit{Coin measurements and basis announcement:}
    \begin{enumerate}
        \item For each round $k$ where Alice chose ``trash'' and Bob chose ``sifted'', Alice measures her coin register $C_k$ in the $\{\ket{+},\ket{-}\}$ basis.
        \item For each detected keep round $k\in\mathcal{D}_\mathrm{keep}$, Alice measures $C_k$ in the $\{\ket{Z},\ket{X}\}$ basis to obtain $\alpha_k$ and announces it. Bob announces his basis $\beta_k$: if he chose ``sifted'' then $\beta_k=\alpha_k$, otherwise $\beta_k\neq\alpha_k$.
    \end{enumerate}

    \item \textit{Round classification:} Alice and Bob determine $\mathcal{D}_{\mathrm{keep},\mathrm{sifted},Z}$ and $\mathcal{D}_{\mathrm{keep},\mathrm{sifted},X}$, the sets of detected keep sifted rounds with $\alpha_k=\beta_k=Z$ and $\alpha_k=\beta_k=X$, respectively.

    \item \textit{Test-round measurements:} For each round in $\mathcal{D}_{\mathrm{keep},\mathrm{sifted},X}$, Alice measures $A_k$ in the $\{\ket{0_X},\ket{1_X}\}$ basis, Bob performs $\{G_0^{(X)},G_1^{(X)}\}$, and they announce their bit outcomes $a_k$ and $b_k$.

    \item \textit{Sifted-key measurements:} For each round in $\mathcal{D}_{\mathrm{keep},\mathrm{sifted},Z}$, Alice measures $A_k$ in the $\{\ket{0_Z},\ket{1_Z}\}$ basis and Bob performs $\{G_0^{(Z)},G_1^{(Z)}\}$. They define their sifted keys from the outcomes $a_k$ and $b_k$.

    \item[9-11.] Same as the \textit{Actual Protocol} (variable-length decision, error correction and verification, privacy amplification).
\end{enumerate}}
\end{mdframed}
\vspace{5pt}

As in the main text, the Hadamard-basis measurement of the trash-sifted coin registers (step~5a) does not affect any announced information (its outcome is never revealed) and is introduced solely for the phase-error estimation argument. The associated \emph{Phase-error estimation protocol} replaces the sifted-key measurement by the conjugate ($X$-basis) measurement on the key rounds {in step 8}:

\vspace{3pt}
\begin{mdframed}
{\noindent\textbf{Phase-error estimation protocol --- single photon}
\begin{enumerate}
    \item[1-7.] Same as \emph{Actual protocol (source replaced) --- single photon}.

    \item[8.] \textit{Phase-error measurements:} For each round in $\mathcal{D}_{\mathrm{keep},\mathrm{sifted},Z}$, Alice measures her qubit register $A_k$ in the $\{\ket{0_X},\ket{1_X}\}$ basis and Bob performs $\{G_0^{(X)},G_1^{(X)}\}$. We denote by $\bm{n_K}$ the number of rounds in $\mathcal{D}_{\mathrm{keep},\mathrm{sifted},Z}$ (i.e.\ the sifted-key length) and by $\bm{\eph}$ the fraction of those rounds in which $a_k\neq b_k$.
\end{enumerate}}
\end{mdframed}
\vspace{5pt}

Exactly as in the main text, the \textit{Actual protocol (source replaced) --- single photon} is statistically equivalent to the single-photon \textit{Actual protocol}, so proving security of the former establishes security of the latter.

\subsection{Variable-length security statement}
\label{appsub:sp_security_statement}

Now, since every detected round contains exactly one photon, the sifted-key length $\bm{n_K}=\abs{\mathcal{D}_{\mathrm{keep},\mathrm{sifted},Z}}$ is itself part of the announced data vector. There is therefore no need to lower-bound a single-photon subset of the key, and the statistical guarantee involves only the phase-error rate. The variable-length security statement simplifies as follows.

\begin{theorem}[Variable-length security via EUR --- single photon]
\label{thm:sp_variable_length_security}
Suppose that, for any initial global source-replacement state $\ket{\Psi_N}_{C_1^N A_1^N T_1^N}$, the \emph{Phase-error estimation protocol --- single photon} satisfies a bound of the form
\begin{equation}
\label{eq:sp_statistical_guarantee}
    \Pr[\,\bm{\eph} > \mathcal{E}_\mathrm{ph}(\bm{\vec{n}})\,]\leq \varepsilon_\mathrm{PE},
\end{equation}
where $\mathcal{E}_\mathrm{ph}(\vec n)$ is a function of the observed data vector $\vec n$. Let $\lambda_\mathrm{EC}(\vec n)$ be the number of bits revealed in error correction, and let
\begin{equation}
\label{eq:sp_final_key_length}
    l(\vec n) = \max\!\Big(0,\;
        n_K\,\big[1-h\big(\mathcal{E}_\mathrm{ph}(\vec n)\big)\big]
        - \lambda_\mathrm{EC}(\vec n)
        - 2\log_2 \tfrac{1}{2\varepsilon_\mathrm{PA}}
        - \log_2 \tfrac{2}{\varepsilon_\mathrm{EV}}\Big),
\end{equation}
where $n_K$ is the observed sifted-key length and $h(\cdot)$ is the binary entropy function (with $h(x)=1$ for $x>1/2$). Then, if Alice and Bob run the \emph{Actual protocol (source replaced) --- single photon} with this choice of $\lambda_\mathrm{EC}$ and $l$, the output key is $(2\sqrt{\varepsilon_\mathrm{PE}}+\varepsilon_\mathrm{PA}+\varepsilon_\mathrm{EV})$-secure.
\end{theorem}

\begin{proof} 
The argument is identical to that of \cref{thm:variable_length_security_decoy} (see \cref{appsub:proof_thm:variable_length_security_decoy}), with two simplifications. Since all detected rounds are single-photon, there is no photon-number outcome vector $\vec m$ and no need to introduce the set $S_{\vec n}$: the entropic uncertainty relation is applied directly to all $n_K$ key rounds, conditional on the observed data $\bm{\vec n}=\vec n$ with smoothing parameter $\kappa(\vec n)\coloneqq\Pr[\bm\eph>\mathcal{E}_\mathrm{ph}(\vec n)\mid \Omega(\vec n)]$, obtaining
\begin{equation}
    H^{\sqrt{\kappa(\vec n)}}_\mathrm{min}\big(Z_A^{n_K}\mid WE\big)_{\rho_{\vert\Omega(\vec n)}}
    \geq n_K\big[1-h\big(\mathcal{E}_\mathrm{ph}(\vec n)\big)\big], \qquad \forall \vec n .
\end{equation}
Since $n_K$ is directly observed (it is a component of $\vec n$), one uses it in place of the lower bound $\mathcal{N}_{K,(1)}(\vec n)$. The remaining steps---bounding the average smoothing parameter by $\varepsilon_\mathrm{PE}$ through \cref{eq:sp_statistical_guarantee}, and invoking the variable-length leftover-hashing argument of \cite[Theorem 4]{tupkaryPhaseError2025}---carry over verbatim.
\end{proof}

\subsection{Quantum coin inequality}
\label{appsub:sp_quantum_coin}

To obtain a bound of the form in \cref{eq:sp_statistical_guarantee}, we again use the quantum-coin argument. The single-photon counterpart of \cref{lem:quantum_coin_inequality_decoy} reads as follows.

\begin{lemma}[Single-photon quantum coin inequality]
\label{lem:sp_quantum_coin_inequality}
Consider the \emph{Phase-error estimation protocol --- single photon}. Let:
\begin{itemize}
    \item $\bm{n}_{\mathrm{keep},\mathrm{sifted},Z}^{\mathrm{err}}$ = number of (phase) errors in detected {keep-sifted} rounds with $\alpha_k=\beta_k=Z$ (= number of phase errors $\bm{\nph}$);
    \item $\bm{n}_{\mathrm{keep},\mathrm{sifted},Z}^{\mathrm{det}}$ = number of detected {keep-sifted} rounds with $\alpha_k=\beta_k=Z$ (= $\bm{n_K}$);
    \item $\bm{n}_{\mathrm{keep},\mathrm{sifted},X}^{\mathrm{err}}$ = number of errors in detected {keep-sifted} rounds with $\alpha_k=\beta_k=X$;
    \item $\bm{n}_{\mathrm{keep},\mathrm{sifted},X}^{\mathrm{det}}$ = number of detected {keep-sifted} rounds with $\alpha_k=\beta_k=X$;
    \item $\bm{n}_\mathrm{sifted}^{\mathrm{det}}$ = total number of detected sifted rounds (keep and trash);
    \item $\bm{n}_{\mathrm{trash},\mathrm{sifted},-}$ = number of trash sifted rounds where Alice obtained $\ket{-}_C$.
\end{itemize}
Then, for any $\epsilon_A>0$, with $\Delta_A=\sqrt{2\,\bm{n}_\mathrm{sifted}^{\mathrm{det}}\ln(1/\epsilon_A)}$,
\begin{equation}
\begin{aligned}
\label{eq:sp_finite_azuma}
    &\bm{n}_{\mathrm{keep},\mathrm{sifted},Z}^{\mathrm{err}}
    \underset{5\epsilon_A}{\leq} \big(\bm{n}_{\mathrm{keep},\mathrm{sifted},Z}^{\mathrm{det}} + \Delta_A\big)\, \times \\
&G_+\!\left(\frac{\bm{n}_{\mathrm{keep},\mathrm{sifted},X}^{\mathrm{err}} + \Delta_A}{\bm{n}_{\mathrm{keep},\mathrm{sifted},X}^{\mathrm{det}} - \Delta_A},\;
        1 - \frac{2 p_{\mathrm{keep}}\big(\bm{n}_{\mathrm{trash},\mathrm{sifted},-} + \Delta_A\big)}
                 {(1-p_{\mathrm{keep}})\big(\bm{n}_{\mathrm{keep},\mathrm{sifted},Z}^{\mathrm{det}} + \bm{n}_{\mathrm{keep},\mathrm{sifted},X}^{\mathrm{det}}\big)}
    \right) + \Delta_A
    \eqqcolon \bm{n}_\mathrm{ph}^{\mathrm{U}},
\end{aligned}
\end{equation}
where $G_+$ is the function defined in \cref{eq:g}. Consequently, the phase-error rate is bounded as
\begin{equation}
    \Pr\!\left[\,\bm{\eph} > \frac{\bm{n}_\mathrm{ph}^{\mathrm{U}}}{\bm{n}_{\mathrm{keep},\mathrm{sifted},Z}^{\mathrm{det}}}\,\right] \leq 5\epsilon_A .
\end{equation}
\end{lemma}

\begin{proof}
The proof is identical to that of \cref{lem:quantum_coin_inequality_decoy} (\cref{appsub:proof_quantum_coin_inequality_decoy}), with the photon-number postselection removed. One conditions on a realization of the set of detected sifted rounds $\mathcal{D}_\mathrm{sifted}$ (rather than its single-photon subset $\mathcal{D}_\mathrm{sifted}^{(1)}$), and applies the Bloch-sphere inequality \cref{eq:bloch_sphere_general} to the coin qubit $C_k$ of each such round, exactly as before. The reordering of measurements (first determining $\mathcal{D}_\mathrm{sifted}$, then applying the joint error POVM $\{\hat m_\mathrm{err},\mathbb{I}-\hat m_\mathrm{err}\}$, and only afterwards measuring the coins) and the subsequent {Azuma's step} are unchanged, with $\Delta_A$ now defined directly in terms of the total number of detected sifted rounds $\bm{n}_\mathrm{sifted}^{\mathrm{det}}$. 
\end{proof}

\subsection{Bounding the bad coin events}
\label{appsub:sp_bad_coin}

The term $\bm{n}_{\mathrm{trash},\mathrm{sifted},-}$ in \cref{eq:sp_finite_azuma} counts outcomes of the fictitious Hadamard-basis measurement on Alice's coin registers during trash rounds {with associated outcome "$-$"}. As in the main text, it depends only on Alice's local source-replacement state---being independent of Eve's attack and Bob's outcomes---and its distribution can therefore be analyzed in a simplified, adversary-free scenario:

\vspace{3pt}
\begin{mdframed}
{\textbf{Simple Coin Measuring Scenario --- single photon}
\begin{enumerate}
    \item \textit{State preparation:} Alice prepares her global source-replacement state $\ket{\Psi_N}_{C_1^N A_1^N T_1^N}$.
    \item \textit{Trash/sifted assignment and coin measurement:} For each round $k\in\{1,\ldots,N\}$, Alice assigns the round to ``trash'' with probability $(1-p_\mathrm{keep})$ and to ``sifted'' with probability $1/2$. If the round is assigned to both ``trash'' and ``sifted'', Alice measures the coin system $C_k$ in the $\{\ket{+},\ket{-}\}$ basis. We denote by $\bm{n}_{\mathrm{trash},\mathrm{sifted},-}$ the number of trash sifted rounds in which Alice obtained $\ket{-}$.
\end{enumerate}}
\end{mdframed}
\vspace{5pt}

Within this scenario, we define the indicator random variables
\begin{equation}
\label{eq:sp_X_k_def}
    \bm{X_k} \coloneqq \mathbf{1}\bigl\{\text{Alice chooses ``trash'' and ``sifted'' and obtains } \ket{-}_{C_k}\text{ on round } k\bigr\},
\end{equation}
so that $\bm{n}_{\mathrm{trash},\mathrm{sifted},-}=\sum_{k=1}^N \bm{X_k}$. The problem reduces to bounding this sum of binary random variables derived from measurements on the (partially known) state $\ket{\Psi_N}_{C_1^N A_1^N T_1^N}$. As in the {decoy-state} case, when the source has memory the $\bm{X_k}$ are not independent; assuming a finite correlation length $l_c$, we partition the rounds into $l_c+1$ groups of mutually independent variables and apply a Chernoff-type bound within each group.

\begin{lemma}[Bad-coin bound by groups --- single photon]
\label{lem:sp_bernstein_by_groups}
Suppose Alice runs the \emph{Phase-error estimation protocol --- single photon} with a source of maximum correlation length $l_c$. For each round $k$, let $a_{\bar k}\coloneqq(a_{k-l_c}^{k-1},a_{k+1}^{k+l_c})$ and $\alpha_{\bar k}\coloneqq(\alpha_{k-l_c}^{k-1},\alpha_{k+1}^{k+l_c})$ denote the bit and basis choices in the $l_c$-neighbourhood of round $k$ (excluding round $k$). Then, for any $\epsilon_C\in(0,1)$,
\begin{equation}
\label{eq:sp_bound_n_trash}
    \bm{n}_{\mathrm{trash},\mathrm{sifted},-}
    \;\underset{(l_c+1)\epsilon_C}{\leq}\;
    n_{\mathrm{trash},\mathrm{sifted},-}^{\mathrm{U}},
\end{equation}
where
\begin{equation}
\label{eq:sp_n_trash_U}
    n_{\mathrm{trash},\mathrm{sifted},-}^{\mathrm{U}}
    \coloneqq
    \frac{N(1-p_{\mathrm{keep}})}{2}\,\Delta_{\mathrm{coin}}^{\mathrm{U}}
    + \sqrt{N(l_c+1)(1-p_{\mathrm{keep}})\,\Delta_{\mathrm{coin}}^{\mathrm{U}}\,\ln\tfrac{1}{\epsilon_C}}
    + \tfrac{2(l_c+1)}{3}\,\ln\tfrac{1}{\epsilon_C}\,.
\end{equation}
Here $\Delta_{\mathrm{coin}}^{\mathrm{U}}$ quantifies the worst-case deviation of Alice's states from an ideal BB84 source,
\begin{equation}
\label{eq:sp_Delta_coin_U}
    \Delta_{\mathrm{coin}}^{\mathrm{U}}
    \coloneqq
    \max_{k}\;\max_{a_{\bar k},\,\alpha_{\bar k}}\;
    \frac{1}{2}\,\Re\!\Bigg(
        1 - \frac{1}{2\sqrt{2}}
        \sum_{a,a'\in\{0,1\}}(-1)^{aa'}\,
        \braket*{\Xi_{a,X\mid a_{\bar k},\alpha_{\bar k}}}{\Xi_{a',Z\mid a_{\bar k},\alpha_{\bar k}}}_{T_k\, T_{k+1}^{k+l_c}}
    \Bigg),
\end{equation}
where, for each round $k$, bit $a_k\in\{0,1\}$, and basis $\alpha_k\in\{Z,X\}$,
\begin{equation}
\label{eq:sp_Xi}
    \ket*{\Xi_{a_k,\alpha_k\mid a_{\bar k},\alpha_{\bar k}}}_{T_k\, T_{k+1}^{k+l_c}}
    \coloneqq
    \bigotimes_{j=k}^{k+l_c}
    \ket*{\psi_{a_{j-l_c}^{j},\,\alpha_{j-l_c}^{j}}}_{T_j}\,.
\end{equation}
\end{lemma}

\begin{proof}
The proof follows that of \cref{lem:bernstein_by_groups_decoy} (\cref{appsub:proof_bernstein_by_groups_decoy}) with the photon-number register removed throughout. 
Two changes occur relative to the {decoy-state} proof. First, every round emits a photon with certainty, so the factor $p_1=\sum_\mu p_\mu p_{1\vert\mu}$ that appears in the decoy expectation is replaced by $1$; hence $E[\bm{X_k}\mid a_{I_{\bar w}},\alpha_{I_{\bar w}}]\leq\tfrac12(1-p_\mathrm{keep})\Delta_\mathrm{coin}^\mathrm{U}$, and the per-group expectation is bounded by $\tfrac{N}{l_c+1}\tfrac12(1-p_\mathrm{keep})\Delta_\mathrm{coin}^\mathrm{U}$. Second, the $\Xi$-states carry no intensity/photon-number registers, so the conditional coin probability $p_{C_k=-}^{(k)}$ reduces to the single-photon overlap expression that defines $\Delta_\mathrm{coin}^\mathrm{U}$ in \cref{eq:sp_Delta_coin_U}. 
\end{proof}

\subsection{Phase-error bound and security}
\label{appsub:sp_phase_error}

In the single-photon case, the counts $\bm{n}_{\mathrm{keep},\mathrm{sifted},Z}^{\mathrm{det}}$, $\bm{n}_{\mathrm{keep},\mathrm{sifted},X}^{\mathrm{det}}$ and $\bm{n}_{\mathrm{keep},\mathrm{sifted},X}^{\mathrm{err}}$ entering \cref{eq:sp_finite_azuma}  are directly observable and contained in the announced data vector $\bm{\vec n}$. The decoy-state estimation step (\cref{lem:decoy_state_bounds}) is therefore unnecessary, and combining \cref{lem:sp_quantum_coin_inequality,lem:sp_bernstein_by_groups} immediately results in a bound on the phase-error rate of the full sifted key.

\begin{theorem}[Phase-error-rate bound for single-photon BB84 with encoding correlations]
\label{thm:sp_phase_error_bound}
Consider the \emph{Phase-error estimation protocol --- single photon} with encoding correlations up to maximum length $l_c$. Then, for any $\epsilon_A>0$ and $\epsilon_C\in(0,1)$,
\begin{equation}
    \Pr\!\left[\,\bm{\eph} > \mathcal{E}_\mathrm{ph}(\bm{\vec n};\epsilon_A,\epsilon_C)\,\right]
    \leq 5\epsilon_A + (l_c+1)\epsilon_C,
\end{equation}
where
\begin{equation}
\label{eq:sp_E_ph_def}
    \mathcal{E}_\mathrm{ph}(\bm{\vec n};\epsilon_A,\epsilon_C)
    \coloneqq
    \frac{
        \big(\bm{n}_{\mathrm{keep},\mathrm{sifted},Z}^{\mathrm{det}} + \Delta_A\big)\,
        G_+\!\left(
            \dfrac{\bm{n}_{\mathrm{keep},\mathrm{sifted},X}^{\mathrm{err}} + \Delta_A}
                  {\bm{n}_{\mathrm{keep},\mathrm{sifted},X}^{\mathrm{det}} - \Delta_A},\;
            1 - \dfrac{2 p_{\mathrm{keep}}\big(n_{\mathrm{trash},\mathrm{sifted},-}^{\mathrm{U}} + \Delta_A\big)}
                     {(1-p_{\mathrm{keep}})\big(\bm{n}_{\mathrm{keep},\mathrm{sifted},Z}^{\mathrm{det}} + \bm{n}_{\mathrm{keep},\mathrm{sifted},X}^{\mathrm{det}}\big)}
        \right) + \Delta_A
    }{
        \bm{n}_{\mathrm{keep},\mathrm{sifted},Z}^{\mathrm{det}}
    },
\end{equation}
with $\Delta_A=\sqrt{2\,\bm{n}_\mathrm{sifted}^{\mathrm{det}}\ln(1/\epsilon_A)}$ and $n_{\mathrm{trash},\mathrm{sifted},-}^{\mathrm{U}}$ given by \cref{eq:sp_n_trash_U}.
\end{theorem}

\begin{proof}
Similar to  \cref{thm:phase_error_bound_quantum_coin_decoy} (\cref{appsub:proof:thm:phase_error_bound_quantum_coin_decoy}). 
\end{proof}

Since the bound in \cref{thm:sp_phase_error_bound} is exactly of the form required by \cref{eq:sp_statistical_guarantee}, combining it with \cref{thm:sp_variable_length_security} immediately gives a security statement.

\begin{corollary}[Security of single-photon BB84 with bounded encoding correlations]
\label{cor:sp_security_bounded}
Consider the single-photon {Actual protocol} with encoding correlations up to maximum length $l_c$. If Alice and Bob set the final key length as in \cref{eq:sp_final_key_length} with $\mathcal{E}_\mathrm{ph}$ given by \cref{eq:sp_E_ph_def}, the protocol produces an $\varepsilon_\mathrm{sec}$-secure key with
\begin{equation}
    \varepsilon_\mathrm{sec} = 2\sqrt{5\epsilon_A + (l_c+1)\epsilon_C}\;+\;\varepsilon_\mathrm{PA}+\varepsilon_\mathrm{EV}.
\end{equation}
\end{corollary}

\begin{proof}
Immediate from \cref{thm:sp_phase_error_bound,thm:sp_variable_length_security} with $\varepsilon_\mathrm{PE}=5\epsilon_A+(l_c+1)\epsilon_C$.
\end{proof}

\subsection{{Application to a partially characterized LTI correlations model}}
\label{appsub:sp_LTI}

{We now particularize to the single-photon analogue of the LTI model of \cref{subsec:particular_correlations_model}. Namely, we assume that the state in \cref{eq:sp_state} takes the form}
\begin{equation}
\label{eq:sp_correlations_model_state}
    \ket*{\psi_{a_1^k,\alpha_1^k}}_{T_k}
    = \cos\!\big(\theta_{a_1^k,\alpha_1^k}\big)\, a_0^\dagger \ket{\rm vac}_{T_k}
    + \sin\!\big(\theta_{a_1^k,\alpha_1^k}\big)\, a_1^\dagger \ket{\rm vac}_{T_k},
\end{equation}
i.e.,\ a qubit encoded in the $XZ$ plane of the Bloch sphere, with the same encoding phase as in \cref{eq:correlations_model_phase,eq:correlations_model_delta,eq:correlations_model_Delta_exp}:
\begin{equation}
\label{eq:sp_correlations_model_phase}
    \theta_{a_1^k,\alpha_1^k} = \hat\theta_{a_k,\alpha_k} + \sum_{l=1}^{k-1}\delta_{a_{k-l},\alpha_{k-l}}^{(l)},
    \qquad
    \abs{\delta_{a_{k-l},\alpha_{k-l}}^{(l)}-\delta_{a'_{k-l},\alpha'_{k-l}}^{(l)}} \leq \Delta_l = \Delta_1 e^{-\lambda(l-1)}.
\end{equation}
{Since the residual contributions $\Delta_l$ decay exponentially but never strictly vanish, the source has correlations of unbounded length.} As in the main text, we (i) define a fictitious source with correlations truncated at a finite length $l_c^\mathrm{eff}$, (ii) prove security for it using the tools above, and (iii) lift the guarantee to the actual source via a trace-distance argument. Step (iii) relies on the model-independent lifting lemma of the main text, {\cref{lem:unbounded_correlations}}.

\begin{lemma}[Unbounded correlations~\cite{curras-lorenzoRigorousPhaseerrorestimation2026}]
\label{lem:sp_unbounded_correlations}
Let $\ket*{\Psi_N^{(\infty)}}_{C_1^N A_1^N T_1^N}$ be the source-replacement state for a source with correlations of unbounded length, and let $\ket*{\Psi_N^{(l_c^\mathrm{eff})}}_{C_1^N A_1^N T_1^N}$ be the corresponding state for a fictitious source truncated at length $l_c^\mathrm{eff}$. If
$T\big(\ketbra*{\Psi_N^{(\infty)}},\ketbra*{\Psi_N^{(l_c^\mathrm{eff})}}\big)\leq d$
and, for the fictitious source, $\Pr_{(l_c^\mathrm{eff})}[\bm\eph>\mathcal{E}_\mathrm{ph}(\bm{\vec n})]\leq\varepsilon_\mathrm{PE}$ for any attack, then, for the actual source and any attack, $\Pr_{(\infty)}[\bm\eph>\mathcal{E}_\mathrm{ph}(\bm{\vec n})]\leq\varepsilon_\mathrm{PE}+d$.
\end{lemma}

\begin{proof}
Identical to \cref{lem:unbounded_correlations}; see \cite[Lemma~1]{curras-lorenzoRigorousPhaseerrorestimation2026}. The statement is unaffected by the absence of the $I_k$ {and} $M_k$ registers.
\end{proof}

\paragraph*{Truncated model.}
The truncated single-photon source is identical to \cref{eq:sp_correlations_model_state,eq:sp_correlations_model_phase}, except that the bit and basis choices beyond length $l_c^\mathrm{eff}$ are replaced by fixed reference choices $a^\ast,\alpha^\ast$; the encoding phase for round $k$ {is the same as \cref{eq:correlations_model_phase_truncated}}
with the same residual bounds $\abs{\delta_{a_{k-l},\alpha_{k-l}}^{(l)}-\delta_{a'_{k-l},\alpha'_{k-l}}^{(l)}}\leq\Delta_l=\Delta_1 e^{-\lambda(l-1)}$. By construction, the state of round $k$ depends on at most $l_c^\mathrm{eff}$ preceding rounds, so \cref{cor:sp_security_bounded} applies with $l_c=l_c^\mathrm{eff}$. To use it we need the coin parameter $\Delta_\mathrm{coin}^\mathrm{U}$.

\begin{lemma}[Coin parameter for the truncated single-photon LTI model]
\label{lem:sp_bound_Delta_coin}
For the truncated single-photon model of {\cref{eq:sp_correlations_model_state,eq:correlations_model_phase_truncated}}, the coin parameter of \cref{eq:sp_Delta_coin_U} satisfies
\begin{equation}
\label{eq:sp_Delta_coin_U_bound}
    \Delta_\mathrm{coin}^\mathrm{U}
    \leq \frac{1}{2}\left[1 - \prod_{l=1}^{l_c^\mathrm{eff}}\cos\!\big(\Delta_l\big)\right].
\end{equation}
\end{lemma}

\begin{proof}
Similar to \cref{lem:bound_Delta_coin_decoy} (\cref{appsub:proof:lem:bound_Delta_coin_decoy}). The main change relative to the decoy proof is the overlap of the later-round states: here each later round $j$ contributes a single-photon overlap $\cos\!\big(\delta_{a_k,\alpha_k}^{(j-k)}-\delta_{a'_k,\alpha'_k}^{(j-k)}\big)\geq\cos(\Delta_{j-k})$, with no intensity/photon-number purification and thus no Poisson average. 
\end{proof}

\paragraph*{Truncation error.}
The cost of truncating the correlations is again controlled by a trace-distance bound, now without the average-intensity factor.

\begin{lemma}[Trace distance for exponentially decaying single-photon correlations]
\label{lem:sp_trace_distance}
Let {the states} $\ket*{\Psi_N^{(\infty)}}_{C_1^N A_1^N T_1^N}$ and $\ket*{\Psi_N^{(l_c^\mathrm{eff})}}_{C_1^N A_1^N T_1^N}$ denote, respectively, the source-replacement states of the unbounded and truncated single-photon LTI models. Then
\begin{equation}
\label{eq:sp_trace_distance}
    T\Big(\ketbra*{\Psi_N^{(\infty)}},\ketbra*{\Psi_N^{(l_c^\mathrm{eff})}}\Big)
    \leq \sqrt{N}\sum_{l=l_c^\mathrm{eff}+1}^{\infty}\Delta_l
    = \frac{\sqrt{N}\,\Delta_1\, e^{-\lambda l_c^\mathrm{eff}}}{1-e^{-\lambda}}.
\end{equation}
To ensure $T\leq d$ for a tolerance $d>0$, it suffices to choose
\begin{equation}
\label{eq:sp_l_c_formula}
    l_c^\mathrm{eff} = \left\lceil \frac{1}{\lambda}\ln\!\left(\frac{\sqrt{N}\,\Delta_1}{d(1-e^{-\lambda})}\right)\right\rceil.
\end{equation}
\end{lemma}

\begin{proof}
Similar to \cref{lem:trace_distance_exponential_decay_decoy} (\cref{appsub:proof:lem:trace_distance_exponential_decay_decoy}). The main difference is that here each per-round overlap is a \emph{single-photon} overlap, $\braket*{\psi_{a_{k-l_c^\mathrm{eff}}^k,\alpha_{k-l_c^\mathrm{eff}}^k}}{\psi_{a_1^k,\alpha_1^k}}=\cos(\Delta\theta_k)$, where the phase difference $\Delta\theta_k=\sum_{l=l_c^\mathrm{eff}+1}^{k-1}\big(\delta_{a_{k-l},\alpha_{k-l}}^{(l)}-\delta_{a^\ast,\alpha^\ast}^{(l)}\big)$ obeys $\abs{\Delta\theta_k}\leq\Delta_\mathrm{tail}(l_c^\mathrm{eff})\coloneqq\sum_{l=l_c^\mathrm{eff}+1}^\infty\Delta_l$. There is no Poisson resummation step (the state is deterministically single photon), so the average intensity $\bar\mu$ of the decoy bound is replaced by $1$. 
\end{proof}

As in the decoy case, $l_c^\mathrm{eff}$ grows only logarithmically in the block size $N$, so the truncation introduces a very mild overhead.

\paragraph*{Combined security statement.}
Collecting the above results gives a complete security guarantee for the single-photon LTI model.

\begin{corollary}[Security of single-photon BB84 with LTI encoding correlations]
\label{cor:sp_security_LTI}
Consider the single-photon {Actual protocol} with the LTI correlations model of \cref{eq:sp_correlations_model_state,eq:sp_correlations_model_phase}, and a truncation length $l_c^\mathrm{eff}$ satisfying \cref{eq:sp_l_c_formula} for some tolerance $d>0$. Let $\mathcal{E}_\mathrm{ph}$ be given by \cref{eq:sp_E_ph_def} with $\Delta_\mathrm{coin}^\mathrm{U}$ evaluated via \cref{eq:sp_Delta_coin_U_bound} (with $l_c=l_c^\mathrm{eff}$). Then
\begin{equation}
    \Pr\!\left[\,\bm\eph > \mathcal{E}_\mathrm{ph}(\bm{\vec n};\epsilon_A,\epsilon_C)\,\right]
    \leq 5\epsilon_A + (l_c^\mathrm{eff}+1)\epsilon_C + d,
\end{equation}
and, if Alice and Bob set the final key length as in \cref{eq:sp_final_key_length}, the protocol produces an $\varepsilon_\mathrm{sec}$-secure key with
\begin{equation}
    \varepsilon_\mathrm{sec} = 2\sqrt{5\epsilon_A + (l_c^\mathrm{eff}+1)\epsilon_C + d}\;+\;\varepsilon_\mathrm{PA}+\varepsilon_\mathrm{EV}.
\end{equation}
\end{corollary}

\begin{proof}
By \cref{cor:sp_security_bounded}, the phase-error bound holds for the truncated source with failure probability $5\epsilon_A+(l_c^\mathrm{eff}+1)\epsilon_C$. By \cref{lem:sp_trace_distance}, the trace distance between the actual and truncated source-replacement states is at most $d$. Applying \cref{lem:sp_unbounded_correlations} lifts the guarantee to the actual source, adding $d$ to the failure probability, and the security parameter follows from \cref{thm:sp_variable_length_security}.
\end{proof}

\subsection{Numerical results}

{As in the decoy-state case, the simulations are based on a standard channel and detection model for BB84 with an active-basis-choice receiver with two threshold single-photon detectors. The expected detection and error statistics are taken from Ref.~\cite{panayiMemoryassistedMeasurementdeviceindependent2014} and are given in \cref{app:channel_model} for completeness.} For the numerical simulations, {we consider the same parameter regime as in the main text. In particular,} we target an overall security parameter of $\varepsilon_\mathrm{sec} = 2 \cdot 10^{-10}$, and fix $\varepsilon_\mathrm{EV} = 10^{-10}$ and $\varepsilon_\mathrm{PA} = \varepsilon_\mathrm{sec}/6$. {The remaining failure parameters are allocated following the same procedure as in \cref{sec:numerics}, the only difference being the count of individual failure events: by \cref{thm:sp_phase_error_bound}, we now have $\varepsilon_\mathrm{PE} = 5\epsilon_A + (l_c+1)\epsilon_C$, i.e.,~$6+l_c$ events in total, since no decoy-state estimation is needed.} For the scenario with finite correlation length, we set
\begin{equation}
\epsilon_A = \epsilon_C = \frac{10^{-20}}{9\,(6 + l_c)},
\end{equation}
{which satisfies $\varepsilon_\mathrm{PE} = (6+l_c)\,\epsilon_A = 10^{-20}/9$ and thus achieves the required $\varepsilon_\mathrm{sec}$.} For the scenario with unbounded correlation length, {we again proceed in two steps: we first assign to $d$ one seventh of the overall failure probability budget, $d = \tfrac{1}{7}\cdot\tfrac{10^{-20}}{9} = 10^{-20}/63$, and compute the required effective truncation length $l_c^\mathrm{eff}$ via \cref{eq:sp_l_c_formula}. Then, we distribute the remaining budget of $\tfrac{6}{63}\cdot 10^{-20}$ evenly among the $6 + l_c^\mathrm{eff}$ remaining failure events, setting}
\begin{equation}
\epsilon_A = \epsilon_C = \frac{6 \cdot 10^{-20}}{63\,(6 + l_c^\mathrm{eff})},
\end{equation}
which achieves the required $\varepsilon_\mathrm{sec}$.

{The results are shown in \cref{fig:results_SP}. Qualitatively, the behavior mirrors that of the decoy-state case in \cref{fig:results}: the secret-key rate degrades as $\Delta_1$ increases, and the curves corresponding to different values of $l_c$ are again nearly indistinguishable, confirming that handling correlations of unbounded length {shows negligible cost}. Quantitatively, however, the penalty is more severe here, as can be seen in the larger gap between the two values of $\Delta_1$ considered.

This already suggests that single-photon implementations are more sensitive to encoding correlations than decoy-state ones, and in \cref{fig:results_comparison} we confirm this by comparing the two directly for a large block size ($N = 10^{16}$) and $\Delta_1 \in \{0.001,0.01,0.1\}$. {The large block size is chosen to suppress statistical fluctuations, which affect the decoy-state protocol more strongly due to the additional concentration inequalities required for the decoy-state estimation, so that the comparison better reflects the intrinsic impact of the correlations.} For weak correlations, the single-photon source performs better, as expected, but for sufficiently large $\Delta_1$ the decoy-state protocol eventually surpasses the secret-key rates achievable by the single-photon case. Although this appears counterintuitive, it follows from the difference in the photon-number statistics of the two sources. Weak coherent pulses with a low mean photon number are predominantly vacuum, and a vacuum emission is identical regardless of Alice's previous setting choices, and therefore does not leak information about them; only the small non-vacuum fraction does. The single-photon source, in contrast, emits exactly one photon per pulse with no vacuum component, so the state of every emitted photon depends on the previous bit-and-basis settings and can therefore leak information about them. Decoy-state implementations are thus inherently more robust to correlated  bit-and-basis encoders.}

\begin{figure}
\centering
\includegraphics[width=10cm]{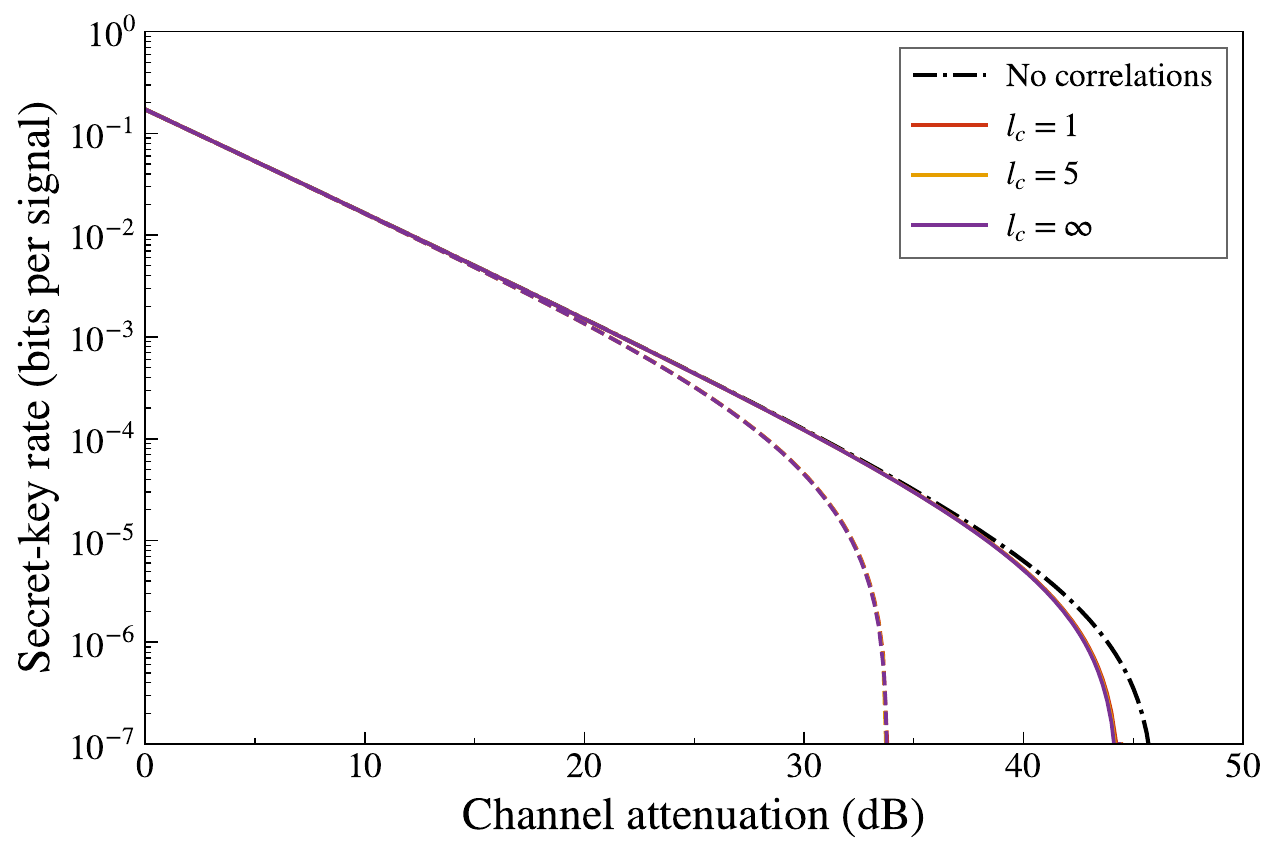}
    \caption{{Secret-key rate as a function of the channel attenuation for a single-photon BB84 protocol with a source suffering from LTI encoding correlations following the exponential-decay model of {\cref{appsub:sp_LTI}}. Solid lines correspond to $\Delta_1 = 0.001$ and dashed lines to $\Delta_1 = 0.01$, {and we consider $N=10^{12}$}. The curves labeled $l_c = 1$ and $l_c = 5$ assume correlations artificially truncated at length $l_c$ (i.e.,~$\Delta_l = 0$ for $l > l_c$), while $l_c = \infty$ corresponds to correlations of unbounded length. Note that both the dashed and the solid lines for different values of $l_c$ are nearly indistinguishable. The dash-dotted line shows the ideal uncorrelated case for reference. }}
    \label{fig:results_SP}
\end{figure}

\begin{figure}
\includegraphics[width=10cm]{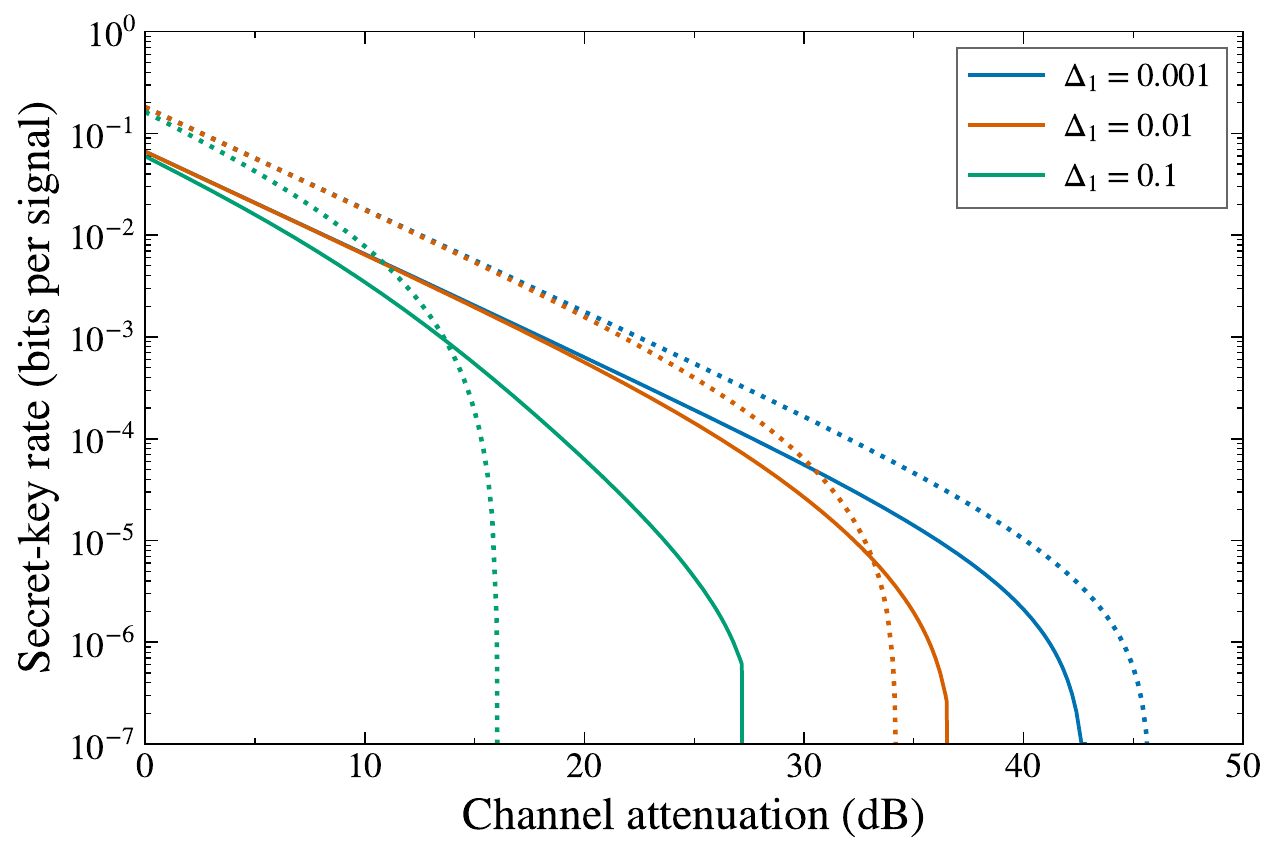}
    \caption{{Secret-key rate as a function of the channel attenuation for a decoy-state (solid lines) and a single-photon (dotted lines) BB84 protocol under the LTI correlations model. For the decoy-state case, we employ three intensity settings, fix the weakest to vacuum, and numerically optimize over the remaining two intensities and their associated probabilities.  Also, we consider $N=10^{16}$, $l_c=\infty$ and $\Delta_1 \in \{0.001,0.01,0.1\}$. }}
    \label{fig:results_comparison}
\end{figure}
}
\section{Concentration inequalities}
\label{app:concentration_ineqs}

{\begin{lemma}[Azuma's inequality \cite{azuma1967weighted} for sums of indicator variables]
\label{lem:azuma}
Let $\bm{X_1}, \ldots, \bm{X_n}$ be a sequence of (not necessarily independent) random variables taking values in $\{0,1\}$, and let $\mathcal{F}_{k-1}$ denote the history up to (and including) trial $k-1$. Define
\begin{equation}
    \bm{S} \coloneqq \sum_{k=1}^{n} \bm{X_k},
    \qquad\quad
    \bm{\tilde{S}} \coloneqq \sum_{k=1}^{n} \Pr[\bm{X_k} = 1 \mid \mathcal{F}_{k-1}].
\end{equation}
Then, for any $\epsilon \in (0,1)$,
\begin{equation}
    \Pr\!\left[\bm{S} > \bm{\tilde{S}} + \sqrt{2 n \ln(1/\epsilon)}\,\right] \leq \epsilon
    \qquad \text{and} \qquad
    \Pr\!\left[\bm{S} < \bm{\tilde{S}} - \sqrt{2 n \ln(1/\epsilon)}\,\right] \leq \epsilon.
\end{equation}
This follows from applying the standard Azuma inequality \cite{azuma1967weighted} to the martingale $\bm{\zeta_j} = \sum_{k=1}^{j}(\bm{X_k} - \mathbb{E}[\bm{X_k}\mid\mathcal{F}_{k-1}])$, whose increments are bounded by $1$ in absolute value.
\end{lemma}}

\begin{lemma}[Bernstein's inequality \cite{bernstein1924modification,boucheron2013concentration}]
\label{lem:bernstein}
Let \( \bm{X_1}, \bm{X_2}, \dots, \bm{X_n} \) be independent Bernoulli random variables with expectations \( q_1, q_2, \dots, q_n \). Define their sum as
\begin{equation}
    \bm{S} \coloneqq \sum_{k=1}^{n} \bm{X_k},
\end{equation}
with expectation
\begin{equation}
    \mu_S \coloneqq \mathbb{E}[\bm{S}] = \sum_{k=1}^{n} q_k.
\end{equation}
Then, for any \(\epsilon_C \in (0, 1)\),
\begin{align}
     \Pr[\bm{S} > \mu_S + \Delta^{+}(\mu_S, \epsilon_C)] \leq \epsilon_C,
\end{align}
where
\begin{equation}
    \Delta^{+}(x, y) = \sqrt{2x \ln(1/y)} + \tfrac{2}{3}\ln(1/y).
\end{equation}
\end{lemma}

{\begin{lemma}[Confidence bounds for the average of independent Bernoulli parameters \cite{mannalathSharpFinite2025}]
\label{lem:relaxed_chernoff}
Let $\bm{\hat{p}}$ be the average of $n$ independent, not necessarily identically distributed, Bernoulli random variables, and let $p = \mathbb{E}[\bm{\hat{p}}]$. For any $\epsilon \in (0,1)$, let $\kappa_{n,\epsilon} = \frac{2}{9n}\ln\frac{1}{\epsilon}$ and
\begin{equation}
    \gamma^{\pm}_{n,\epsilon}(x) = \frac{1}{1+4\kappa_{n,\epsilon}}\left[3\kappa_{n,\epsilon} + (1-2\kappa_{n,\epsilon})\,x \pm 3\sqrt{\kappa_{n,\epsilon}\left(\kappa_{n,\epsilon}+x-x^2\right)}\right],
\end{equation}
and define
\begin{equation}
    \Gamma^{-}_{n,\epsilon}(x) =
    \begin{cases}
        \gamma^{-}_{n,\epsilon}(x) & \text{if } x \in \left[\frac{3\kappa_{n,\epsilon}}{1+\kappa_{n,\epsilon}},\, 1\right], \\[4pt]
        -\epsilon & \text{otherwise},
    \end{cases}
    \qquad\quad
    \Gamma^{+}_{n,\epsilon}(x) =
    \begin{cases}
        \gamma^{+}_{n,\epsilon}(x) & \text{if } x \in \left[0,\, \frac{1-2\kappa_{n,\epsilon}}{1+\kappa_{n,\epsilon}}\right], \\[4pt]
        1+\epsilon & \text{otherwise}.
    \end{cases}
\end{equation}
Then,
\begin{equation}
    \Pr[\Gamma^{-}_{n,\epsilon}(\bm{\hat{p}}) \geq p] \leq \epsilon
    \qquad \text{and} \qquad
    \Pr[\Gamma^{+}_{n,\epsilon}(\bm{\hat{p}}) \leq p] \leq \epsilon.
\end{equation}
\end{lemma}

\begin{proof}
    See \cite[Prop.~1]{mannalathSharpFinite2025}. 
\end{proof}}

{\section{Channel and detection model used in the simulations}
\label{app:channel_model}

In this appendix, we specify the channel and detection model used to generate the statistics entering the secret-key-rate evaluations of \cref{fig:results,fig:results_SP,fig:results_comparison}. We consider a lossy channel followed by an active-basis-choice receiver with two threshold single-photon detectors, and adopt the standard approach in which double-click events (caused, e.g., by dark counts) are assigned a uniformly random bit value. We use the formulas in Ref.~\cite{panayiMemoryassistedMeasurementdeviceindependent2014} to calculate the expected detection and error statistics. The overall system transmittance is
\begin{equation}
    \eta = \eta_d \, 10^{-\alpha_{\mathrm{att}} L / 10},
\end{equation}
where $\eta_d$ is the detection efficiency of Bob's detectors {(which we assume is the same for both of them)}, $\alpha_{\mathrm{att}}$ is the attenuation coefficient of the fiber (in dB/km), and $L$ is the channel length, such that the channel attenuation plotted in the figures is $\alpha_{\mathrm{att}} L$.

\paragraph*{Decoy-state implementation.}
For a phase-randomized {weak} coherent pulse of intensity $\mu$, the probability that Bob observes a detection in a given round, and the probability that he observes a detection with a bit error, are respectively given by
\begin{equation}
\label{eq:channel_qdet_qerr}
    q_{\mathrm{det}}^{(\mu)} = g_C(\mu) + g_E(\mu),
    \qquad\quad
    q_{\mathrm{err}}^{(\mu)} = p_{\mathrm{mis}}\, g_C(\mu) + (1-p_{\mathrm{mis}})\, g_E(\mu),
\end{equation}
where $p_{\mathrm{mis}}$ is the misalignment probability (set to $p_{\mathrm{mis}} = 0$ in our simulations), and~\cite{panayiMemoryassistedMeasurementdeviceindependent2014}
\begin{equation}
\label{eq:channel_gC_gE}
    g_C(\mu) = \Big(1-\frac{p_d}{2}\Big)\Big[1 - (1-p_d)\, e^{-\mu\eta}\Big],
    \qquad\quad
    g_E(\mu) = p_d \Big[\Big(1-\frac{p_d}{2}\Big) e^{-\mu\eta} + \frac{1-e^{-\mu\eta}}{2}\Big],
\end{equation}
with $p_d$ the dark-count probability of each detector. Here, $g_C(\mu)$ {($g_E(\mu)$)} is the probability that Bob registers a detection whose recorded bit coincides with (differs from) the one encoded by Alice, for a perfectly aligned setup. Both quantities include dark-count contributions: in particular, {they account} for the fact that a click caused solely by a dark count still results in the correct bit half of the time, so that $g_C(0) = g_E(0) = p_d(1-p_d/2)$, reflecting the uniformly random outcome associated with a vacuum emission. Since $g_E(\mu) \propto p_d$, in the absence of misalignment all errors originate from dark counts, as expected.

Since Alice's bit, basis, keep/trash and intensity choices are independent of the detection process in this model, the expected values of the observed quantities entering our bounds are, for $\gamma \in \{Z,X\}$,
\begin{equation}
\label{eq:channel_expected_counts}
    E\big[\bm{n}_{\mathrm{keep},\mathrm{sifted},\gamma,\mu}^{\mathrm{det}}\big]
    = \frac{N p_\mu p_{\mathrm{keep}}}{4}\, q_{\mathrm{det}}^{(\mu)},
    \qquad
    E\big[\bm{n}_{\mathrm{keep},\mathrm{sifted},X,\mu}^{\mathrm{err}}\big]
    = \frac{N p_\mu p_{\mathrm{keep}}}{4}\, q_{\mathrm{err}}^{(\mu)},
    \qquad
    E\big[\bm{n}_{\mathrm{sifted}}^{\mathrm{det}}\big]
    = \frac{N}{2} \sum_\mu p_\mu\, q_{\mathrm{det}}^{(\mu)},
\end{equation}
where the factor $1/4$ accounts for the probability that Alice chooses basis $\gamma$ and Bob's basis matches, and the factor $1/2$ in the last expression {accounts} for the probability of a basis match (this last quantity includes both keep and trash rounds). Since the model is basis symmetric, the expected $Z$-basis error counts, which determine the error-correction leakage, coincide with the $X$-basis ones, and we set $\lambda_{\mathrm{EC}} = f\, n_{\mathrm{keep},\mathrm{sifted},Z}^{\mathrm{det}}\, h\big(e_Z\big)$, where $e_Z$ is the resulting $Z$-basis bit-error rate and $f$ is the error-correction inefficiency. As is customary, when evaluating the finite-key bounds we substitute the random observed counts by their expected values in \cref{eq:channel_expected_counts}. For each value of the channel attenuation, we numerically optimize the secret-key rate over the two non-vacuum intensities, their probabilities, and $p_{\mathrm{keep}}$.

\paragraph*{Single-photon implementation.}
For the single-photon source of \cref{app:single_photon}, the detection and error probabilities are obtained from \cref{eq:channel_qdet_qerr,eq:channel_gC_gE} by simply replacing the vacuum-survival probability $e^{-\mu\eta}$ of a phase-randomized {weak} coherent pulse with that of a single-photon pulse, $1-\eta$. That is, the probability that Bob observes a detection in a given round, and the probability that he observes a detection with a bit error, are respectively given by
\begin{equation}
    Q_1 = g_C^{(1)} + g_E^{(1)} = 1 - (1-p_d)^2 (1-\eta),
    \qquad\quad
    q_{\mathrm{err}}^{(1)} = p_{\mathrm{mis}}\, g_C^{(1)} + (1-p_{\mathrm{mis}})\, g_E^{(1)},
\end{equation}
where
\begin{equation}
    g_C^{(1)} = \Big(1-\frac{p_d}{2}\Big)\Big[1 - (1-p_d)(1-\eta)\Big],
    \qquad\quad
    g_E^{(1)} = p_d \Big[\Big(1-\frac{p_d}{2}\Big)(1-\eta) + \frac{\eta}{2}\Big],
\end{equation}
so that, since $p_{\mathrm{mis}} = 0$ in our simulations, the bit-error rate $q_{\mathrm{err}}^{(1)}/Q_1 = g_E^{(1)}/Q_1$ is determined solely by dark counts. Since in this case all detected rounds are single-photon rounds, the expected counts follow from expressions analogous to \cref{eq:channel_expected_counts} with the sum over intensities removed, and no decoy-state estimation is needed. Here, the only parameter to optimize is $p_{\mathrm{keep}}$.}

\end{document}